\documentclass{lmcs} 
\pdfoutput=1

\usepackage{lastpage}
\lmcsdoi{17}{4}{5}
\lmcsheading{}{\pageref{LastPage}}{}{}%
{May~25,~2020}{Nov.~03,~2021}{}

\usepackage[utf8]{inputenc}

\keywords{cubical type theory, parametricity, computational type theory, modal type theory}

\usepackage{booktabs}
\usepackage{tikz-cd}
\usepackage{proof}
\usepackage{mathpartir}
\usepackage{xparse}
\usepackage{etoolbox}
\usepackage{bm}
\usepackage{bbm}
\usepackage{mathtools}
\usepackage{manfnt}
\usepackage{pifont}
\usepackage{stmaryrd}
\usepackage{hyperref}
\usepackage[capitalize,noabbrev]{cleveref}
\usepackage{amssymb}
\usepackage{forloop}
\usepackage{fontawesome}

\usepackage{fonttable}

\DeclareFontFamily{OMX}{MnSymbolE}{}
\DeclareFontShape{OMX}{MnSymbolE}{m}{n}{
    <-6>  MnSymbolE5
   <6-7>  MnSymbolE6
   <7-8>  MnSymbolE7
   <8-9>  MnSymbolE8
   <9-10> MnSymbolE9
  <10-12> MnSymbolE10
  <12->   MnSymbolE12}{}
\DeclareFontShape{OMX}{MnSymbolE}{b}{n}{
    <-6>  MnSymbolE-Bold5
   <6-7>  MnSymbolE-Bold6
   <7-8>  MnSymbolE-Bold7
   <8-9>  MnSymbolE-Bold8
   <9-10> MnSymbolE-Bold9
  <10-12> MnSymbolE-Bold10
  <12->   MnSymbolE-Bold12}{}

\DeclareSymbolFont{MnSymE}{OMX}{MnSymbolE}{m}{n}

\DeclareMathSymbol{\llangle}{\mathopen}{MnSymE}{116}
\DeclareMathSymbol{\rrangle}{\mathclose}{MnSymE}{121}
\DeclareMathSymbol{\llangle}{\mathopen}{MnSymE}{116}
\DeclareMathSymbol{\rrangle}{\mathclose}{MnSymE}{121}

\DeclareFontEncoding{LS1}{}{}
\DeclareFontSubstitution{LS1}{stix}{m}{n}
\DeclareSymbolFont{arrows1}{LS1}{stixsf}{m}{n}
\DeclareMathDelimiter{\Ddownarrow}{\mathrel}{arrows1}{"60}{arrows1}{"60} 

\usetikzlibrary{calc}

\usetikzlibrary{decorations.markings}
\tikzset{
  ->-/.style={
    decoration={
      markings,
      mark=at position #1 with {\arrow{>}}},
    postaction={decorate}
  }
}

\tikzset{
  -<-/.style={
    decoration={
      markings,
      mark=at position #1 with {\arrow{<}}},
    postaction={decorate}
  }
}

\mathchardef\mh="2D 
\renewcommand{\_}{\ensuremath{\rule{1ex}{.4pt}}}
\renewcommand{\phi}{\varphi}

\makeatletter
\def\rightharpoonupfill@{\arrowfill@\relbar\relbar\rightharpoonup}
\newcommand{\overrightharpoonup}{%
\mathpalette{\overarrow@\rightharpoonupfill@}}
\makeatother

\definecolor{Revolutionary}{RGB}{232,70,68}

\newcommand{\eqdef}{\coloneqq}

\newcommand{\sem}[1]{\llbracket{#1}\rrbracket}

\newcommand{\BB}{\ensuremath{\mathbb{B}}}
\newcommand{\BI}{\ensuremath{\mathbb{I}}}
\newcommand{\BN}{\ensuremath{\mathbb{N}}}
\newcommand{\BZ}{\ensuremath{\mathbb{Z}}}

\newcommand{\BFI}{\ensuremath{\mathbf{I}}}

\newcommand{\CU}{\ensuremath{\mathcal{U}}}

\newcommand{\GD}{\ensuremath{\Delta}}

\newcommand{\GG}{\ensuremath{\Gamma}}

\newcommand{\GO}{\ensuremath{\Omega}}

\newcommand{\GPS}{\ensuremath{\Psi}}

\newcommand{\GX}{\ensuremath{\Xi}}
\newcommand{\Ga}{\ensuremath{\alpha}}
\newcommand{\Gb}{\ensuremath{\beta}}
\newcommand{\Gf}{\ensuremath{\phi}}
\newcommand{\Gg}{\ensuremath{\gamma}}
\newcommand{\Gd}{\ensuremath{\delta}}
\newcommand{\Ge}{\ensuremath{\varepsilon}}

\newcommand{\Gm}{\ensuremath{\mu}}

\newcommand{\Gps}{\ensuremath{\psi}}

\newcommand{\Gt}{\ensuremath{\tau}}

\newcommand{\Gx}{\ensuremath{\xi}}

\newcommand{\define}[4]{\expandafter#1\csname#3#4\endcsname{#2{#4}}}

\forcsvlist{\define{\newcommand}{\bm}{bm}}%
{a,b,c,d,e,f,g,h,i,j,k,l,m,n,o,p,q,r,s,t,u,v,w,x,y,z}

\NewDocumentCommand{\bad}{s}{\IfBooleanF{#1}{\ \ }\text{\ding{56}}}

\newcommand{\dataheading}[1]{\mathsf{data}~{#1}~\mathsf{where}}
\NewDocumentCommand\niceconstr{m o m g}{\mid {#1}\IfValueT{#2}{({#2})} \in {#3} \IfValueT{#4}{~[{#4}]}}

\newcommand{\rulename}[1]{\textsc{#1}}

\newcommand{\usubst}[3]{\ensuremath{#1 [#2 / #3]}}
\newcommand{\usubstdim}[3]{\ensuremath{#1 [{#2}/{#3}]}}
\newcommand{\usubstdims}[2]{{#1}{#2}}
\newcommand{\usubstlist}[2]{{#1}{#2}}

\NewDocumentCommand\typi{s o m m}{%
  \ensuremath{({#2}{:}{#3}) \IfBooleanF{#1}{\to}{\,} #4}}
\newcommand{\tyarr}[2]{\ensuremath{#1 \to #2}}
\newcommand{\tmlam}[2][]{\ensuremath{\lambda{#1}.{#2}}}

\newcommand{\tysigma}[3][]{\ensuremath{({#1}:{#2}) \times #3}}
\newcommand{\typrod}[2]{{#1} \times {#2}}
\NewDocumentCommand\tmfst{g}{\mathsf{fst}\IfValueT{#1}{(#1)}}
\NewDocumentCommand\tmsnd{g}{\mathsf{snd}\IfValueT{#1}{(#1)}}
\NewDocumentCommand\tmpair{s m m}{%
  \ensuremath{\IfBooleanF{#1}{\langle} #2,#3\IfBooleanF{#1}{\rangle}}}

\NewDocumentCommand\typath{g g g}{%
  \ensuremath{\mathsf{Path}\IfValueT{#1}{_{#1}\IfValueT{#2}{(#2,#3)}}}}
\newcommand{\tmplam}[2][]{\ensuremath{\lambda^{\BI} #1. #2}}
\newcommand{\tmpapp}[2]{\ensuremath{#1 @ #2}}

\NewDocumentCommand\tyv{g g g g}{%
  \ensuremath{\mathsf{V}\IfValueT{#1}{_{#1}\IfValueT{#2}{(#2,#3,#4)}}}}
\newcommand{\tmvin}[3]{\ensuremath{\mathsf{vin}_{#1}(#2,#3)}}
\newcommand{\tmvproj}[3]{\ensuremath{\mathsf{vproj}_{#1}(#2,#3)}}
\newcommand{\tyglue}{\mathsf{Glue}}
\newcommand{\tyg}{\mathsf{G}}

\newcommand{\tyuniv}{\CU}
\newcommand{\tyunivptd}{\tyuniv_{\mathsf{pt}}}

\NewDocumentCommand\tybridge{g g g}{%
  \ensuremath{\mathsf{Bridge}\IfValueT{#1}{_{#1}\IfValueT{#2}{(#2,#3)}}}}
\newcommand{\tmblam}[2][]{\ensuremath{\lambda^{\BFI} \bm{#1}. #2}}
\newcommand{\tmbapp}[2]{\ensuremath{#1 \bm{@} {#2}}}

\NewDocumentCommand\tmextent{s g g g g g}{%
  \ensuremath{\mathsf{\IfBooleanTF{#1}{ex}{extent}}\IfValueT{#2}{_{#2}\IfValueT{#3}{(#3;\IfValueTF{#4}{#4,#5,#6}{\cdots})}}}}

\NewDocumentCommand\tygel{G{{}} g g g}{%
  \ensuremath{\mathsf{Gel}_{#1}\IfValueT{#2}{(#2,#3,#4)}}}
\NewDocumentCommand\tmgel{G{{}} g g g}{%
  \ensuremath{\mathsf{gel}_{#1}\IfValueT{#2}{(#2,#3,#4)}}}
\NewDocumentCommand{\tmungel}{g}{\ensuremath{\mathsf{ungel}\IfValueT{#1}{(#1)}}}

\NewDocumentCommand\tyid{g g g}{%
  \ensuremath{\mathsf{Id}\IfValueT{#1}{_{#1}\IfValueT{#2}{(#2,#3)}}}}

\newcommand{\tyz}{\BZ}
\newcommand{\tyzmod}{\BZ/2\BZ}
\NewDocumentCommand\tmzin{g}{\mathsf{in}\IfValueT{#1}{(#1)}}
\NewDocumentCommand\tmzmod{g g}{\mathsf{mod}\IfValueT{#1}{(#1,#2)}}
\NewDocumentCommand\tmzinc{g}{\mathsf{inc}\IfValueT{#1}{(#1)}}
\NewDocumentCommand\tmzmodelim{g g g g}{\mathsf{mod}\mh\mathsf{elim}\IfValueT{#1}{_{#1}(#2,#3,#4)}}

\newcommand{\tyi}{\mathsf{line}}
\NewDocumentCommand\tmiin{g}{\mathsf{in}\IfValueT{#1}{(#1)}}
\NewDocumentCommand\tmielim{g g g}{\mathsf{interval}\mh\mathsf{elim}\IfValueT{#1}{_{#1}(#2,#3)}}

\newcommand{\tmj}[3]{\mathsf{J}_{#1}(#2,#3)}

\newcommand{\pto}{\ensuremath{\mathrel{\to_*}}}

\newcommand{\tylinv}[3]{\mathsf{Linv}(#1,#2,#3)}
\newcommand{\tyrinv}[3]{\mathsf{Rinv}(#1,#2,#3)}
\NewDocumentCommand\tyisiso{g g g}{\mathsf{isIso}\IfValueT{#1}{(#1,#2,#3)}}
\newcommand{\tyiso}[2]{\mathsf{Iso}(#1,#2)}
\newcommand{\tmua}[3]{\mathsf{ua}(#1,#2,#3)}
\newcommand{\tmidiso}[1]{\mathsf{idiso}(#1)}

\newcommand{\tychurchbool}{\BB}

\newcommand{\tybool}{\ensuremath{\mathsf{bool}}}
\newcommand{\tmtrue}{\ensuremath{\mathsf{tt}}}
\newcommand{\tmfalse}{\ensuremath{\mathsf{ff}}}
\NewDocumentCommand\tmif{G{{}} g g g}{%
  \ensuremath{\mathsf{if}_{#1}
    \IfValueT{#2}{\IfNoValueTF{#3}{({#2})}{({#2};{#3},{#4})}}}}
\NewDocumentCommand\tmbooleta{g}{\mathsf{\tybool\mh\eta}\IfValueT{#1}{(#1)}}

\newcommand{\tynot}[1]{\lnot {#1}}

\newcommand{\tmboolnot}{\mathsf{not}}

\newcommand{\tmsmbasel}{\mathsf{\circledast^L}}
\newcommand{\tmsmbaser}{\mathsf{\circledast^R}}
\NewDocumentCommand{\tmsmpair}{g g}{\llangle \IfValueTF{#1}{#1}{-},\IfValueTF{#2}{#2}{-} \rrangle}
\NewDocumentCommand{\tmsmgluel}{g g}{\mathsf{spoke^L}\IfValueT{#1}{(#1,#2)}}
\NewDocumentCommand{\tmsmgluer}{g g}{\mathsf{spoke^R}\IfValueT{#1}{(#1,#2)}}

\NewDocumentCommand\tmbridgefunext{g}{\ensuremath{\mathsf{bridge{\mh}funext}\IfValueT{#1}{(#1)}}}
\NewDocumentCommand\tmbridgeisoext{g}{\ensuremath{\mathsf{bridge{\mh}isoext}\IfValueT{#1}{(#1)}}}

\NewDocumentCommand\tmloosen{g g}{\ensuremath{\mathsf{loosen}\IfValueT{#1}{_{#1}\IfValueT{#2}{(#2)}}}}
\NewDocumentCommand\tmtighten{o g}{\ensuremath{\mathsf{tighten}\IfValueT{#1}{_{#1}}\IfValueT{#2}{(#2)}}}
\NewDocumentCommand\tmloosentighten{o g}{\ensuremath{\mathsf{inv}\IfValueT{#1}{_{#1}}\IfValueT{#2}{(#2)}}}
\newcommand{\tyisbdisc}[1]{\ensuremath{\mathsf{isBDisc}(#1)}}
\newcommand{\tybdisc}{\ensuremath{\CU}_{\mathsf{BDisc}}}

\newcommand{\tyisprop}[1]{\mathsf{isProp}(#1)}

\newcommand{\tyleminfty}{\mathsf{LEM}_\infty}
\newcommand{\tylem}{\mathsf{LEM}_{-1}}
\newcommand{\tywlem}{\mathsf{WLEM}}

\newcommand{\tmconcinv}[5]{\mathsf{conc{\mh}inv}_{#1}^{#2,#3}(#4,#5)}
\newcommand{\tmorcnx}[2]{\mathsf{connect}_{#1}(#2)}

\NewDocumentCommand\tygr{s G{{}} g g g}{%
  \ensuremath{\mathsf{Gr}\IfBooleanT{#1}{^*}_{#2}\IfValueT{#3}{(#3,#4,#5)}}}
\newcommand{\tmsmashgr}[1]{\wedge\mh\mathsf{graph}_{\bm{#1}}}

\newcommand{\etc}[1]{\ensuremath{\overrightharpoonup{#1}}}
\newcommand{\tube}[2]{\ensuremath{#1\hookrightarrow #2}}
\newcommand{\arraytube}[2]{\ensuremath{#1&\hookrightarrow& #2}}
\newcommand{\sys}[2]{\etc{\tube{#1}{#2}}}

\NewDocumentCommand\NewCoercionOperator{m m O{\rightsquigarrow} O{M}}{%
  \NewDocumentCommand#1{s G{{}} g g g}{%
    \ensuremath{\mathsf{#2}_{##2}%
    \IfBooleanTF{##1}
      {^{r #3 s}(#4)}
      {\IfValueT{##3}{^{##3 #3 ##4}(##5)}}}}
}

\NewCoercionOperator{\coe}{coe}

\NewDocumentCommand\NewCompositionOperator{s m m O{\rightsquigarrow} O{M} O{y.N_i}}{%
  \IfBooleanTF{#1}
    {\NewDocumentCommand#2{s g g g g}{%
      \IfBooleanTF{##1}
        {\ensuremath{\mathsf{#3}^{r #4 s}(#5;\sys{##2}{#6})}}
        {\IfNoValueTF{##2}
          {\ensuremath{\mathsf{#3}}}
          {\ensuremath{\mathsf{#3}^{##2 #4 ##3}(##4\IfValueT{##5}{;##5})}}}}}
    {\NewDocumentCommand#2{s G{{}} g g g g}{%
      \IfBooleanTF{##1}
        {\ensuremath{\mathsf{#3}_{##2}^{r #4 s}(#5;\sys{##3}{#6})}}
        {\IfNoValueTF{##3}
          {\ensuremath{\mathsf{#3}_{##2}}}
          {\ensuremath{\mathsf{#3}_{##2}^{##3 #4 ##4}(##5\IfValueT{##6}{;##6})}}}}}
}
\NewDocumentCommand\NewBigCompositionOperator{s m m O{\rightsquigarrow} O{M} O{y.N_i}}{%
  \IfBooleanTF{#1}
    {\NewDocumentCommand#2{s g g g g}{%
      \IfBooleanTF{##1}
        {\ensuremath{\mathsf{#3}^{r #4 r'}\left(#5;\sys{##2}{#6}\right)}}
        {\IfNoValueTF{##2}
          {\ensuremath{\mathsf{#3}}}
          {\ensuremath{\mathsf{#3}^{##2 #4 ##3}\left(##4\IfValueT{##5}{;##5}\right)}}}}}
    {\NewDocumentCommand#2{s G{{}} g g g g}{%
      \IfBooleanTF{##1}
        {\ensuremath{\mathsf{#3}_{##2}^{r #4 r'}\left(#5;\sys{##3}{#6}\right)}}
        {\IfNoValueTF{##3}
          {\ensuremath{\mathsf{#3}_{##2}}}
          {\ensuremath{\mathsf{#3}_{##2}^{##3 #4 ##4}\left(##5\IfValueT{##6}{;##6}\right)}}}}}
}

\NewCompositionOperator{\hcomp}{hcom}
\NewBigCompositionOperator{\bighcomp}{hcom}
\NewCompositionOperator{\comp}{com}

\newcommand{\ctxnil}{\cdot}
\NewDocumentCommand\ctxsnoc{m o m}{\IfValueTF{#2}{\ifblank{#1}{}{{#1},} {#2}:{#3}}{{#1}.{#3}}}
\NewDocumentCommand\ctxpdim{m o}{\IfValueTF{#2}{\ifblank{#1}{}{{#1},} {#2}:\BI}{{#1}.\BI}}
\NewDocumentCommand\ctxbdim{m o}{\IfValueTF{#2}{\ifblank{#1}{}{{#1},} \bm{#2}:\BFI}{{#1}.\BFI}}
\NewDocumentCommand\ctxpeq{m m m}{\IfValueTF{#2}{\ifblank{#1}{}{{#1},} {#2} = {#3}}{{#1}.{#2} = {#3}}}
\NewDocumentCommand\ctxcst{m m}{\IfValueTF{#2}{\ifblank{#1}{}{{#1},} {#2}}{{#1}.{#2}}}
\NewDocumentCommand\ctxres{s m m}{\IfBooleanTF{#1}{{#2}.{\setminus}{#3}}{\ifblank{#2}{}{{#2}}{\setminus}{#3}}}
\newcommand{\ctxmod}[2]{{#1}\sx{.}{#2}}

\newcommand{\ofp}[1]{#1 : \BI}

\NewDocumentCommand\wfctx{s d!! d<> m}
  {{#4}\ \mathsf{ctx}\IfValueT{#2}{_{#2}}\IfValueT{#3}{\ BAD}}
\NewDocumentCommand\eqctx{s d!! d<> m m}
  {{#4} = {#5}\ \mathsf{ctx}\IfValueT{#3}{\ BAD}}
\NewDocumentCommand\wftele{s d!! d<> o m}
  {\IfValueT{#3}{{#3} \Vdash}\IfValueT{#4}{{#4} \IfBooleanTF{#1}{\vdash\IfValueT{#2}{_{#2}}}{\gg}} {#5} \ \mathsf{ctx}\IfValueT{#3}{\ BAD}}
\NewDocumentCommand\wfpdim{s d!! d<> o m}
  {\IfValueT{#3}{{#3} \Vdash}\IfValueT{#4}{{#4} \IfBooleanTF{#1}{\vdash\IfValueT{#2}{_{#2}}}{\gg}} {#5} \IfBooleanTF{#1}{:}{\in} \BI}
\NewDocumentCommand\eqpdim{s d!! d<> o m m}
  {\IfValueT{#3}{{#3} \Vdash}\IfValueT{#4}{{#4} \IfBooleanTF{#1}{\vdash\IfValueT{#2}{_{#2}}}{\gg}} {#5} = {#6} \IfBooleanTF{#1}{:}{\in} \BI}
\NewDocumentCommand\wfbdim{s d!! d<> o m}
  {\IfValueT{#3}{{#3} \Vdash}\IfValueT{#4}{{#4} \IfBooleanTF{#1}{\vdash\IfValueT{#2}{_{#2}}}{\gg}} {#5} \IfBooleanTF{#1}{:}{\in} \BFI}
\NewDocumentCommand\eqbdim{s d!! d<> o m m}
  {\IfValueT{#3}{{#3} \Vdash}\IfValueT{#4}{{#4} \IfBooleanTF{#1}{\vdash\IfValueT{#2}{_{#2}}}{\gg}} {#5} = {#6} \IfBooleanTF{#1}{:}{\in} \BFI}
\NewDocumentCommand\wfcst{s d!! d<> o m}
  {\IfValueT{#3}{{#3} \Vdash}\IfValueT{#4}{{#4} \IfBooleanTF{#1}{\vdash\IfValueT{#2}{_{#2}}}{\gg}} {#5}\ \mathsf{constraint}}
\NewDocumentCommand\eqcst{s d!! d<> o m m}
  {\IfValueT{#3}{{#3} \Vdash}\IfValueT{#4}{{#4} \IfBooleanTF{#1}{\vdash\IfValueT{#2}{_{#2}}}{\gg}} {#5} = {#6}\ \mathsf{constraint}}
\NewDocumentCommand\wftype{s d!! d<> o m}
  {\IfValueT{#3}{{#3} \Vdash}\IfValueT{#4}{{#4} \IfBooleanTF{#1}{\vdash\IfValueT{#2}{_{#2}}}{\gg}} {#5}\ \mathsf{type}}
\NewDocumentCommand\eqtype{s d!! d<> o m m}
  {\IfValueT{#3}{{#3} \Vdash}\IfValueT{#4}{{#4} \IfBooleanTF{#1}{\vdash\IfValueT{#2}{_{#2}}}{\gg}} {#5} = {#6}\ \mathsf{type}}
\NewDocumentCommand\wfpretype{s d!! d<> o m}
  {\IfValueT{#3}{{#3} \Vdash}\IfValueT{#4}{{#4} \IfBooleanTF{#1}{\vdash\IfValueT{#2}{_{#2}}}{\gg}} {#5}\ \mathsf{pretype}}
\NewDocumentCommand\eqpretype{s d!! d<> o m m}
  {\IfValueT{#3}{{#3} \Vdash}\IfValueT{#4}{{#4} \IfBooleanTF{#1}{\vdash\IfValueT{#2}{_{#2}}}{\gg}} {#5} = {#6}\ \mathsf{pretype}}
\NewDocumentCommand\wftm{s d!! d<> o m m}
  {\IfValueT{#3}{{#3} \Vdash}\IfValueT{#4}{{#4} \IfBooleanTF{#1}{\vdash\IfValueT{#2}{_{#2}}}{\gg}} {#5} \IfBooleanTF{#1}{:}{\in} {#6}}
\NewDocumentCommand\eqtm{s d!! d<> o m m m}
  {\IfValueT{#3}{{#3} \Vdash}\IfValueT{#4}{{#4} \IfBooleanTF{#1}{\vdash\IfValueT{#2}{_{#2}}}{\gg}} {#5} = {#6} \IfBooleanTF{#1}{:}{\in} {#7}}
\NewDocumentCommand\wflist{s d!! d<> o m m}
  {\IfValueT{#3}{{#3} \Vdash}\IfValueT{#4}{{#4} \IfBooleanTF{#1}{\vdash\IfValueT{#2}{_{#2}}}{\gg}} {#5} \IfBooleanTF{#1}{:}{\in} {#6}}
\NewDocumentCommand\eqlist{s d!! d<> o m m m}
  {\IfValueT{#3}{{#3} \Vdash}\IfValueT{#4}{{#4} \IfBooleanTF{#1}{\vdash\IfValueT{#2}{_{#2}}}{\gg}} {#5} = {#6} \IfBooleanTF{#1}{:}{\in} {#7}}
\NewDocumentCommand\wfsubst{s d!! d<> o m m}
  {\IfValueT{#3}{{#3} \Vdash}\IfValueT{#4}{{#4} \IfBooleanTF{#1}{\vdash\IfValueT{#2}{_{#2}}}{\gg}} {#5} \IfBooleanTF{#1}{:}{\in} {#6}}
\NewDocumentCommand\eqsubst{s d!! d<> o m m m}
  {\IfValueT{#3}{{#3} \Vdash}\IfValueT{#4}{{#4} \IfBooleanTF{#1}{\vdash\IfValueT{#2}{_{#2}}}{\gg}} {#5} = {#6} \IfBooleanTF{#1}{:}{\in} {#7}}
\NewDocumentCommand\wfj{s d!! d<> o m}
  {\IfValueT{#4}{{#4} \IfBooleanTF{#1}{\vdash\IfValueT{#2}{_{#2}}}{\gg}} {#5}}

\newcommand{\steps}{\ensuremath{\longmapsto}}
\newcommand{\msteps}{\ensuremath{\longmapsto^*}}
\newcommand{\evals}{\ensuremath{\Downarrow}}

\newcommand{\isval}[1]{#1\ \mathsf{val}}

\newcommand{\dctxsubst}[3]{{#1} \Vdash {#2} \in {#3}} 
\newcommand{\id}{\ensuremath{\mathrm{id}}}

\newcommand{\syseval}[1]{{#1}^{\Downarrow}}
\newcommand{\releval}[1]{{#1}^{\Downarrow}}

\newcommand{\sysbridge}[1]{\textsc{Bridge}(#1)}
\newcommand{\sysgel}[1]{\textsc{Gel}(#1)}

\newcommand{\isvalcoh}[1]{\mathsf{Coh}(#1)}

\newcommand{\candwftype}[4]{{#1} \Vdash {#2} \downarrow {#3} \in {#4}}
\newcommand{\candeqtype}[5]{{#1} \Vdash {#2} \sim {#3} \downarrow {#4} \in {#5}}

\newcommand{\candeqtm}[4]{{#1} \Vdash {#2} \sim {#3} \in {#4}}

\newcommand{\sx}[1]{#1}

\newcommand{\tysubst}[2]{{#1}\sx{[}{#2}\sx{]}} 

\newcommand{\tmvar}{\mathsf{q}}
\newcommand{\tmsubst}[2]{{#1}\sx{[}{#2}\sx{]}} 

\newcommand{\pdimvar}{\mathsf{q}_\BI}
\newcommand{\pdimsubst}[2]{{#1}\sx{[}{#2}\sx{]}} 

\newcommand{\bdimvar}{\mathsf{q}_\BFI}
\newcommand{\bdimsubst}[2]{{#1}\sx{[}{#2}\sx{]}} 

\NewDocumentCommand{\substnil}{s}{\IfBooleanTF{#1}{\sx{!}}{\cdot}}
\NewDocumentCommand{\substid}{o}{\sx{\mathsf{id}}\IfValueT{#1}{_{#1}}}
\newcommand{\substconc}[2]{{#1} \mathbin{\sx{\circ}} {#2}}
\NewDocumentCommand{\substsnoc}{s m m o}{\IfBooleanTF{#1}{{#2}\sx{.}{#3}}{\ifblank{#2}{}{{#2},}{#3}/{#4}}}
\NewDocumentCommand\substproj{o}{\sx{\mathsf{p}\IfValueT{#1}{^{#1}}}}
\NewDocumentCommand{\substbdim}{s m m o}{\IfBooleanTF{#1}{{#2}\sx{.}{#3}}{\ifblank{#2}{}{{#2},}{#3}/{#4}}}
\NewDocumentCommand{\substpdim}{s m m o}{\IfBooleanTF{#1}{{#2}\sx{.}{#3}}{\ifblank{#2}{}{{#2},}{#3}/{#4}}}
\newcommand{\substbtranspose}[1]{{#1}^{\sx{\dagger}}}

\newcommand{\substbface}[1]{\sx{{\bm{#1}}_\BFI}}
\newcommand{\substprojbdim}{\sx{\mathsf{p}_\BFI}}
\newcommand{\substbex}{\sx{\mathsf{ex}_\BFI}}

\newcommand{\substprojpdim}{\sx{\mathsf{p}_\BI}}
\newcommand{\substpb}{\mathsf{ex}_{\BI\BFI}}

\newcounter{loopcntr}
\newcommand{\rpt}[2][1]{%
  \forloop{loopcntr}{0}{\value{loopcntr}<#1}{#2}%
}
\newcommand{\funcsnoc}[2][1]{{#2}^{\rpt[#1]{\times}}}
\newcommand{\funcbdim}[2][1]{{#2}^{\rpt[#1]{\BFI}}}
\newcommand{\funcres}[2]{{#1}{\setminus}{#2}}

\RequirePackage{graphicx}

\NewDocumentCommand\YoSymScaled{m}{
  \raisebox{#1em * \real{-.02}}{%
    \includegraphics[quiet=true,draft=false,height=#1em * \real{.67},keepaspectratio]{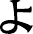}%
    \hspace{0.1em}%
  }
}

\NewDocumentCommand{\YoSym}{}{
  \mathord{%
    \mathchoice{\YoSymScaled{1}}{\YoSymScaled{1}}{\YoSymScaled{\defaultscriptratio}}{\YoSymScaled{\defaultscriptscriptratio}}
  }%
}

\newcommand{\cube}{\square}
\newcommand{\cart}{\cube_{\mathit{c}}}
\newcommand{\aff}{\cube_{\mathit{a}}}
\newcommand{\bicube}{\cube_{\mathit{ca}}}
\newcommand{\Set}{\mathbf{Set}}
\newcommand{\op}[1]{#1^{\mathrm{op}}}
\newcommand{\PSh}[1]{[\op{#1},\Set]}
\newcommand{\yon}[1]{\YoSym{#1}}

\newcommand{\subdec}{\GO_{\mathit{dec}}}

\newcommand{\cubres}{\mathit{Res}}
\newcommand{\cubext}{\mathit{Ext}}
\newcommand{\cubprod}[1]{{#1} \otimes \BFI}

\newcommand{\funcid}{\mathit{Id}}
\newcommand{\sembridge}[3]{\mathit{Bridge}(#1,#2,#3)}
\newcommand{\semblam}[1]{\mathit{lam}^\BFI(#1)}
\newcommand{\sembapp}[1]{\mathit{app}^\BFI(#1)}
\newcommand{\semgel}[4]{\mathit{Gel}_{#1}(#2,#3,#4)}

\newcommand{\sembvar}[1]{\mathit{var}(#1)}

\NewDocumentCommand\tydisc{g}{\ensuremath{\mathsf{Disc}\IfValueT{#1}{(#1)}}}
\NewDocumentCommand\tyglobal{g}{\ensuremath{\mathsf{Glo}\IfValueT{#1}{(#1)}}}
\NewDocumentCommand\tycodisc{g}{\ensuremath{\mathsf{Codisc}\IfValueT{#1}{(#1)}}}

\NewDocumentCommand\tmletmod{g g g g g g}{\mathsf{letmod}\IfValueT{#1}{^{#1}_{#2}({#3},{#4},{#5},#6)}}
\NewDocumentCommand\tmunmod{m m g}{\mathsf{unmod}^{#1}_{#2}\IfValueT{#3}{(#3)}}

\NewDocumentCommand\tmdisc{g}{\ensuremath{\mathsf{disc}\IfValueT{#1}{(#1)}}}
\NewDocumentCommand\tmglobal{g}{\ensuremath{\mathsf{glo}\IfValueT{#1}{(#1)}}}
\NewDocumentCommand\tmcodisc{g}{\ensuremath{\mathsf{codisc}\IfValueT{#1}{(#1)}}}
\NewDocumentCommand\tmletdisc{g g g g g}{\mathsf{letdisc}\IfValueT{#1}{^{#1}({#2},{#3},{#4},{#5})}}
\NewDocumentCommand\tmletglobal{g g g g g}{\mathsf{letglo}\IfValueT{#1}{^{#1}({#2},{#3},{#4},{#5})}}
\NewDocumentCommand\tmletcodisc{g g g g g}{\mathsf{letcodisc}\IfValueT{#1}{^{#1}({#2},{#3},{#4},{#5})}}
\NewDocumentCommand\tmunglobal{g}{\ensuremath{\mathsf{unglo}\IfValueT{#1}{(#1)}}}
\NewDocumentCommand\tmuncodisc{g}{\ensuremath{\mathsf{uncodisc}\IfValueT{#1}{(#1)}}}

\NewDocumentCommand\tmstrip{g}{\ensuremath{\mathsf{strip}\IfValueT{#1}{(#1)}}}

\newtheorem{rul}[thm]{Rule}

\def\eg{\emph{e.g.}}
\def\cf{\emph{cf.}}
\def\ie{\emph{i.e.}}
\def\etal{\emph{et al.}}

\begin{document}

\title[Internal Parametricity for Cubical Type Theory]{Internal Parametricity for Cubical Type Theory}
\titlecomment{{\lsuper*}This article is an extended version of~\cite{cavallo20b}.}

\author[E.~Cavallo]{Evan Cavallo}	
\author[R.~Harper]{Robert Harper}	
\address{Department of Computer Science, Carnegie Mellon University, Pittsburgh, Pennsylvania, USA}	
\email{\{ecavallo,rwh\}@cs.cmu.edu}  

\thanks{This material is based on research sponsored by Air Force Office of Scientific Research through MURI
  grants FA9550--15--1--0053 and FA9550--21--0009 (Tristan Nguyen, program manager). Any opinions, findings and conclusions or  recommendations expressed in this material are those of the authors and do not necessarily reflect the views of the AFOSR} 





\begin{abstract}
  We define a computational type theory combining the contentful equality structure of cartesian cubical type
  theory with internal parametricity primitives. The combined theory supports both univalence and its
  relational equivalent, which we call relativity. We demonstrate the use of the theory by analyzing
  polymorphic functions between higher inductive types, observe how cubical equality regularizes parametric
  type theory, and examine the similarities and discrepancies between cubical and parametric type theory,
  which are closely related. We also abstract a formal interface to the computational interpretation and show
  that this also has a presheaf model.
\end{abstract}

\maketitle

\section*{Introduction}\label{sec:introduction}

In the past decade or so, the study of dependent type theory has been transformed by a growing recognition of
the importance of \emph{contentful} (or \emph{proof-relevant}) \emph{equality}. At its root, the idea is
simple: \emph{a proof of an equality is a piece of data}. To go a bit a farther, \emph{a proof of equality may
  play a non-trivial role in computation}. From the type-theoretic perspective, where the computational
content of proofs has always been emphasized (``proofs as programs''), it is completely natural to think of
equality this way. Nevertheless, it has been common to treat proofs of equality as irrelevant: we prove
equalities to check code correctness or to prove a theorem, but we do not expect those proofs to influence how
our code runs.

That expectation was shaken by Hofmann and Streicher's \emph{groupoid model}~\cite{hofmann95} of
Martin-L\"{o}f's intensional type theory (ITT)~\cite{martin-lof75}. Intensional type theory includes the
\emph{identity type}: for every type $A$ and elements $M,N \in A$, there is a type $\tyid{A}{M}{N}$ whose
elements are proofs that $M$ and $N$ are ``equal''. (We henceforth call these elements \emph{identities} or
\emph{identifications}.)  Hofmann and Streicher's model is designed to falsify the principle of
\emph{uniqueness of identity proofs}, which states that all proofs of a given identity are themselves
identical. They thereby show that this principle is, oddly enough, independent of ITT\@. Far from being a
contrived counter-model, the groupoid model demonstrates that contentful equality arises quite naturally in
mathematics. Hofmann and Streicher highlight isomorphism as the premiere example: two isomorphic sets are
essentially ``the same'', but the same two sets can be isomorphic in many different ways. Awodey and Warren~\cite{warren08,awodey09} and van den Berg and Garner~\cite{van-den-berg12} generalized the groupoid model
construction to produce models where there are not only distinct proofs of identities, but distinct proofs of
identities between proofs of identities and so on. Voevodsky, who was separately developing a simplicial model
with similar properties~\cite{kapulkin12}, proposed to extend ITT with his \emph{univalence axiom}, which
asserts precisely that identifications between types correspond to isomorphisms.

Voevodsky's univalence axiom codifies a kind of reasoning that is already ubiquitous in informal mathematics,
that of treating isomorphic objects as interchangeable. In fact, the axiom has far-reaching consequences, as
subsequently explored in the fields of \emph{homotopy type theory}~\cite{hott-book} and \emph{univalent
  foundations}~\cite{voevodsky15,unimath}. As a simple but characteristic example, it implies \emph{function
  extensionality} as a corollary: functions are identical when they are identical on all arguments~\cite[\S4.9]{hott-book}. Analogous extensionality principles for equality in coinductive types (\eg,~\cite{ahrens15}) and quotients (\eg,~\cite{kraus19}) follow as well. In short, univalence regularizes the
behavior of equality throughout type theory.

Of course, there is one sense in which univalent ITT is spectacularly ill-behaved: by introducing an axiom, we
destroy the computational content of type theory. There is no way to run a program written in ITT that uses
the univalence axiom, because the ``proof'' of the axiom does not compute. This was finally addressed by the
development of \emph{cubical type theories}~\cite{cchm,angiuli18,orton18,abcfhl,cavallo20}, a family of
univalent type theories (with constructive models) where the univalence axiom follows from more fundamental
primitives that do compute. The central principle of cubical type theory is that equalities in a type
$A$---now called \emph{paths}---are represented by maps from an interval object $\BI$ into $A$.

Cubical type theory will be our starting point, our setting to explore contentful equality. In this work, we
develop \emph{internal parametricity} as an effective tool to reason about contentful equality,
which---despite its remarkable usefulness---presents new difficulties as well.

\subsection*{The challenges of contentful equality} As users of ITT have long known, a lack of uniqueness of
identity proofs has some frustrating consequences. To put it pithily, when equalities are not always equal, we
sometimes need to prove that they are. For example, we typically need to know that composition of equalities
(\ie, transitivity) is associative. When we have contentful equality in mind, these ``coherence'' proofs
\emph{are} mathematically significant, but their proofs are often tedious, uninteresting, and difficult to
conceptualize, especially as one gets to the point of proving equalities between equalities between
equalities.

The problem is most acute when we work with quotients. In cubical type theory, as in homotopy type theory and
the univalent foundations, inductive types and quotient types both arise as specializations of \emph{higher
  inductive types}~\cite{hott-book,coquand18,cavallo19b}. Where an inductive type is defined by constructors
that generate elements of the type, a higher inductive type is defined by a specification of element and
\emph{path} constructors. As a simple example, we can specify the type $\tyzmod$ of integers \emph{mod} 2 in
cubical type theory as the following higher inductive type.

\[
  \begin{array}{l}
    \dataheading{\tyzmod} \\
    \niceconstr{\tmzin}[\ctxsnoc{}[n]{\tyz}]{\tyzmod} \\
    \niceconstr{\tmzmod}[\ctxpdim{\ctxsnoc{}[n]{\tyz}}[x]]{\tyzmod}{\tube{x=0}{\tmzin{n}} \mid \tube{x=1}{\tmzin{n+2}}}
  \end{array}
\]

The first constructor of this type is standard: whenever we have an integer $n : \tyz$, we get
$\tmzin{n} \in \tyzmod$. The second is a path constructor: whenever we have $n : \tyz$, we get a path from
$\tmzin{n}$ to $\tmzin{n+2}$. That path is represented by a term $\tmzmod{n}{x}$ depending on an
\emph{interval variable} $x$, together with equations declaring that $\tmzmod{n}{0}$ is $\tmzin{n}$ and
$\tmzmod{n}{1}$ is $\tmzin{n+2}$. The interval is to be thought of roughly as the real interval from analysis:
as $x : \BI$ varies from $0$ to $1$, the constructor $\tmzmod{n}{x}$ draws a line from $\tmzin{n}$ to
$\tmzin{n+2}$. Pictorially, we have something like the following.
\[
  \begin{tikzcd}
    \cdots \ar[bend right]{rr}[below]{\tmzmod{-3}{x}} &
    \tmzin{-2} \ar[bend left]{rr}{\tmzmod{-2}{x}} &
    \tmzin{-1} \ar[bend right]{rr}[below]{\tmzmod{-1}{x}} &
    \tmzin{0} \ar[bend left]{rr}{\tmzmod{0}{x}} &
    \tmzin{+1} \ar[bend right]{rr}[below]{\tmzmod{+1}{x}} &
    \tmzin{+2} &
    \cdots
  \end{tikzcd}
\]

To construct a map from $\tyzmod$ to another type, we simply explain where to send $\tmzin{n}$ and
$\tmzmod{n}{x}$, just as in ordinary induction. For example, the increment map
$\tmzinc \in \tyzmod \to \tyzmod$ is defined by the clauses $\tmzinc{\tmzin{n}} \eqdef \tmzin{n+1}$ and
$\tmzinc{\tmzmod{n}{x}} \eqdef \tmzmod{n+1}{x}$. In order for the definition to be sensible, we need to check
that $\tmzinc{\tmzmod{n}{0}} = \tmzinc{\tmzin{n}}$ and $\tmzinc{\tmzmod{n}{1}} =
\tmzinc{\tmzin{n+2}}$. Similarly, we can define addition by an iterated induction of the following form.
\[
  \begin{array}{lclcl}
    \tmzin{m} &+& \tmzin{n} &\eqdef& \tmzin{m+n} \\
    \tmzmod{m}{x} &+& \tmzin{n} &\eqdef& \cdots \\
    \tmzin{m} &+& \tmzmod{n}{y} &\eqdef& \cdots \\
    \tmzmod{m}{x} &+& \tmzmod{n}{y} &\eqdef& \cdots \\
  \end{array}
\]

The final clause of this definition depends on two interval variables $x,y : \BI$. We can visualize it as a
square with a boundary determined by the other clauses.
\[
  \begin{tikzpicture}
    \draw (-2.2 , 2.4) [->] to node [above] {\small $x$} (-1.7 , 2.4) ;
    \draw (-2.2 , 2.4) [->] to node [left] {\small $y$} (-2.2 , 1.9) ;
    \node (tl) at (1.5 , 2) {$\bullet$} ;
    \node (tr) at (6 , 2) {$\bullet$} ;
    \node (bl) at (1.5 , 0.5) {$\bullet$} ;
    \node (br) at (6 , 0.5) {$\bullet$} ;
    \draw (tl) [->] to node [above] {$\tmzmod{m}{x} + \tmzin{n}$} (tr) ;
    \draw (tl) [->] to node [left] {$\tmzin{m} + \tmzmod{n}{y}$} (bl) ;
    \draw (tr) [->] to node [right] {$\tmzin{m+2} + \tmzmod{n}{y}$} (br) ;
    \draw (bl) [->] to node [below] {$\tmzmod{m}{x} + \tmzin{n+2}$} (br) ;
  \end{tikzpicture}
\]
Finding a term to fill this square is not so simple, particularly if the edge clauses are already defined in a
complicated way.

Iterated induction on higher inductive types is a frequent source of such coherence obligations. Particularly
notorious instances, which will serve as a test case in this paper, are proofs establishing the algebraic
structure of the \emph{smash product}~\cite[\S6.8]{hott-book}. The smash product $\wedge_*$ is a binary
operator on pointed types, pairs $A_* = \tmpair{A}{a_0}$ of types $A$ equipped with a chosen ``basepoint''
element $a_0 \in A$. We will define the product in \cref{sec:practice:smash}; for now, it suffices to know
that we define its underlying type as a higher inductive type. The smash product is a natural notion of tensor
product for the category of pointed types. In particular, suppose we write $A_* \to B_*$ for the type of
basepoint-preserving functions between pointed types $A_*$ and $B_*$, which we can make into a pointed type
$A_* \to_* B_*$ by taking the unique basepoint-preserving constant function as its basepoint. Then we have a
(pointed) isomorphism $A_* \pto (B_* \pto C_*) \simeq (A_* \wedge_* B_*) \pto C_*$. The smash product appears
as a basic tool in \emph{synthetic homotopy theory}, the study of higher-dimensional structure (\emph{homotopy
  theory}) through the lens of univalent type theory.

We would like to know that the smash product is commutative, associative, and so on. To construct a commutator
$A_* \wedge_* B_* \to B_* \wedge_* A_*$, we naturally go by induction on elements of $A_* \wedge_* B_*$; to
construct an associator $(A_* \wedge_* B_*) \wedge_* C_* \to A \wedge_* (B \wedge_* C)$, we need iterated
induction on the two instances of $\wedge_*$ in the domain. This is already quite non-trivial, but it gets
worse. If we want to prove that our associator is an isomorphism, then we need to prove \emph{equalities}
between elements of $(A_* \wedge_* B_*) \wedge_* C_*$ (and $A_* \wedge_* (B_* \wedge_* C_*)$) by
induction. This increases the dimension by another notch, forcing us to reason with 3-dimensional terms. Going
further, we can ask whether the associator satisfies the \emph{pentagon identity}, which equates the two
\emph{a priori} distinct ways of re-associating from $((A_* \wedge_* B_*) \wedge_* C_*) \wedge_* D_*$ to
$A_* \wedge_* (B_* \wedge_* (C_* \wedge_* D_*))$.
\[
  \begin{tikzcd}[column sep=-5em]
    && ((A_* \wedge_* B_*) \wedge_* C_*) \wedge_* D_* \ar{dll}[sloped]{\simeq} \ar{drr}[sloped]{\simeq} \\
    (A_* \wedge_* (B_* \wedge_* C_*)) \wedge_* D_* \ar{dr}[sloped,below]{\simeq} &&&& (A_* \wedge_* B_*) \wedge_* (B_* \wedge_* C_*) \ar{dl}[sloped,below]{\simeq} \\
    & A_* \wedge_* ((B_* \wedge_* C_*) \wedge_* D_*) \ar{rr}[sloped,below]{\simeq} && A_* \wedge_* (B_* \wedge_* (C_* \wedge_* D_*))
  \end{tikzcd}
\]
This is an equality between elements of a thrice-iterated smash product, so its proof requires constructing
4-dimensional terms. Going further, we might also want to check that these proofs are natural in the arguments
$A_*$, $B_*$, $C_*$, and $D_*$! There is, in fact, an infinite tower of coherence conditions that we expect
the smash product to satisfy, making it into an \emph{$\infty$-coherent symmetric monoidal product}.

Sadly, it quickly becomes first painful and then infeasible to construct these proofs by hand. In homotopy
type theory, Van Doorn verifies that the smash product is a 1-coherent symmetric monoidal product by first
proving the isomorphism $A_* \pto (B_* \pto C_*) \simeq (A_* \wedge_* B_*) \pto C_*$ and using this to obtain
the other results~\cite[\S4.3]{van-doorn18}. (1-coherence goes as far as the pentagon and its cousin the
hexagon identity, which relates the associator and unit laws.) As Van Doorn notes~\cite[Remark
4.3.29]{van-doorn18}, there is a gap in the argument: roughly, the proofs use that the above is a pointed
isomorphism natural in $A_*$, $B_*$, $C_*$, but only proves that it is natural as an unpointed
isomorphism. Once again, there is no doubt that the gap can be filled, but to do so involves a prohibitive
amount of path manipulation. Seeking to avoid all this, Brunerie suggests automating coherence proofs, using a
simple strategy of searching for opportunities to apply the elimination principle for the equality type~\cite{brunerie18}. Unfortunately, this approach also reaches its practical limit around the 1-coherence
mark. In either case, while it might be possible to reach the 2-coherences with enough effort and
optimization, there is little hope of handling general $n$-coherences.

\subsection*{Parametricity}

We propose a novel approach to these problems using a well-established tool from computer science: Reynolds'
\emph{parametricity}~\cite{reynolds83}. Parametricity is a versatile technique used to prove uniformity
properties of terms constructed in type theory; these are popularly known as ``theorems for free!''\ after
Wadler~\cite{wadler89}.  Reynolds' original results concerned the simply typed $\lambda$-calculus with type
variables. Since his seminal paper, parametricity has been extended in innumerable directions---most notably
for our purposes, to dependent type theory~\cite{takeuti01,bernardy10,krishnaswami13,atkey14}.

To motivate Reynolds' insight, suppose we have been given a family of functions
$F \in \typi[A]{\tyuniv}{A \to A}$. There is one obvious term that $F$ could be: the polymorphic identity
function $\tmlam[A]{\tmlam[a]{a}}$. Moreover, this would appear to be the \emph{only} term $F$ could be: if we
are given a type $A$ we know nothing about except that it has an element $a : A$, then the only way we can
produce an element of $A$ is by using the one given to us. This kind of reasoning relies on the fact that
there is no \emph{type-case} function in the type theory; there is no way to write a function like the
following that inspects the shape of $A$.
\[
  \tmlam[A]{\tmlam[a]{(\text{if $A$ is $\tybool$ then $\tmfalse$ else $a$})}} \in \typi[A]{\tyuniv}{A \to A}
\]
Reynolds translated this apparently syntactic property---the lack of constructs for inspecting types---into a
semantic one: if we take a term in type theory and interpret it in set theory, it has an action on
relations. In the case of a term $F \in \typi[A]{\tyuniv}{A \to A}$, its set-theoretic interpretation
$\sem{F}$ has the following property.

\begin{fact}%
  \label{fact:polymorphic-identity}
  Let a pair of sets $A,B$ and a relation $R \subseteq A \times B$ be given. If $R(a,b)$ for some $a \in A$
  and $b \in B$, then $R(\sem{F}Aa,\sem{F}Bb)$.
\end{fact}

This property actually suffices to show that $\sem{F}$ is the polymorphic identity function. Briefly, for any
set $A$ and $a \in A$, we can define the relation $R \subseteq A \times 1$ by $R(a',\_) \eqdef (a' = a)$;
then we have $R(a,*)$, so $R(\sem{F}Aa,\sem{F}1*)$. Note that Fact~\ref{fact:polymorphic-identity} also
immediately implies (though trivially in this case) that $\sem{F}$ is \emph{natural}: for any function of sets
$f \in A \to B$ and $a \in B$, we have $f \circ \sem{F}A = \sem{F}B \circ f$.

In essence, Reynolds' proof consists in defining a relational model of type theory, which Robinson and
Rosolini~\cite{robinson94} reinterpret as a model in the category of \emph{reflexive graphs}. Each type is
modeled by a reflexive graph, with vertices representing elements in the ordinary sense and edges defining a
relation on those elements. Functions take vertices to vertices and edges to edges. Fact~\ref{fact:polymorphic-identity} is then the action of $\sem{F}$ on edges. Atkey, Ghani, and Johann extend the
reflexive graph model to dependent type theory~\cite{atkey14}. In particular, Atkey \etal\ define a universe
whose vertices are sets (discrete reflexive graphs) and edges are relations between those sets. The astute
reader will notice a similarity to Hofmann and Streicher's groupoid model; note that a groupoid is simply a
reflexive graph supporting composition and inverse operations. (Atkey \etal\ make this comparison themselves.)

Can parametricity be used to conquer the problem of smash product coherences? Suppose we have managed to
define an associator $F \in (A_* \wedge_* B_*) \wedge_* C_* \to A_* \wedge_* (B_* \wedge_* C_*)$ and a
candidate inverse $G \in A_* \wedge_* (B_* \wedge_* C_*) \to (A_* \wedge_* B_*) \wedge_* C_*$. (Let us
quantify implicitly over $A_*,B_*,C_*$ for the moment.) For one, we certainly expect parametricity to
guarantee that these functions are natural in their type arguments. To show that they form an isomorphism, we
would need to show $G \circ F$ is the identity function (likewise for $F \circ G$). This is a pointed function
$(A_* \wedge_* B_*) \wedge_* C_* \to (A_* \wedge_* B_*) \wedge_* C_*$; perhaps parametricity can show that the
identity is the \emph{only} such function. (In truth, there is the possibility that it is a constant function,
but we can exclude that case by testing it at $A = B = C = \tybool$.) The pentagon identity establishes the
equality of two isomorphisms
$E, E' \in ((A_* \wedge_* B) \wedge_* C_*) \wedge_* D_* \simeq A_* \wedge_* (B_* \wedge_* (C_* \wedge_*
D_*))$; this we can recast as showing that the composite $E^{-1} \circ E$, regarded as a pointed function
$((A_* \wedge_* B_*) \wedge_* C_*) \wedge_* D_* \to ((A_* \wedge_* B_*) \wedge_* C_*) \wedge_* D_*$, is the
identity. Ultimately, all the higher coherences can be expressed as properties of types of the following form,
where  $A^1_*,\ldots,A^n_*$ are universally quantified type variables.
\[
  (A^1_* \wedge_* \cdots \wedge_* A^n_*) \pto (A^1_* \wedge_* \cdots \wedge_* A^n_*)
\]
We will indeed be able to use parametricity to characterize types of this form, showing that their only
inhabitants are identity and constant functions.

\subsection*{Internalizing parametricity}

Rather than constructing a model and showing that the denotations of terms satisfy parametricity properties,
as Reynolds did, we follow Bernardy, Coquand, and Moulin's recent work~\cite{bernardy12,bernardy13,bernardy15,moulin16} by \emph{internalizing} parametricity as part of our type
theory. Bernardy and Moulin introduce so-called \emph{parametricity primitives}, new type and term formers
that make it possible to prove theorems such as the following.
\[
  \typi*[f]{\typi[A]{\tyuniv}{\tyarr{A}{A}}}{\typi*[A]{\tyuniv}{\typi[a]{A}{\tyid{A}{fAa}{a}}}}
\]
Notably, these primitives have a computational interpretation. We take the ideas of internal parametricity and
apply them to contentful equality, producing a \emph{parametric cubical type theory}.

Internalizing parametricity has the advantage of allowing us to use parametricity results without going
outside the theory. It is, moreover, coherent with the perspective that leads us to the univalence axiom. From
one angle, univalence serves to internalize the action of type-theoretic constructions on isomorphisms. In
much the same way, internal parametricity expresses the action of constructions on relations. We are not the
first to remark on the similarity between the two---both Atkey \etal\ and Bernardy \etal\ make the
observation---but we will endeavor here to sharpen the comparison. Parametric type theory bears a strong
resemblance to cubical type theory, particularly as presented by Bernardy, Coquand, and Moulin (BCM)~\cite{bernardy15}. We will explore that resemblance here, with special attention to the points at which
cubical and parametric type theory diverge.

\subsection*{Contributions}

Our results can be divided into several camps, depending on how they relate to the interplay between internal
parametricity and cubical equality.

First, we establish that parametricity primitives can in fact be added to cubical type theory. Our combined
type theory is grounded in a computational interpretation in the style of Allen~\cite{allen87}, following the
work of Angiuli \etal\ for cubical type theory~\cite{angiuli18}. Starting from the computational
interpretation, we abstract a formal, generalized algebraic type theory. We show that this theory also has
interpretation in (some variety of) Kan bicubical sets. In all these constructions, the cubical side is
already fairly well understood, so we focus on the parametricity primitives.

Next, we come to applications. On the one hand, we use internal parametricity as a tool for proving theorems
in cubical type theory. Here, the smash product is our representative example of a higher inductive type with
complex algebraic structure.  We show that in internally parametric type theory, we can obtain the higher
coherence properties of the smash product in a uniform way. While the proofs are still not trivial, they are
distinguished from the prior work by their scalability: it is not much more difficult to obtain $n$-coherent
structure than $1$-coherent structure.

On the other hand, we use the well-behaved equality of cubical type theory to regularize parametric type
theory. Just as cubical equality produces an extensionality principle for function types, it implies
extensionality principles for the parametricity primitives. In the presence of univalence, we can also make do
with a weaker version of $\tygel$-types, the other parametricity primitive, than is used in the BCM
theory. This allows us to give a simpler model of the theory, avoiding the technical device of \emph{refined
  presheaves} used in the BCM model.

Finally, we compare the design principles underlying cubical and parametric type theory. In both cases, some
kind of structures on pairs of types are represented by maps out of an interval object. In cubical type
theory, the structures are isomorphisms; in parametric type theory, they are relations. As we will see,
parametric type theory has its own analogue of the univalence axiom. However, in parametric type theory it is
key that relations are represented by \emph{affine}, not structural, maps out of the interval object. This
puts parametric type theory in especially close correspondence with the Bezem-Coquand-Huber (BCH) cubical set
model~\cite{bch}, the first constructive model of univalent type theory. Conversely, an affine path interval
does not give rise to a particularly well-behaved contentful equality, being particularly problematic for
modeling higher inductive types; the BCH model has largely been supplanted by structural cubical type theories
and models.

\subsection*{Outline}

We begin by informally reviewing cubical type theory in \cref{sec:cubical}, closely following the
presentations of Angiuli \etal~\cite{angiuli18,abcfhl,angiuli19}. In \cref{sec:parametric}, we mix in the
parametricity primitives. As we go, we compare the components of internal parametricity to their cubical
counterparts.

In \cref{sec:practice}, we put the theory to work, going through a variety of examples that display first
ordinary internal parametricity, then the regularizing effects of cubical equality, and finally the
application of parametricity to the problem of the smash product. In particular, we show how the interaction
between the parametricity primitives and inductive types can be characterized using the relational equivalence
of univalence. We also define and explore the properties of the \emph{sub-universe of bridge-discrete types},
which plays a role in internal parametricity analogous to that of the \emph{identity extension lemma} in
external parametricity. Some of our results are already valid in non-cubical parametric type theory but are
observed for the first time here.

We get precise about the theory beginning in \cref{sec:computational}, where we lay out its computational
interpretation. In \cref{sec:formal} we abstract a generalized algebraic formal type theory which has the
computational interpretation as a model, and in \cref{sec:presheaf} we describe a second model in Kan
cartesian-affine bicubical sets. We consider related work and future directions in \cref{sec:related}.

\section{Cubical type theory}\label{sec:cubical}

Cubical type theory is an extension of Martin-L\"{o}f type theory with an explicitly contentful
equality. These equalities are called \emph{paths}, as they intuitively mimic the notion of path from
topology. To wit, a path in a topological space $X$ is a function $p : \BI \to X$ from the unit interval
$\BI = [0,1]$ into $X$. Such a path connects the endpoints $p(0),p(1) \in X$. In cubical type theory, we
likewise have a type-like object, the interval ``$\BI$'', which contains two distinguished constants $0,1$. We
express paths by hypothesizing \emph{interval variables}: a path in a type $\wftype[\GG]{A}$ is a term
$\wftm[\ctxpdim{\GG}[x]]{P}{A}$ depending on an interval variable $x$. The path connects two endpoints,
$\wftm[\GG]{\usubstdim{P}{0}{x}}{A}$ and $\wftm[\GG]{\usubstdim{P}{1}{x}}{A}$, obtained by substituting the
constants $0,1$ for the interval variable. This judgmental notion of path is internalized by \emph{path
  types}.  Beyond this basic apparatus, every type in cubical type theory supports \emph{Kan operations},
called \emph{coercion} and \emph{composition}, which are used to manipulate paths. Coercion transports terms
between types that are connected by a path; composition implements operations such as transitivity and
symmetry of paths. Finally, additional machinery is required to obtain univalence, the correspondence between
paths of types and isomorphisms.

We follow Angiuli et al.'s account of cubical type theory~\cite{angiuli18,abcfhl}, known as \emph{cartesian
  cubical type theory}. Other cubical type theories and models~\cite{bch,cchm,awodey18,orton18,cavallo20} vary
in their treatment of the interval and formulation of the Kan operations. Although we commit to one theory
here for simplicity, we expect that this paper can be replayed without difficulty using any other.

To begin at the beginning, cubical type theory is---like Martin-L\"{o}f's type theories~\cite{martin-lof75,martin-lof82}---based on four judgments: \emph{$A$ is a type}, \emph{$A$ and $B$ are equal
  types}, \emph{$M$ has type $A$}, and \emph{$M$ and $N$ are equal elements of type $A$}, all relative to a
context $\GG$ of typed variables.
\begin{mathpar}
  \wftype[\GG]{A}
  \and
  \eqtype[\GG]{A}{B}
  \and
  \wftm[\GG]{M}{A}
  \and
  \eqtm[\GG]{M}{N}{A}
\end{mathpar}
A final judgment $\wfctx{\GG}$ (\emph{$\GG$ is a context}) specifies the well-formed variable contexts, which
are lists of assumptions of the form $\ctxsnoc{}[a]{A}$ (\emph{$a$ ranges over terms of type $A$}) among
others we will introduce in a moment. (We will follow standard practice in omitting the prefix $\GG \gg$ from
judgments when the context is irrelevant to the discussion.) Note that the equality judgments express an
external, contentless equality, which is distinct from the contentful path equality. The external ``exact''
equality is necessary on the judgmental level, but it need not be accessible from within the theory.

It is useful to further introduce a \emph{substitution} judgment $\wfsubst[\GG']{\Gg}{\GG}$ (with equality
counterpart $\eqsubst[\GG']{\Gg}{\Gg'}{\GG}$); a substitution is a list
$\Gg = (\substsnoc{\substsnoc{}{M_1}[a_1],\ldots}{M_n}[a_n])$ instantiating each variable in $\GG$ with a term
over the variables in $\GG'$. We write $N\Gg$ for the application of $\Gg$ to a term $N$, that is, the result
of replacing each occurrence of $a_i$ in $N$ with $M_i$. Each of the judgments above is preserved by
substitution; for example, if $\wfsubst[\GG']{\Gg}{\GG}$ and $\wftm[\GG]{M}{A}$, then $\wftm[\GG']{M\Gg}{A\Gg}$.

We think of these judgments as speaking about programs $A,B,M,N$ in some untyped language with an operational
semantics. They are \emph{behavioral specifications}: $\wftype[\GG]{A}$ means that for any instantiation of
the hypotheses $\GG$, the program $A$ computes a value that names some specification. Likewise,
$\wftm[\GG]{M}{A}$ means that $M$ computes to a value satisfying the specification computed by $A$. We use the
notation $\gg$ and $\in$ (as opposed to the typical $\vdash$ and $:$) to indicate that we are speaking about 
this computational interpretation; we will develop a purely formal counterpart for the theory in
\cref{sec:formal}. For the moment, we will be vague about the exact meaning of ``computes'' in the cubical
setting, in the interest of first giving a sense of the shape of cubical and parametric type theory. We lay
out the computational interpretation in detail in \cref{sec:computational}. Until that point, we describe the
system by presenting inference rules that will turn out to be true in the semantics; note that these are
theorems, not definitions.

\subsection{The interval}

Cubical type theory adds a new form of judgment, $\wfpdim[\GG]{r}$ (\emph{$r$ is an interval term}), and its
associated equality judgment $\eqpdim[\GG]{r}{s}$. The two endpoints are interval terms, and we can add
interval variables to the context.
\begin{mathpar}
  \inferrule
  { }
  {\wfpdim[\GG]{0}}
  \and
  \inferrule
  { }
  {\wfpdim[\GG]{1}}
  \and
  \inferrule
  {\wfctx{\GG}}
  {\wfctx{\ctxpdim{\GG}[x]}}
  \and
  \inferrule
  { }
  {\wfpdim[\ctxpdim{\GG}[x]]{x}}
\end{mathpar}
Interval variables behave just like term variables, at least in the sense that they are \emph{structural}: we
have weakening, contraction, and exchange principles, as embodied by the following substitution rules defined
for any $\wfctx{\GG}$.
\begin{mathpar}
  \inferrule[$\BI$-Weakening]
  { }
  {\wfsubst[\ctxpdim{\GG}[x]]{\substprojpdim}{\GG}}
  \and
  \inferrule[$\BI$-Contraction]
  { }
  {\wfsubst[\ctxpdim{\GG}[z]]{(\substsnoc{\substsnoc{\substid[\GG]}{z}[x]}{z}[y])}{(\ctxpdim{\ctxpdim{\GG}[x]}[y])}}
  \and
  \inferrule[$\BI$-Exchange]
  { }
  {\wfsubst[\ctxpdim{\ctxpdim{\GG}[y]}[x]]{(\substsnoc{\substsnoc{\substid[\GG]}{x}[x]}{y}[y])}{(\ctxpdim{\ctxpdim{\GG}[x]}[y])}}
\end{mathpar}
We may also exchange interval variable assumptions with term variable assumptions when it makes type sense to
do so. The contraction and exchange substitutions may be derived from the following more fundamental rule,
which allows us to extend a substitution by a path interval term.
\begin{mathpar}
  \inferrule[$\BI$-Subst]
  {\wfsubst[\GG']{\Gg}{\GG} \\
    \wfpdim[\GG']{r}}
  {\wfsubst[\GG']{(\substbdim{\Gg}{r}[x])}{(\ctxpdim{\GG}[x])}}
\end{mathpar}

Finally, cubical type theory includes one more way to extend the context: with a \emph{constraint}, an
assumption that two interval terms are (exactly) equal. These become relevant when we introduce composition
below.
\begin{mathpar}
  \inferrule
  {\wfpdim[\GG]{r} \\
    \wfpdim[\GG]{s}}
  {\wfcst[\GG]{r = s}}
  \and
  \inferrule
  {\wfcst[\GG]{\Gx}}
  {\wfctx{(\ctxcst{\GG}{\Gx})}}
  \and
  \inferrule
  {\wfpdim[\GG]{r} \\
    \wfpdim[\GG]{s}}
  {\eqpdim[\ctxcst{\GG}{r=s}]{r}{s}}
\end{mathpar}
Once again, we have weakening, exchange, and contraction for constraints.

Aside from these additions, the judgmental apparatus of cubical type theory matches ordinary Martin-L\"{o}f
type theory. We take standard type formers (functions, products, universes) for granted and proceed to the
novel components: $\typath$-types, the Kan operations, $\tyv$-types (which underlie univalence), and higher
inductive types.

\subsection[Path-types]{\texorpdfstring{$\typath$}{Path}-types}

\begin{figure}
  \begin{mathpar}
    \inferrule[Path-Form]
    {\wftype[\ctxpdim{\GG}[x]]{A} \\
      \wftm[\GG]{M_0}{\usubstdim{A}{0}{x}} \\
      \wftm[\GG]{M_1}{\usubstdim{A}{1}{x}}}
    {\wftype[\GG]{\typath{x.A}{M_0}{M_1}}}
    \and
    \inferrule[Path-Intro]
    {\wftm[\ctxpdim{\GG}[x]]{M}{A}}
    {\wftm[\GG]{\tmplam[x]{M}}{\typath{x.A}{\usubstdim{M}{0}{x}}{\usubstdim{M}{1}{x}}}}
    \and
    \inferrule[Path-Elim]
    {\wftm[\GG]{P}{\typath{x.A}{M_0}{M_1}} \\
      \wfpdim[\GG]{r}}
    {\wftm[\GG]{\tmpapp{P}{r}}{\usubstdim{A}{r}{x}}}
    \and
    \inferrule[Path-$\beta$]
    {\wftm[\ctxpdim{\GG}[x]]{M}{A}}
    {\eqtm[\GG]{\tmpapp{(\tmplam[x]{M})}{r}}{\usubstdim{M}{r}{x}}{\usubstdim{A}{r}{x}}}
    \and
    \inferrule[Path-$\partial$]
    {\wftm[\GG]{P}{\typath{x.A}{M_0}{M_1}} \\ \Ge \in \{0,1\}}
    {\eqtm[\GG]{\tmpapp{P}{\Ge}}{M_\Ge}{\usubstdim{A}{\Ge}{x}}}
    \and
    \inferrule[Path-$\eta$]
    {\wftm[\GG]{P}{\typath{x.A}{M_0}{M_1}}}
    {\eqtm[\GG]{P}{\tmplam[x]{\tmpapp{P}{x}}}{\typath{x.A}{M_0}{M_1}}}
  \end{mathpar}
  \caption{Rules for $\typath$-types}%
  \label{fig:path-types}
\end{figure}

$\typath$-types simply internalize dependence on an interval variable, much as function types internalize
dependence on a term variable. When we have a type $\wftype[\ctxpdim{}[x]]{A}$ depending on an interval
variable $x$ and elements $\wftm{M_0}{\usubstdim{A}{0}{x}}$ and $\wftm{M_1}{\usubstdim{A}{1}{x}}$ inhabiting
its endpoints, we can form the type $\typath{x.A}{M_0}{M_1}$ of \emph{paths from $M_0$ to $M_1$ over
  $x.A$}. Recall that the univalence axiom, which we will validate in due time, identifies paths between types
with isomorphisms. With that intuition in mind, we think of an element of $\typath{x.A}{M_0}{M_1}$ as a proof
that $M_0$ corresponds to $M_1$ along the isomorphism between $\usubstdim{A}{0}{x}$ and $\usubstdim{A}{1}{x}$
represented by $x.A$. In the special case that $A$ does not depend on $x$, an element of
$\typath{\_.A}{M_0}{M_1}$ is simply an identification between $M_0$ and $M_1$ in $A$. (In that case, we
generally write $\typath{A}{M_0}{M_1}$ rather than $\typath{\_.A}{M_0}{M_1}$.)

Rules for $\typath$-types are displayed in \cref{fig:path-types}. Like functions, we introduce paths by
abstraction: if $\wftm[\ctxpdim{}[x]]{M}{A}$, then $\tmplam[x]{M}$ is a path from $\usubstdim{M}{0}{x}$ to
$\usubstdim{M}{1}{x}$. Conversely, if we have a path $\wftm{P}{\typath{x.A}{M_0}{M_1}}$, we can apply it to
any interval term $r$ to get an element $\wftm{\tmpapp{P}{r}}{\usubstdim{A}{r}{x}}$. (Moreover, we have
$\tmpapp{P}{0} = M_0$ and $\tmpapp{P}{1} = M_1$.) Abstraction and application interact via the usual $\beta$-
and $\eta$-rules for function types.

Although many theorems rely on the Kan operations introduced in the next section, we can observe some basic
facts about paths already. First, we have reflexive paths given by interval variable weakening.
\begin{mathpar}
  \inferrule
  {\wftm{M}{A}}
  {\wftm{\tmplam[x]{M}}{\typath{A}{M}{M}}}
\end{mathpar}
Second, functions act on paths. Note that we also use weakening here when we apply $F$ in a context extended
with $\ctxpdim{}[x]$.
\begin{mathpar}
  \inferrule
  {\wftm{F}{\typi[a]{A}{B}} \\
    \wftm{P}{\typath{A}{M_0}{M_1}}}
  {\wftm{\tmplam[x]{F(\tmpapp{P}{x})}}{\typath{x.\usubst{B}{\tmpapp{P}{x}}{a}}{FM_0}{FM_1}}}
\end{mathpar}
Finally, we have function extensionality: functions are path-equal when they are pointwise
path-equal. Although function extensionality is a (non-trivial) consequence of univalence~\cite[\S4.9]{hott-book}, cubically it follows more directly from exchange of term and interval variables.
\begin{mathpar}
  \inferrule
  {\wftm{F_0,F_1}{\typi[a]{A}{B}} \\
    \wftm{H}{\typi[a]{A}{\typath{B}{F_0a}{F_1a}}}}
  {\wftm{\tmplam[x]{\tmlam[a]{\tmpapp{Ha}{x}}}}{\typath{\typi[a]{A}{B}}{F_0}{F_1}}}
\end{mathpar}
It is easy to see that this function is an isomorphism---its inverse simply exchanges the arguments in the
opposite order.

The preceding argument can more generally characterize $\typath{x.\typi[a]{A}{B}}{F_0}{F_1}$ when $B$ depends
on $x$, but not when $A$ does: if $A$ depends on $x$, then the type
``$\typi[a]{A}{\typath{x.B}{F_0a}{F_1a}}$'' is nonsensical. In the most general case, we can instead construct
a map taking paths between functions to functions from paths to paths: ``equal functions take equal arguments
to equal results.''

\begin{lem}\label{lem:funapp-over}
  Let $\wftype[\ctxpdim{}[x]]{A}$, $\wftype[\ctxsnoc{\ctxpdim{}[x]}[a]{A}]{B}$,
  $\wftm{F_0}{\usubstdim{(\typi[a]{A}{B})}{0}{x}}$, and $\wftm{F_1}{\usubstdim{(\typi[a]{A}{B})}{1}{x}}$ be
  given. Then we have the following principle.
  \begin{mathpar}
    \inferrule
    {\wftm{Q}{\typath{x.\typi[a]{A}{B}}{F_0}{F_1}}}
    {\wftm{\mathsf{funapp}(Q)}{\typi*[a_0]{\usubstdim{A}{0}{x}}{\typi*[a_1]{\usubstdim{A}{1}{x}}{\typi[p]{\typath{x.A}{a_0}{a_1}}{\typath{x.\usubst{B}{\tmpapp{p}{x}}{a}}{F_0a_0}{F_1a_1}}}}}}
  \end{mathpar}
\end{lem}
\begin{proof}
  $\mathsf{funapp}(Q) \eqdef \tmlam[a_0]{\tmlam[a_1]{\tmlam[p]{\tmplam[x]{(\tmpapp{Q}{x})(\tmpapp{p}{x})}}}}$.
\end{proof}

Constructing an inverse to this function will require the coercion operator introduced in the following
section.

\subsection{Kan operations: coercion and composition}\label{sec:cubical:kan}

\begin{figure}
  \begin{mathpar}
    \mprset{vskip=0.3ex}
    \inferrule[Coercion]
    {\wftype[\GG,\ofp{x}]{A} \\
      \wfpdim[\GG]{r,s} \\
      \wftm[\GG]{M}{\usubstdim{A}{r}{x}}}
    {\wftm[\GG]{\coe{x.A}{r}{s}{M}}{\usubstdim{A}{s}{x}} \\\\
      \eqtm[\GG]{\coe{x.A}{r}{r}{M}}{M}{\usubstdim{A}{r}{x}}}
    \and
    \inferrule[Homogeneous composition]
    {\wftype[\GG]{A} \\
      \wfpdim[\GG]{r,s} \\
      \wftm[\GG]{M}{A} \\\\
      (\forall i)\; \wfcst[\GG]{\Gx_i} \\
      (\forall i)\; \wftm[\ctxpdim{\ctxcst{\GG}{\Gx_i}}[x]]{N_i}{A} \\
      (\forall i)\; \eqtm[\ctxcst{\GG}{\Gx_i}]{M}{\usubstdim{N_i}{r}{x}}{A} \\
      (\forall i,j)\; \eqtm[\ctxpdim{\ctxcst{\ctxcst{\GG}{\Gx_i}}{\Gx_j}}[x]]{N_i}{N_j}{A}}
    {\wftm[\GG]{\hcomp{A}{r}{s}{M}{\sys{\xi_i}{x.N_i}}}{A} \\\\
      (\forall j)\; \eqtm[\ctxcst{\GG}{\xi_j}]{\hcomp{A}{r}{s}{M}{\sys{\xi_i}{x.N_i}}}{\usubstdim{N_j}{s}{x}}{A} \\\\
      \eqtm[\GG]{\hcomp{A}{r}{r}{M}{\sys{\xi_i}{x.N_i}}}{M}{A}}
    \and
    \inferrule[Heterogeneous composition]
    {\wftype[\ctxpdim{\GG}[x]]{A} \\
      \wfpdim[\GG]{r,s} \\
      \wftm[\GG]{M}{\usubstdim{A}{r}{x}} \\\\
      (\forall i)\; \wfcst[\GG]{\Gx_i} \\
      (\forall i)\; \wftm[\ctxpdim{\ctxcst{\GG}{\Gx_i}}[x]]{N_i}{A} \\
      (\forall i)\; \eqtm[\ctxcst{\GG}{\Gx_i}]{M}{\usubstdim{N_i}{r}{x}}{\usubstdim{A}{r}{x}} \\
      (\forall i,j)\; \eqtm[\ctxpdim{\ctxcst{\ctxcst{\GG}{\Gx_i}}{\Gx_j}}[x]]{N_i}{N_j}{A}}
    {\wftm[\GG]{\comp{x.A}{r}{s}{M}{\sys{\xi_i}{x.N_i}}}{\usubstdim{A}{s}{x}} \\\\
      (\forall j)\; \eqtm[\ctxcst{\GG}{\xi_j}]{\comp{x.A}{r}{s}{M}{\sys{\xi_i}{x.N_i}}}{\usubstdim{N_j}{s}{x}}{\usubstdim{A}{s}{x}} \\\\
      \eqtm[\GG]{\comp{x.A}{r}{r}{M}{\sys{\xi_i}{x.N_i}}}{M}{\usubstdim{A}{r}{x}}}
  \end{mathpar}
  \caption{Rules for coercion, homogeneous composition, and heterogeneous composition}%
  \label{fig:kan}
\end{figure}

The judgmental path structure of cubical type theory endows each type with a ``path'' relation. So far, this
relation is not quite a proper notion of equality. For one, while it is reflexive, it need not be symmetric or
transitive. Perhaps more importantly, we do not know that type families \emph{respect} paths in the following
sense. If we have some family $\wftype[\ctxsnoc{}[a]{A}]{B}$ and a path $\wftm{P}{\typath{A}{M_0}{M_1}}$, we
expect that for every element of $BM_0$, there is a corresponding element of $BM_1$. If we think of $B$ as a
predicate on elements of $A$, we are saying that $M_1$ should satisfy the same properties as $M_0$. In fact,
we would expect that $BM_0$ and $BM_1$ are isomorphic. At the moment, however, we only know that there is a
path $x.B(\tmpapp{P}{x})$ from $BM_0$ to $BM_1$. What we need, then, is one direction of the univalence axiom:
the ability to transform paths between types into isomorphisms. This is effected by the \emph{coercion}
operator $\coe$, which satisfies the first rule in \cref{fig:kan}.

Given a term at some index $r$ of a type path $x.A$, coercion produces an element at any other $s$. We can
show that $\coe{x.A}{r}{s}{-} \in \usubstdim{A}{r}{x} \to \usubstdim{A}{s}{x}$ is in fact an isomorphism. The
full proof relies on composition, which we have not yet introduced, but we can at least see that
$\coe{x.A}{1}{0}{-}$ is inverse to $\coe{x.A}{0}{1}{-}$ up to a path.
\begin{mathpar}
  \inferrule
  {\wftm{M}{\usubstdim{A}{0}{x}}}
  {\wftm{\tmplam[y]{\coe{x.A}{y}{0}{\coe{x.A}{0}{y}{M}}}}{\typath{\usubstdim{A}{0}{x}}{M}{\coe{x.A}{1}{0}{\coe{x.A}{0}{1}{M}}}}}
\end{mathpar}

Operationally, coercion evaluates by cases on the shape of the type path $x.A$. For example, the following
equation describes the behavior of coercion at a product type $x.\tysigma[a]{A}{B}$.
\begin{mathpar}
  \inferrule
  {\wftype[\ctxpdim{}[x]]{A} \\
    \wftype[\ctxsnoc{\ctxpdim{}[x]}[a]{A}]{B} \\
    \wftm{M}{\usubstdim{(\tysigma[a]{A}{B})}{r}{x}}}
  {\eqtm{\coe{x.\tysigma[a]{A}{B}}{r}{s}{M}}{\tmpair{\coe{x.A}{r}{s}{\tmfst{M}}}{\coe{x.\usubst{B}{\coe{x.A}{r}{x}{\tmfst{M}}}{a}}{r}{s}{\tmsnd{M}}}}{\usubstdim{(\tysigma[a]{A}{B})}{s}{x}}}
\end{mathpar}

\emph{Homogeneous composition} (which we will often just call composition) serves a more technical purpose: to
evaluate coercions along lines of the form $x.\typath{y.A}{N_0}{N_1}$. For the moment, let us assume that $A$
does not depend on $x$. In order to execute such a coercion, we must be able to adjust the endpoints of a
given path by another pair of paths. That is, given $\wftm{M}{\typath{y.A}{M_0}{M_1}}$ and lines $x.N_0$,
$x.N_1$ fitting into the following shape, we should be able to produce a new, ``adjusted'' path shown as a
dashed line below.
\[
  \begin{tikzpicture}
    \draw (-0.2 , 2.4) [->] to node [above] {\small $y$} (0.3 , 2.4) ;
    \draw (-0.2 , 2.4) [->] to node [left] {\small $x$} (-0.2 , 1.9) ;
    \node (tl) at (1.5 , 2) {$M_0$} ;
    \node (tr) at (6 , 2) {$M_1$} ;
    \node (bl) at (1.5 , 0.5) {$\bullet$} ;
    \node (br) at (6 , 0.5) {$\bullet$} ;
    \draw (tl) [->] to node [above] {$\tmpapp{M}{y}$} (tr) ;
    \draw (tl) [->] to node [left] {$N_0$} (bl) ;
    \draw (tr) [->] to node [right] {$N_1$} (br) ;
    \draw (bl) [->,dashed] to node [below] {$\exists$} (br) ;
  \end{tikzpicture}
\]

Homogeneous composition, written $\hcomp$, is a generalized form of this operation that adjusts the boundary
of a term, a boundary being specified by a sequence of constraints on interval variables. As an example, the
adjusted path above is obtained as the following composite.
\[
  \wftm[\ctxpdim{}[y]]{\hcomp{A}{0}{1}{\tmpapp{M}{y}}{\tube{y=0}{x.N_0},\tube{y=1}{x.N_1}}}{A}
\]
The general operator has the form $\hcomp{A}{r}{s}{M}{\sys{\Gx_i}{x.N_i}}$; it is characterized by the second
rule of \cref{fig:kan}. We use the notation $\sys{\Gx_i}{x.N_i}$ to denote a finite list of constraint-line
pairs $\tube{\Gx_1}{x.N_1},\ldots,\tube{\Gx_n}{x.N_n}$, implicitly quantifying over an indexing variable
$i$. Like coercion, we define homogeneous composition by case analysis of the type argument. Where the special
case involving a pair of constraints $y = 0$ and $y = 1$ on a single interval variable is enough for
\emph{coercion} in the path type, the general form becomes necessary to implement \emph{composition} in the
path type; the general form thus represents a ``strengthened induction hypothesis''.

To handle coercion along $x.\typath{y.A}{N_0}{N_1}$ when $A$ \emph{does} depend on $x$, we can combine
coercion and composition into a unified \emph{heterogeneous composition} operator, $\comp$, which coerces an
input across a type line while simultaneously adjusting by a boundary path along that line. Defined as
follows, $\comp$ satisfies the third rule shown in \cref{fig:kan}.
\begin{align*}
  {\comp{x.A}{r}{s}{M}{\sys{\xi_i}{x.N_i}} \eqdef \hcomp{\usubstdim{A}{s}{x}}{r}{s}{\coe{x.A}{r}{s}{M}}{\sys{\xi_i}{x.\coe{x.A}{x}{s}{N_i}}}}
\end{align*}
Both $\hcomp$ and $\coe$ can be recovered from $\comp$, so the latter is may be taken as primitive instead, as
in~\cite{cchm,angiuli18}. Either way, the ability to decompose $\comp$ into $\hcomp$ and $\coe$ plays a key
role in defining Kan operations for higher inductive types~\cite{coquand18,cavallo19b}.

Coercion and composition are together referred to as the \emph{Kan operations}, being inspired by the Kan
condition of algebraic topology~\cite{kan55}. For each type we wish to introduce to cubical type theory, we
must explain how the Kan operations evaluate at that type. This can be carried out for all the standard type
formers of Martin-L\"{o}f type theory (functions, products, inductive types, universes); we refer to Angiuli~\cite{angiuli19} for a thorough accounting of those results.

Using coercion, we can prove the converse to Lemma~\ref{lem:funapp-over}: if two functions take equal arguments to equal results, then they are equal as functions.

\begin{lem}\label{lem:funext-over}
  Let $\wftype[\ctxpdim{}[x]]{A}$, $\wftype[\ctxsnoc{\ctxpdim{}[x]}[a]{A}]{B}$,
  $\wftm{F_0}{\usubstdim{(\typi[a]{A}{B})}{0}{x}}$, and $\wftm{F_1}{\usubstdim{(\typi[a]{A}{B})}{1}{x}}$ be
  given. Then we have the following.
  \begin{mathpar}
    \inferrule
    {\wftm{H}{\typi*[a_0]{\usubstdim{A}{0}{x}}{\typi*[a_1]{\usubstdim{A}{1}{x}}{\typi[p]{\typath{x.A}{a_0}{a_1}}{\typath{x.\usubst{B}{\tmpapp{p}{x}}{a}}{F_0a_0}{F_1a_1}}}}}}
    {\wftm{\mathsf{funext}(H)}{\typath{x.\typi[a]{A}{B}}{F_0}{F_1}}}
  \end{mathpar}
\end{lem}
\begin{proof}
  $\mathsf{funext}(H) \eqdef
    \tmplam[x]{\tmlam[a]{H(\coe{x.A}{x}{0}{a})(\coe{x.A}{x}{1}{a})(\tmplam[y]{\coe{x.A}{x}{y}{a}})}}$.
\end{proof}

Essentially, given an interval variable $\ctxpdim{}[x]$ and an element $a$ of $A$ (at index $x$), we can
extend the point $a$ to a path over $x.A$ by coercion.

Coercion and composition also give us an analogue of the Martin-L\"{o}f identity type elimination principle
(often called ``J'') for paths.

\begin{lem}\label{lem:j}
  Let $\wftype{A}$ and $\wftm{M}{A}$ be given. Suppose we are given the following:
  \begin{itemize}[label=$\triangleright$]
  \item $\wftype[\ctxsnoc{\ctxsnoc{}[a]{A}}[p]{\typath{A}{M}{a}}]{C}$,
  \item $\wftm{N}{\usubst{C}{M,\tmplam{\_}{M}}{a,p}}$,
  \item $\wftm{M'}{A}$ and $\wftm{P}{\typath{A}{M}{M'}}$.
  \end{itemize}
  Then there is some $\wftm{\tmj{a.p.C}{N}{P}}{\usubst{C}{M',P}{a,p}}$.
\end{lem}
\begin{proof}
  Define an auxiliary $\wftm[\ctxpdim{\ctxpdim{}[x]}[y]]{Q}{A}$ as follows.
  \[
    Q \eqdef \hcomp{A}{0}{y}{\tmpapp{P}{0}}{\tube{x=0}{\_.\tmpapp{P}{0}},\tube{x=1}{y.\tmpapp{P}{y}}}
  \]
  Set $\tmj{a.p.C}{N}{P} \eqdef \coe{x.\usubst{C}{\usubstdim{Q}{1}{y},\tmplam[y]{Q}}{a,p}}{0}{1}{N}$.
\end{proof}

This is slightly weaker than the elimination principle enjoyed by Martin-L\"{o}f's elimination principle, as
it is not the case that $\eqtm{\tmj{a.p.C}{N}{\tmplam[\_]{M}}}{N}{\usubst{C}{M,\tmplam[\_]{M}}{a,p}}$ in
general; this equation may be shown to hold up to a path, but does not hold up to exact equality. One may
separately introduce identity types to cubical type theory that do satisfy this principle, either via a
special construction~\cite{cchm,abcfhl} or as particular indexed inductive types~\cite{cavallo19b}, and in
this case one has $\tyid{A}{M}{M'} \simeq \typath{A}{M}{M'}$. By univalence, this isomorphism implies that
path and identity types satisfy the same theorems; in particular, it justifies our citing theorems about
identity types in homotopy type theory as theorems about path types going forward. Of course, these theorems
are often more easily proven in cubical type theory by reasoning directly with paths.

\subsection[V-types and univalence]{\texorpdfstring{$\tyv$}{V}-types and univalence}\label{sec:cubical:v}

\begin{figure}
  \begin{mathpar}
    \inferrule[V-Form]
    {\wftype[\ctxpeq{\GG}{r}{0}]{A} \\
      \wftype[\GG]{B} \\
      \wftm[\ctxpeq{\GG}{r}{0}]{I}{\tyiso{A}{B}}}
    {\wftype[\GG]{\tyv{r}{A}{B}{I}}}
    \and
    \inferrule[V-Form-$\partial_0$]
    {\wftype[\GG]{A} \\
      \wftype[\GG]{B} \\
      \wftm{I}{\tyiso{A}{B}}}
    {\eqtype[\GG]{\tyv{0}{A}{B}{I}}{A}}
    \and
    \inferrule[V-Form-$\partial_1$]
    {\wftype[\GG]{B}}
    {\eqtype[\GG]{\tyv{1}{A}{B}{I}}{B}}
    \and
    \inferrule[V-Intro]
    {\wftm[\ctxpeq{\GG}{r}{0}]{M}{A} \\
      \wftm[\GG]{N}{B} \\
      \eqtm[\ctxpeq{\GG}{r}{0}]{\tmfst{I}(M)}{N}{B}}
    {\wftm[\GG]{\tmvin{r}{M}{N}}{\tyv{r}{A}{B}{I}}}
    \and
    \inferrule[V-Intro-$\partial_0$]
    {\wftm[\GG]{M}{A} \\
      \wftm[\GG]{N}{B} \\
      \eqtm[\GG]{\tmfst{I}(M)}{N}{B}}
    {\eqtm[\GG]{\tmvin{0}{M}{N}}{M}{A}}
    \and
    \inferrule[V-Intro-$\partial_1$]
    {\wftm[\GG]{N}{B}}
    {\eqtm[\GG]{\tmvin{1}{M}{N}}{N}{B}}
    \and
    \inferrule[V-Elim]
    {\wftm[\GG]{P}{\tyv{r}{A}{B}{I}}}
    {\wftm[\GG]{\tmvproj{r}{P}{I}}{B}}
    \and
    \inferrule[V-Elim-$\partial_0$]
    {\wftm[\GG]{P}{A} \\
      \wftm{I}{\tyiso{A}{B}}}
    {\eqtm[\GG]{\tmvproj{0}{P}{I}}{\tmfst{I}(P)}{B}}
    \and
    \inferrule[V-Elim-$\partial_1$]
    {\wftm[\GG]{P}{B}}
    {\eqtm[\GG]{\tmvproj{1}{P}{I}}{P}{B}}
  \end{mathpar}
  \caption{Rules for $\tyv$-types. See~\cite{angiuli19} for $\beta$- and $\eta$-rules.}%
  \label{fig:v}
\end{figure}

The Kan operations account for one direction of the univalence axiom: the mapping from paths between types to
isomorphisms. The inverse is defined using $\tyv$-types, which produce paths in the universe from
isomorphisms.%
\footnote{%
  Some formulations of cubical type theory instead use \emph{$\tyglue$-types}, which have $\tyv$-types as a
  special case. The points we make here about $\tyv$-types apply equally well to $\tyglue$-types.
}

First, let us take the opportunity to define isomorphism precisely.

\begin{defi}
  Let a function $\wftm{F}{\tyarr{A}{B}}$ be given. The types $\tylinv{A}{B}{F}$ and $\tyrinv{A}{B}{F}$ of
  left and right inverses to $F$ are defined as follows.
  \begin{align*}
    \tylinv{A}{B}{F} &\eqdef \tysigma[g]{\tyarr{B}{A}}{(\typi[a]{A}{\typath{A}{g(Fa)}{a}})} \\
    \tyrinv{A}{B}{F} &\eqdef \tysigma[g]{\tyarr{B}{A}}{(\typi[b]{B}{\typath{B}{F(gb)}{b}})}
  \end{align*}
  We say $F$ is an isomorphism when it is equipped with a left and right inverse.
  \begin{align*}
    \tyisiso{A}{B}{F} &\eqdef \typrod{\tylinv{A}{B}{F}}{\tyrinv{A}{B}{F}}
  \end{align*}
  The type of isomorphisms between $A$ and $B$ is then
  $\tyiso{A}{B} \eqdef \tysigma[f]{\tyarr{A}{B}}{\tyisiso{A}{B}{f}}$.
\end{defi}

Isomorphisms are frequently known as \emph{equivalences} in the literature on univalent type theory. There are
several isomorphic formulations of the type $\tyiso{A}{B}$; we refer to~\cite[Chapter 4]{hott-book} for more
details. (Our definition is there called a \emph{bi-invertible map}). A key property of $\tyisiso{A}{B}{F}$ is
that it is a proposition in the following sense~\cite[Theorem 4.3.2]{hott-book}.

\begin{defi}\label{def:proposition}
  $\wftype{A}$ is a \emph{proposition} if any two elements of $A$ are equal up to a path, as captured by the
  following type.
  \[
    \tyisprop{A} \eqdef \typi*[a]{A}{\typi[b]{A}{\typath{A}{a}{b}}}
  \]
\end{defi}

While the $\tyv$-type is used principally to convert isomorphisms to paths, it is a bit more general: it takes
a path and an isomorphism and composes them to produce a new path. That is, if we have a path of types $B$ in
a direction $x$ and an isomorphism $I$ between some $A$ and $\usubstdim{B}{0}{x}$, their $\tyv$-type fits into
the following (``$\tyv$-shaped'') diagram.
\[
  \begin{tikzcd}[column sep=4em,row sep=1.7em]
    A \ar[phantom]{d}[left=0.5em,font=\normalsize]{I}[sloped,font=\normalsize,xshift=0.1ex]{\simeq} \ar[dashed]{dr}[font=\normalsize]{\tyv{x}{A}{B}{I}} \\
    B_0 \ar{r}[below,font=\normalsize]{B} & B_1 \\[-1.8em]
    {x\to}
  \end{tikzcd}
\]

Rules for $\tyv$-types are shown in \cref{fig:v}. We convert isomorphisms to paths in the universe by applying
$\tyv$ with a degenerate path.
\begin{mathpar}
  \inferrule
  {\wftm{A}{\tyuniv} \\
    \wftm{B}{\tyuniv} \\
    \wftm{I}{\tyiso{A}{B}}}
  {\wftm{\tmua{A}{B}{I} \eqdef \tmplam[x]{\tyv{x}{A}{B}{I}}}{\typath{\tyuniv}{A}{B}}}
\end{mathpar}

Here, $x$ does not appear in $B$, so we are composing the isomorphism $I$ with the reflexive path $\_.B$. This
reflexive path corresponds to the identity isomorphism on $B$, so when we pre-compose with $I$ we simply get a
path corresponding to $I$.

We will not be using $\tyv$-types directly in the future, only the univalence axiom that they enable. Rather,
we introduce them here in order to make a comparison with their parametric equivalent in
\cref{sec:parametric:gel}. For that purpose, let us give some intuition as to why $\tyv$ is formulated as it
is. Univalence involves a ``dimension shift'': it takes a point in the type of isomorphisms and produces a
path in the universe, which is an element one dimension higher. However, we cannot impose in the typing rule
for $\tyv{x}{A}{B}{I}$ that $A,B,I$ live ``one dimension lower,'' \ie, are degenerate in $x$, because this 
property is not stable under substitution. For example, $\tmzmod{M}{x}$ may be degenerate in some $y$, but
$\usubstdim{\tmzmod{M}{x}}{y}{x}$ is certainly not degenerate in $\usubstdim{y}{y}{x}$. All aspects of type
theory should be stable under substitution, so this is a non-starter. Instead, we structure $\tyv{r}$ in such
a way that it does not involve a dimension shift; both the input and the output vary in the direction $r$.

\subsection{Higher inductive types}

Finally, cubical type theory can include a variety of \emph{higher inductive types}. These can be seen as a
mutual generalization of inductive types and quotients; they are inductive definitions that permit \emph{path
  constructors} in addition to ordinary constructors.

It is beyond the scope of this work to give a comprehensive account of higher inductive types in cartesian
cubical type theory; for that, we refer to~\cite{cavallo19b}. We will instead go by way of example, expanding
on the type $\tyzmod$ of integers \emph{mod} 2 specified in the introduction.
\[
  \begin{array}{l}
    \dataheading{\tyzmod} \\
    \niceconstr{\tmzin}[\ctxsnoc{}[n]{\tyz}]{\tyzmod} \\
    \niceconstr{\tmzmod}[\ctxpdim{\ctxsnoc{}[n]{\tyz}}[x]]{\tyzmod}{\tube{x=0}{\tmzin{n}} \mid \tube{x=1}{\tmzin{n+2}}}
  \end{array}
\]

The $\tmzmod$ constructor exemplifies the format of a path constructor: it takes one or more interval
variables as arguments, and it has a specified boundary which can refer to its arguments and previous
construtors. This specification indicates the following introduction and boundary rules for $\tmzin$ and
$\tmzmod$.
\begin{mathpar}
  \inferrule
  {\wftm[\GG]{N}{\tyz}}
  {\wftm[\GG]{\tmzin{N}}{\tyzmod}}
  \and
  \inferrule
  {\wftm[\GG]{N}{\tyz} \\
    \wfpdim[\GG]{r}}
  {\wftm[\GG]{\tmzmod{N}{r}}{\tyzmod}}
  \\
  \inferrule
  {\wftm[\GG]{N}{\tyz}}
  {\eqtm[\GG]{\tmzmod{N}{0}}{\tmzin{N}}{\tyzmod}}
  \and
  \inferrule
  {\wftm[\GG]{N}{\tyz}}
  {\eqtm[\GG]{\tmzmod{N}{1}}{\tmzin{N + 2}}{\tyzmod}}
\end{mathpar}
The eliminator for $\tyzmod$ naturally takes clauses to handle the $\tmzin$ and $\tmzmod$ cases. The $\tmzmod$
case is required to cohere with the $\tmzin$ case on its boundary, which ensures that every function out of
$\tyzmod$ takes $\tmzin{n}$ and $\tmzin{n+2}$ to path-equal results.
\begin{mathpar}
  \inferrule
  {\wftype[\ctxsnoc{\GG}[a]{\tyzmod}]{C} \\
    \wftm[\GG]{M}{\tyzmod} \\
    \wftm[\ctxsnoc{\GG}[n]{\tyz}]{Q_{\tmzin}}{\usubst{C}{\tmzin{n}}{a}} \\
    \wftm[\ctxpdim{\ctxsnoc{\GG}[n]{\tyz}}[x]]{Q_{\tmzmod}}{\usubst{C}{\tmzmod{n}{x}}{a}} \\\\
    \eqtm[\ctxsnoc{\GG}[n]{\tyz}]{\usubstdim{Q_{\tmzmod}}{0}{x}}{Q_{\tmzin}}{\usubst{C}{\tmzin{n}}{a}} \\
    \eqtm[\ctxsnoc{\GG}[n]{\tyz}]{\usubstdim{Q_{\tmzmod}}{1}{x}}{\usubst{Q_{\tmzin}}{n+2}{n}}{\usubst{C}{\tmzin{n+2}}{a}}}
  {\wftm[\GG]{\tmzmodelim{a.C}{M}{n.Q_{\tmzin}}{n.x.Q_{\tmzmod}}}{\usubst{C}{M}{a}}}
\end{mathpar}
When applied to a constructor, the eliminator steps accordingly as shown below.
\begin{mathpar}
  \eqtm{\tmzmodelim{a.C}{\tmzin{N}}{n.Q_{\tmzin}}{n.x.Q_{\tmzmod}}}{\usubst{Q_{\tmzin}}{N}{n}}{\usubst{C}{\tmzin{N}}{a}}
  \and
  \eqtm{\tmzmodelim{a.C}{\tmzmod{N}{r}}{n.Q_{\tmzin}}{n.x.Q_{\tmzmod}}}{\usubstdim{\usubst{Q_{\tmzmod}}{N}{n}}{r}{x}}{\usubst{C}{\tmzmod{N}{r}}{a}}
\end{mathpar}

\section{Parametric type theory}\label{sec:parametric}

We now proceed to add parametricity primitives to our cubical type theory. We follow the blueprint of
Bernardy, Coquand, and Moulin (BCM)~\cite{bernardy15}, which is a substantial simplification of Bernardy and
Moulin's original parametric theory~\cite{bernardy12}. The BCM parametric type theory has the same basic shape
as cubical type theory: relatedness is represented by maps out of an interval object $\BFI$. We henceforth
refer to $\BI$ as the \emph{path interval} and $\BFI$ as the \emph{bridge interval}; we call maps out of
$\BFI$ \emph{bridges}, following~\cite{nuyts17}. (As a general rule, we use boldface to distinguish bridge
constructs from their path equivalents.) The connection between internal parametricity and cubical type theory
has never been a secret; Bernardy and Moulin already remark on the similarity in~\cite{bernardy12}, and later
iterations of their work resemble cubical type theory even more strongly.

We go a bit further and compare the two in detail over the course of this section. First, there is the obvious
difference: parametric type theory has no analogues of coercion and composition. More subtle is the difference
between the two intervals $\BI$ and $\BFI$: the path interval behaves structurally, but the bridge interval is
\emph{affine}. This has two essential effects on the theory. First, it enables a ``function extensionality''
principle analogous to Lemma~\ref{lem:funext-over} that does not rely on coercion. Second, it means that we
can avoid the $\tyv$-shape of $\tyv$-types, instead supporting a type former ($\tygel$) that directly converts
relations to bridges.

On a more mundane level, we present the parametricity elements using a notation more similar to that of
cubical type theory. For a translation to Bernardy \etal's (substantially different) notation, see
Figure~\ref{fig:translation} on page~\pageref{fig:translation}.

\subsection{The bridge interval}\label{sec:parametric:interval}

Recall our intuition for a term $\wftm[\ctxpdim{}[x]]{M}{A}$: the path $x.A$ stands for an isomorphism
$\usubstdim{A}{0}{x} \simeq \usubstdim{A}{1}{x}$ via univalence, and $x.M$ is a proof that
$\usubstdim{M}{0}{x}$ corresponds to $\usubstdim{M}{1}{x}$ across this isomorphism. Likewise, a \emph{bridge}
of types $\wftype[\ctxbdim{}[\bmx]]{A}$ stands for a binary relation on $\usubstdim{A}{\bm0}{\bmx}$ and
$\usubstdim{A}{\bm1}{\bmx}$, and a term $\wftm[\ctxbdim{}[\bmx]]{M}{A}$ is a proof that
$\usubstdim{M}{\bm0}{\bmx}$ and $\usubstdim{M}{\bm1}{\bmx}$ stand in this relation.

We start with a judgment $\wfbdim[\GG]{\bmr}$. Like the path interval, it is populated by two endpoint
$\bm0$ and $\bm1$, and we can suppose bridge interval variables.
\begin{mathpar}
  \inferrule
  { }
  {\wfbdim[\GG]{\bm0}}
  \and
  \inferrule
  { }
  {\wfbdim[\GG]{\bm1}}
  \and
  \inferrule
  {\wfctx{\GG}}
  {\wfctx{\ctxbdim{\GG}[x]}}
  \and
  \inferrule
  { }
  {\wfbdim[\ctxbdim{\GG}[x]]{\bmx}}
\end{mathpar}
Unlike path variables, however, we will only have weakening and exchange for the bridge interval: the
contraction principle fails. The bridge interval is thus substructural, in particular
\emph{affine}.

The lack of contraction means that we cannot always apply a bridge variable substitution
$\usubstdim{-}{\bmy}{\bmx}$ to a term $M$: if $M$ already mentions $\bmy$, this amounts to contracting $\bmy$
and $\bmx$. What we have is \emph{fresh substitution}: we can substitute a variable $\bmy$ for $\bmx$ in $M$
when $\bmy$ does not occur in $M$ (\ie, is \emph{apart} from $M$). To formulate fresh substitution for open
terms, we define the following \emph{context restriction} operation, roughly following Cheney's approach to
nominal type theory~\cite{cheney12}. Intuitively, given a context $\GG$ and interval term $\bmr$ in that
context, $\ctxres{\GG}{\bmr}$ is the part of $\GG$ guaranteed to be apart from $\bmr$: when $\bmr$ is a
variable $\bmx$, it includes all other bridge variables, all path variables, constraints that do not involve
$\bmr$, and those term variables that are introduced before $\bmr$. The constants $\bm0$ and $\bm1$ are
considered to be apart from everything.  That is, we define $\ctxres{\GG}{\bmr} \eqdef \GG$ when
$\eqbdim[\GG]{\bmr}{\bm\Ge}$ for some $\bm\Ge \in \{\bm0,\bm1\}$ and as follows otherwise.
\begin{align*}
  \ctxres{(\ctxpdim{\GG}[y])}{\bmx} &\eqdef \ctxpdim{\ctxres{\GG}{\bmx}}[y] \\
  \ctxres{(\ctxsnoc{\GG}[a]{A})}{\bmx} &\eqdef \ctxres{\GG}{\bmx} \\
  \ctxres{(\ctxbdim{\GG}[\bmy])}{\bmx} &\eqdef
  \left\{
    \begin{array}{ll}
      \GG & \text{if $\bmx = \bmy$} \\
      \ctxbdim{\ctxres{\GG}{\bmx}}[\bmy] & \text{if $\bmx \neq \bmy$}
    \end{array}
  \right. \\
  \ctxres{(\ctxcst{\GG}{\Gx})}{\bmx} &\eqdef%
  \left\{
    \begin{array}{ll}
      \ctxres{\GG}{\bmx} & \text{if $\bmx$ occurs in $\Gx$} \\
      \ctxcst{\ctxres{\GG}{\bmx}}{\Gx} & \text{otherwise}
    \end{array}
  \right.
\end{align*}
We then have the following rule for extending a substitution by a bridge interval term.
\begin{mathpar}
  \inferrule[$\BFI$-Subst]
  {\wfbdim[\GG']{\bmr} \\
    \wfsubst[\ctxres{\GG'}{\bmr}]{\Gg}{\GG}}
  {\wfsubst[\GG']{(\substbdim{\Gg}{\bmr}[\bmx])}{\GG}}
\end{mathpar}
The restriction in the premises prevents us from deriving, in particular, the following contraction or
``diagonal'' substitution, which attempts to substitute the same bridge variable $\bmx$ for two distinct
variables $\bmy$ and $\bmz$.
\[
  \wfsubst[\ctxbdim{}[x]]{(\substbdim{\substbdim{}{\bmx}[\bmy]}{\bmx}[\bmz])}{(\ctxbdim{\ctxbdim{}[y]}[z])} \bad
\]

When working with a context of the form $(\ctxbdim{\GG}[x],\GG')$, we therefore think of the variables in
$\GG$ as being apart from $\bmx$: we are disallowed from substituting a term that mentions $\bmx$ for a
variable in $\GG$: in a substitution. On the other hand, we \emph{can} substitute terms that mention $\bmx$
for variables in $\GG'$. In accordance with this intuition, we can exchange term variables past bridge
variables in one direction but not the other, as witnessed by the following substitution.
\begin{mathpar}
  \wfsubst[\ctxbdim{\ctxsnoc{}[a]{A}}[x]]{(\substsnoc{\substpdim{}{\bmx}[\bmx]}{a}[a])}{(\ctxsnoc{\ctxbdim{}[x]}[a]{A})}
\end{mathpar}
In the domain of this substitution, $\ctxsnoc{}[a]{A}$ ranges over fewer terms: only those elements of $A$
that are apart from $\bmx$.

In keeping with the lack of contraction, we allow constraints only to identify bridge variables with
constants, not with other variables.
\begin{mathpar}
  \inferrule[$\BFI$-Constraint]
  {\wfbdim[\GG]{\bmr} \\ \Ge \in \{0,1\}}
  {\wfcst[\GG]{\bmr = \bm\Ge}}
\end{mathpar}

We note that affine variables are also central to nominal sets~\cite{pitts13}, where they are used to
represent variable names in syntax. The BCH model of univalent type theory in cubical sets~\cite{bch,bezem19}
is also based on an affine interval (and has been presented in a nominal style by Pitts~\cite{pitts14}). We
say more about the BCH model in \cref{sec:parametric:bch}.

\subsection[Bridge-types]{\texorpdfstring{$\tybridge$}{Bridge}-types}\label{sec:parametric:bridge-types}

\begin{figure}
  \begin{mathpar}
    \inferrule[Bridge-Form]
    {\wftype[\ctxbdim{\GG}[x]]{A} \\
      \wftm[\GG]{M_0}{\usubstdim{A}{\bm0}{\bmx}} \\
      \wftm[\GG]{M_1}{\usubstdim{A}{\bm1}{\bmx}}}
    {\wftype[\GG]{\tybridge{\bmx.A}{M_0}{M_1}}}
    \and
    \inferrule[Bridge-Intro]
    {\wftm[\ctxbdim{\GG}[x]]{M}{A}}
    {\wftm[\GG]{\tmblam[x]{M}}{\tybridge{\bmx.A}{\usubstdim{M}{\bm0}{\bmx}}{\usubstdim{M}{\bm1}{\bmx}}}}
    \and
    \inferrule[Bridge-Elim]
    {\wfbdim[\GG]{\bmr} \\
      \wftm[\ctxres{\GG}{\bmr}]{P}{\tybridge{\bmx.A}{M_0}{M_1}}}
    {\wftm[\GG]{\tmbapp{P}{\bmr}}{\usubstdim{A}{\bmr}{\bmx}}}
    \and
    \inferrule[Bridge-$\beta$]
    {\wfbdim[\GG]{\bmr} \\
      \wftm[\ctxbdim{\ctxres{\GG}{\bmr}}[x]]{M}{A}}
    {\eqtm[\GG]{\tmbapp{(\tmblam[x]{M})}{\bmr}}{\usubstdim{M}{\bmr}{\bmx}}{\usubstdim{A}{\bmr}{\bmx}}}
    \and
    \inferrule[Bridge-$\partial$]
    {\wftm[\GG]{P}{\tybridge{\bmx.A}{M_0}{M_1}} \\ \Ge \in \{0,1\}}
    {\eqtm[\GG]{\tmbapp{P}{\bm\Ge}}{M_\Ge}{\usubstdim{A}{\bm\Ge}{\bmx}}}
    \and
    \inferrule[Bridge-$\eta$]
    {\wftm[\GG]{P}{\tybridge{\bmx.A}{M_0}{M_1}}}
    {\eqtm[\GG]{P}{\tmblam[x]{\tmbapp{P}{\bmx}}}{\tybridge{\bmx.A}{M_0}{M_1}}}
  \end{mathpar}
  \caption{Rules for $\tybridge$-types}%
  \label{fig:bridge-types}
\end{figure}

We define $\tybridge$-types exactly as we define $\typath$-types: elements of $\tybridge{\bmx.A}{M_0}{M_1}$
are elements of $A$ in an abstracted bridge variable $\bmx$ that agree with $M_0$ and $M_1$ on their
endpoints. We give rules for $\tybridge$-types in \cref{fig:bridge-types}. The only difference is that a
bridge can only be applied to a fresh variable, in keeping with the judgmental structure: $\tmbapp{P}{\bmr}$
makes sense when $\bmr$ is apart from $P$.

\subsection[The extent operator]{The \texorpdfstring{$\tmextent$}{extent} operator}\label{sec:parametric:extent}

\begin{figure}
  \centering
  \begin{mathpar}
    \mprset{vskip=0.1em}
    \inferrule[Extent]
    {\wfbdim[\GG]{\bmr} \\
      \wftype[\ctxbdim{\ctxres{\GG}{\bmr}}[x]]{A} \\
      \wftype[\ctxsnoc{\ctxbdim{\ctxres{\GG}{\bmr}}[x]}[a]{A}]{B} \\
      \wftm[\GG]{M}{\usubstdim{A}{\bmr}{\bmx}} \\
      \wftm[\ctxsnoc{\ctxres{\GG}{\bmr}}[a_0]{\usubstdim{A}{\bm0}{\bmx}}]{N_0}{\usubst{\usubstdim{B}{\bm0}{\bmx}}{a_0}{a}} \\
      \wftm[\ctxsnoc{\ctxres{\GG}{\bmr}}[a_1]{\usubstdim{A}{\bm1}{\bmx}}]{N_1}{\usubst{\usubstdim{B}{\bm1}{\bmx}}{a_1}{a}} \\
      \wftm[\ctxsnoc{\ctxsnoc{\ctxsnoc{\ctxres{\GG}{\bmr}}[a_0]{\usubstdim{A}{\bm0}{\bmx}}}[a_1]{\usubstdim{A}{\bm1}{\bmx}}}[\overline{a}]{\tybridge{\bmx.A}{a_0}{a_1}}]{\overline{N}}{\tybridge{\bmx.\usubst{B}{\tmbapp{\overline{a}}{\bmx}}{a}}{N_0}{N_1}}
    }
    {\wftm[\GG]{\tmextent{\bmr}{M}{a_0.N_0}{a_1.N_1}{a_0.a_1.\overline{a}.\overline{N}}}{\usubst{\usubstdim{B}{\bmr}{\bmx}}{M}{a}}}
    \and
    \inferrule[Extent-$\partial$]
    {\cdots \\
      \Ge \in \{0,1\} \\
      \wftm[\GG]{M}{\usubstdim{A}{\bm\Ge}{\bmx}}}
    {\eqtm[\GG]{\tmextent{\bm\Ge}{M}}{\usubst{N_\Ge}{M}{a_\Ge}}{\usubst{\usubstdim{B}{\bm\Ge}{\bmx}}{M}{a}}}
    \and
    \inferrule[Extent-$\beta$]
    {\cdots \\
      \wftm[\ctxbdim{\ctxres{\GG}{\bmr}}[x]]{M}{A}}
    {\eqtm[\GG]{\tmextent{\bmr}{\usubstdim{M}{\bmr}{\bmx}}}{\tmbapp{\usubst{\usubst{\usubst{\overline{N}}{\usubstdim{M}{\bm0}{\bmx}}{a_0}}{\usubstdim{M}{\bm1}{\bmx}}{a_1}}{\tmblam[x]{M}}{\overline{a}}}{\bmr}}{\usubst{\usubstdim{B}{\bmr}{\bmx}}{M}{a}}}
  \end{mathpar}
  \caption{Rules for the $\tmextent$ operator. The elided premises in the second and third rules match those
    of the first rule.}%
  \label{fig:extent}
\end{figure}

As we have mentioned, the first reason for using affine variables is connected to function extensionality. If
we follow the standard relational model of type theory---more generally, the standard definition of a logical
relation at function type---we expect the following isomorphism, a bridge equivalent of Lemmas~\ref{lem:funapp-over} and~\ref{lem:funext-over}.
\begin{gather*}
  \tybridge{\bmx.\typi[a]{A}{B}}{F_0}{F_1} \\
  \simeq \\
  \typi*[a_0]{\usubstdim{A}{\bm0}{\bmx}}{\typi*[a_1]{\usubstdim{A}{\bm1}{\bmx}}{\typi[p]{\tybridge{\bmx.A}{a_0}{a_1}}{\tybridge{\bmx.\usubst{B}{\tmbapp{p}{\bmx}}{a}}{F_0a_0}{F_1a_1}}}}
\end{gather*}
To go from bottom to top, we can repeat the proof of Lemma~\ref{lem:funapp-over} without issue. On the other
hand, the proof of Lemma~\ref{lem:funext-over} relies on the presence of $\coe$, which has no equivalent in
parametric type theory. Instead, we will introduce a new operator to validate this principle, $\tmextent$,
which relies on the substructurality of the bridge interval.

Rules for $\tmextent$ are displayed in \cref{fig:extent}. The operator is essentially a fully applied version
of the principle we are looking for.

\begin{lem}\label{lem:bridge-funext}
  Let $\wftype[\ctxpdim{}[x]]{A}$, $\wftype[\ctxsnoc{\ctxpdim{}[x]}[a]{A}]{B}$,
  $\wftm{F_0}{\usubstdim{(\typi[a]{A}{B})}{\bm0}{\bmx}}$, and
  $\wftm{F_1}{\usubstdim{(\typi[a]{A}{B})}{\bm1}{\bmx}}$ be given. Then we have the following.
  \begin{mathpar}
    \inferrule
    {\wftm{H}{\typi*[a_0]{\usubstdim{A}{\bm0}{\bmx}}{\typi*[a_1]{\usubstdim{A}{\bm1}{\bmx}}{\typi[p]{\tybridge{\bmx.A}{a_0}{a_1}}{\tybridge{\bmx.\usubst{B}{\tmbapp{p}{\bmx}}{a}}{F_0a_0}{F_1a_1}}}}}}
    {\wftm{\tmbridgefunext{H}}{\tybridge{\bmx.\typi[a]{A}{B}}{F_0}{F_1}}}
  \end{mathpar}
\end{lem}
\begin{proof}
  $\tmbridgefunext{H} \eqdef \tmblam[x]{\tmlam[a]{\tmextent{\bmx}{a}{a_0.F_0a_0}{a_1.F_1a_1}{a_0.a_1.\overline{a}.Ha_0a_1\overline{a}}}}$.
\end{proof}

As shown in the rule \rulename{Extent-$\beta$}, $\tmextent{\bmr}$ evaluates by \emph{capturing} the
occurrences of $\bmr$ in its principal argument $M$. That is,
$\tmextent{\bmx}{M}{a_0.F_0a_0}{a_1.F_1a_1}{a_0.a_1.\overline{a}.Ha_0a_1\overline{a}}$ evaluates by passing
$\usubstdim{M}{\bm0}{\bmx}$, $\usubstdim{M}{\bm1}{\bmx}$, and $\tmblam[x]{M}$ to $H$. That this is
possible depends on affinity because $\tmblam[x]{-}$ does not necessarily commute with diagonal
substitutions. Specifically, if we have some term $M(\bmx,\bmy)$ that depends on two variables, we can get
different results by abstracting before or after substitution as follows.
\[
  \begin{tikzcd}
    M(\bmx,\bmy) \ar[Mapsto]{d}[left]{\tmblam[x]{-}} \ar[Mapsto]{rr}{\usubstdim{}{\bmy}{\bmx}} & & M(\bmy,\bmy) \ar[Mapsto]{d}{\tmblam[x]{-}} \\
    \tmblam[x]{M(\bmx,\bmy)} \ar[Mapsto]{r}[below]{\usubstdim{}{\bmy}{\bmx}} & \tmblam[x]{M(\bmx,\bmy)} \ar[phantom]{r}{\neq} & \tmblam[x]{M(\bmy,\bmy)}
  \end{tikzcd}
\]
We call the operator $\tmextent$ because $\tmextent{\bmr}{M}$ reveals the extent of the term $M$ in the direction
$\bmr$: either $\bmr$ is a constant, in which case $M$ is simply a point, or $\bmr$ is a variable
$\bmx$, in which case $M$ is a point on a line $\tmblam[x]{M}$ in that direction.

The conditions under which \rulename{Extent-$\beta$} applies are somewhat subtle. In short, the requirement is
that $M$ not depend on any term variables that are not apart from $\bmx$. For example, $\tmextent{\bmx}{a}$ can
be reduced only when $a$ appears prior to $\bmx$ in the context. Once again, this relates to the commutativity
of substitutions and capture, in this case the difference between $\usubst{(\tmblam[x]{a})}{Q(\bmx)}{a}$ and
$\tmblam[x]{(\usubst{a}{Q(\bmx)}{a})}$. Note, however, that an $\tmextent$ term containing no term variables
always reduces, so this issue is invisible to the closed operational semantics; it is merely a matter of the
degree to which we can extend the closed reduction rule to an equality for open terms.

We can show that $\tmbridgefunext$ is in fact an isomorphism, with inverse given by the bridge equivalent of
Lemma~\ref{lem:funapp-over}. One inverse condition is \rulename{Extent-$\beta$}, while the other is an
``$\eta$-principle for $\tmextent$'' that can be proven up to path equality using $\tmextent$ itself, much as
dependent elimination for inductive types gives such weak $\eta$-principles.

\begin{prop}\label{prop:bridge-funext-iso}
  Let $\wftype[\ctxpdim{}[x]]{A}$, $\wftype[\ctxsnoc{\ctxpdim{}[x]}[a]{A}]{B}$,
  $\wftm{F_0}{\usubstdim{(\typi[a]{A}{B})}{\bm0}{\bmx}}$, and
  $\wftm{F_1}{\usubstdim{(\typi[a]{A}{B})}{\bm1}{\bmx}}$ be given. Then we have the following.
  \begin{gather*}
    \tybridge{\bmx.\typi[a]{A}{B}}{F_0}{F_1} \\
    \simeq \\
    \typi*[a_0]{\usubstdim{A}{\bm0}{\bmx}}{\typi*[a_1]{\usubstdim{A}{\bm1}{\bmx}}{\typi[p]{\tybridge{\bmx.A}{a_0}{a_1}}{\tybridge{\bmx.\usubst{B}{\tmbapp{p}{\bmx}}{a}}{F_0a_0}{F_1a_1}}}}
  \end{gather*}
\end{prop}

We can also show that the function extensionality principle induces a corresponding principle for bridges in
isomorphism types. We leave the proof to the reader; one can prove it using $\tmextent$ directly, but it also
follows formally from Proposition~\ref{prop:bridge-funext-iso} and the correspondence between bridges over
path types and paths over bridge types.

\begin{prop}\label{prop:bridge-isoext}
  Let $\wftype[\ctxpdim{}[x]]{A,B}$, $\wftm{I_0}{\usubstdim{(A \simeq B)}{\bm0}{\bmx}}$, and
  $\wftm{I_1}{\usubstdim{(A \simeq B)}{\bm1}{\bmx}}$ be given. Then we have the following.
  \begin{mathpar}
    \inferrule
    {\wftm{H}{\typi*[a_0]{\usubstdim{A}{\bm0}{\bmx}}{\typi[a_1]{\usubstdim{A}{\bm1}{\bmx}}{\tybridge{\bmx.A}{a_0}{a_1} \simeq \tybridge{\bmx.B}{\tmfst{I_0}(a_0)}{\tmfst{I_1}(a_1)}}}}}
    {\wftm{\tmbridgeisoext{H}}{\tybridge{\bmx.A \simeq B}{I_0}{I_1}}}
  \end{mathpar}
\end{prop}

\subsection[Gel-types and relativity]{\texorpdfstring{$\tygel$}{Gel}-types and relativity}\label{sec:parametric:gel}

Finally, we come to the equivalent of univalence in parametric type theory, which we call \emph{relativity}:
the correspondence between bridges of types and relations. One direction of the correspondence is given by
$\tybridge$-types: given a bridge of types $\wftype[\ctxbdim{}[x]]{A}$, we have a relation
$\tybridge{\bmx.A}{-}{-}$ on $\usubstdim{A}{\bm0}{\bmx}$ and $\usubstdim{A}{\bm1}{\bmx}$ (which we
henceforth simply write as $\tybridge{\bmx.A}$). As with $\tyv$-types for univalence, the inverse will be
effected by introducing a new type constructor, which we call the \emph{$\tygel$-type}. These resemble the
$\tyg$-types of the BCH model, but apply to r\textsf{el}ations rather than isomorphisms, hence the name.

\begin{figure}
  \centering
    \begin{mathpar}
      \inferrule[Gel-Form]
      {\wfbdim[\GG]{\bmr} \\
        \wftype[\ctxres{\GG}{\bmr}]{A_0} \\
        \wftype[\ctxres{\GG}{\bmr}]{A_1} \\
        \wftype[\ctxres{\GG}{\bmr}, a_0 : A_0, a_1 : A_1]{R}
      }
      {\wftype[\GG]{\tygel{\bmr}{A_0}{A_1}{a_0.a_1.R}}}
      \and
      \inferrule[Gel-Intro]
      {\wftm[\ctxres{\GG}{\bmr}]{M_0}{A_0} \\
        \wftm[\ctxres{\GG}{\bmr}]{M_1}{A_1} \\
        \wftm[\ctxres{\GG}{\bmr}]{P}{\usubst{R}{M_0,M_1}{a_0,a_1}}}
      {\wftm[\GG]{\tmgel{\bmr}{M_0}{M_1}{P}}{\tygel{\bmr}{A_0}{A_1}{a_0.a_1.R}}}
      \and
      \inferrule[Gel-Form-$\partial$]
      {\Ge \in \{0,1\} \\
        \wftype[\GG]{A_\Ge}}
      {\eqtype[\GG]{\tygel{\bm\Ge}{A_0}{A_1}{a_0.a_1.R}}{A_\Ge}}
      \and
      \inferrule[Gel-Intro-$\partial$]
      {\Ge \in \{0,1\} \\
        \wftm[\GG]{M_\Ge}{A_\Ge}}
      {\eqtm[\GG]{\tmgel{\bm\Ge}{M_0}{M_1}{P}}{M_\Ge}{A_\Ge}}
      \and
      \inferrule[Gel-Elim]
      {\wftm[\ctxbdim{\GG}[x]]{Q}{\tygel{\bmx}{A_0}{A_1}{R}}}
      {\wftm[\GG]{\tmungel{\bmx.Q}}{\usubst{R}{\usubstdim{Q}{\bm0}{\bmx},\usubstdim{Q}{\bm1}{\bmx}}{a_0,a_1}}}
      \and
      \inferrule[Gel-$\beta$]
      {\wftm[\GG]{P}{\usubst{R}{M_0,M_1}{a_0,a_1}}}
      {\eqtm[\GG]{\tmungel{\bmx.\tmgel{\bmx}{M_0}{M_1}{P}}}{P}{\usubst{R}{M_0,M_1}{a_0,a_1}}}
      \and
      \inferrule[Gel-$\eta$]
      {\wfbdim[\GG]{\bmr} \\
        \wftype[\ctxres{\GG}{\bmr}]{A_0} \\
        \wftype[\ctxres{\GG}{\bmr}]{A_1} \\
        \wftype[\ctxres{\GG}{\bmr}, a_0 : A_0, a_1 : A_1]{R} \\
        \wftm[\ctxbdim{\ctxres{\GG}{\bmr}}[x]]{Q}{\tygel{\bmx}{A_0}{A_1}{a_0.a_1.R}}}
      {\eqtm[\GG]{\usubstdim{Q}{\bmr}{\bmx}}{\tmgel{\bmr}{\usubstdim{Q}{\bm0}{\bmx}}{\usubstdim{Q}{\bm1}{\bmx}}{\tmungel{\bmx.Q}}}{\tygel{\bmr}{A_0}{A_1}{a_0.a_1.R}}}
    \end{mathpar}
    \caption{Rules for $\tygel$-types.}%
  \label{fig:gel}
\end{figure}

We provide rules for $\tygel$-types in \cref{fig:gel}. Unlike the $\tyv$-type, the $\tygel$-type directly
converts relations to bridges of types: for any relation $\wftm[a_0 : A_0, a_1 : A_1]{R}{\tyuniv}$, we have
$\wftm{\tmblam[x]{\tygel{\bmx}{A_0}{A_1}{a_0.a_1.R}}}{\tybridge{\tyuniv}{A_0}{A_1}}$. The introduction rule turns
a witness for the relation $\wftm[\GG]{P}{\usubst{R}{M_0,M_1}{a_0,a_1}}$ into a bridge
$\wftm{\tmblam[x]{\tmgel{\bmx}{M_0}{M_1}{P}}}{\tybridge{\bmx.\tygel{\bmx}{A_0}{A_1}{a_0.a_1.R}}{M_0}{M_1}}$ over
the corresponding $\tygel$-type, while the elimination rule conversely turns such a bridge into a witness.
When we have a relation in the form $\wftm{R}{A_0 \times A_1 \to \tyuniv}$, we will abbreviate
$\tygel{\bmr}{A_0}{A_1}{a_0.a_1.R\tmpair{a_0}{a_1}}$ as $\tygel{\bmr}{A_0}{A_1}{R}$.

The problem of shifting dimensions in $\tyv$-types, described in \cref{sec:cubical:v}, is no longer an issue
when we have affine interval variables; we can express degeneracy in $\bmr$ using the context restriction
$\ctxres{-}{\bmr}$. This is fortunate, as the trick for deriving univalence from $\tyv$-types would not apply
here. For univalence, we rely on the fact that the constant path $\tmplam[\_]{B}$ corresponds to the identity
isomorphism on $B$; thus we can transform isomorphisms $A \simeq B$ into paths by composing with
$\tmplam[\_]{B}$ in a $\tyv$-type. On the other hand, the constant bridge $\tmblam[\_]{A}$ does \emph{not}
necessarily correspond to the identity relation (\ie, the path relation $\typath{B}$); rather, it corresponds
to the bridge relation $\tybridge{B}$. In particular, $\tmblam[\_]{\tyuniv}$ will correspond to
$\tmlam[\tmpair{A}{B}]{(A \times B \to \tyuniv)}$, not $\tmlam[\tmpair{A}{B}]{(A \simeq B)}$. Thus, a
$\tyv$-like type would only give us bridges for those relations that factor through the bridge relation on one
endpoint---more generally, through some bridge $\bmx.B$ we already have in hand.

We only mean in the above to give some intuition for the difference between the affine and structural
situation, not for example to prove beyond a shadow of a doubt that no $\tygel$-like type can exist
structurally. However, we note that in the bisimplicial set semantics of Riehl and Shulman's directed type
theory~\cite{riehl17}, a similar setting, an issue of dimension shift does indeed prevent the existence of a
universe where arrows correspond to relations~\cite{riehl18}.

We now proceed to prove the relativity principle.

\begin{thm}\label{thm:relativity}
  For any $\wftm{A_0,A_1}{\tyuniv}$,
  $\wftm{\tmlam[C]{\tybridge{\bmx.\tmbapp{C}{\bmx}}}}{\tybridge{\tyuniv}{A_0}{A_1} \to (A_0 \times A_1 \to \tyuniv)}$ is an
  isomorphism.
\end{thm}
\begin{proof}
  As candidate inverse, we of course take $\tmlam[R]{\tmblam[x]\tygel{\bmx}{A_0}{A_1}{R}}$.

  First we show that this is a left inverse, \ie, that the following holds.
  \[
    \typi[R]{A_0 \times A_1 \to \tyuniv}{\typath{A_0 \times A_1 \to \tyuniv}{\tybridge{\bmx.\tygel{\bmx}{A_0}{A_1}{R}}}{R}}
  \]
  Let $R : A_0 \times A_1 \to \tyuniv$ be given. We need to construct a path in $A_0 \times A_1 \to \tyuniv$,
  so we apply function extensionality and univalence. Then for every $a_0 : A_0$ and $a_1 : A$, we need an
  isomorphism $\tybridge{\bmx.\tygel{\bmx}{A_0}{A_1}{R}}{a_0}{a_1} \simeq R\tmpair{a_0}{a_1}$. This isomorphism
  is implemented exactly by the introduction and elimination forms of the $\tygel$-type, and the inverse
  conditions hold (up to exact equality) by \rulename{Gel-$\beta$} and \rulename{Gel-$\eta$}.

  Now we show it is also a right inverse.
  \[
    \typi[C]{\tybridge{\tyuniv}{A_0}{A_1}}{\typath{\tybridge{\tyuniv}{A_0}{A_1}}{\tmblam[x]{\tygel{\bmx}{A_0}{A_1}{\tybridge{\bmx.\tmbapp{C}{\bmx}}}}}{C}}
  \]
  Let $C : \tybridge{\tyuniv}{A_0}{A_1}$ be given. We are asked to provide a square with the following
  boundary.
  \[
    \begin{tikzpicture}
      \def\height{2}
      \def\width{6}
      \def\awidth{1.5}
      \def\aheight{0.5}
      \def\alength{0.5}
      \def\pad{1}

      \draw (-\awidth, \height+\aheight) [->,thick] to node [above] {\small $\bmx$} (-\awidth+\alength, \height+\aheight) ;
      \draw (-\awidth, \height+\aheight) [->] to node [left] {\small $y$} (-\awidth, \height+\aheight-\alength) ;
      \node (tl) at (0, \height) {$A_0$} ;
      \node (tr) at (\width, \height) {$A_1$} ;
      \node (bl) at (0, 0) {$A_0$} ;
      \node (br) at (\width, 0) {$A_1$} ;
      \draw (tl) [->,thick] to node [above] {$\tygel{\bmx}{A_0}{A_1}{\tybridge{\bmx.\tmbapp{C}{\bmx}}}$} (tr) ;
      \draw (bl) [->,thick] to node [below] {$\tmbapp{C}{\bmx}$} (br) ;
      \draw (tl) [->] to node [left] {$A_0$} (bl) ;
      \draw (tr) [->] to node [right] {$A_1$} (br) ;
      \node at (-\awidth-\pad, 0) {} ;
      \node at (\width+\awidth+\pad, 0) {} ;
    \end{tikzpicture}
  \]
  By ``flipping'' this square---\ie, using the correspondence between bridges of paths and paths of bridges
  given by exchange of variables---it suffices to show the following.
  \[
    \tybridge{\bmx.\typath{\tyuniv}{\tygel{\bmx}{A_0}{A_1}{\tybridge{\bmx.\tmbapp{C}{\bmx}}}}{\tmbapp{C}{\bmx}}}{\tmplam[\_]{A_0}}{\tmplam[\_]{A_1}}
  \]
  Now we apply univalence, converting the path type in the universe to a type of isomorphisms. Here we use the
  fact that the constant paths $\tmplam[\_]{A_\Ge}$ correspond to identity isomorphisms $\tmidiso{A_\Ge}$
  across univalence. This reduces our goal to the following.
  \[
    \tybridge{\bmx.\tygel{\bmx}{A_0}{A_1}{\tybridge{\bmx.\tmbapp{C}{\bmx}}} \simeq \tmbapp{C}{\bmx}}{\tmidiso{A_0}}{\tmidiso{A_1}}
  \]
  Finally we apply Proposition~\ref{prop:bridge-isoext}, reducing the goal once more.
  \[
    \typi*[a_0]{A_0}{\typi[a_1]{A_1}{\tybridge{\tygel{\bmx}{A_0}{A_1}{\tybridge{\bmx.\tmbapp{C}{\bmx}}}}{a_0}{a_1} \simeq \tybridge{\bmx.\tmbapp{C}{\bmx}}{a_0}{a_1}}}
  \]
  This is a consequence of the left inverse condition we have already proven.
\end{proof}

Note that the proof of relativity relies on univalence; not surprising, since it is an isomorphism between
types that involve the universe. (It also relies directly on function extensionality, both for paths and
bridges.) In~\cite{bernardy15}, which does not include univalence, relativity is instead ensured by imposing
stronger equations on $\tygel$-types---precisely the equations $\tybridge{\bmx.\tygel{\bmx}{A_0}{A_1}{R}} = R$
and $C = \tmblam[x]{\tygel{\bmx}{A_0}{A_1}{\tybridge{\bmx.\tmbapp{C}{\bmx}}}}$ required for the proof. (These
equations are there named \rulename{Pair-Pred} and \rulename{Surj-Typ}.) These equations make it more
difficult to construct a presheaf model, as we discuss further in \cref{sec:presheaf}.

\subsection{Using affine variables for paths}\label{sec:parametric:bch}

Before we dive into using parametric cubical type theory, let us take one more moment to reflect on structural
and substructural interval variables. We have seen why affinity is important for parametric type theory, but
is structurality important for cubical type theory? The Bezem-Coquand-Huber model gives a partial negative
answer: there is a model of univalent type theory in presheaves on the affine cube category~\cite{bch,bezem19}. While no one has attempted to design a type theory based on this model, it is plausible
that it could be done.

Unfortunately, affine interval variables create problems for modeling higher inductive types. Consider, for
example, the following extremely simple type, which has a single path constructor with no fixed
boundary.
\[
  \begin{array}{l}
    \dataheading{\tyi} \\
    \niceconstr{\tmiin}[\ofp{x}]{\tyi} \\
  \end{array}
\]
This specification generates the following elimination principle and computation rule, which essentially says
that maps out of $\tyi$ correspond to terms in a context extended with an interval variable.
\begin{mathpar}
  \inferrule
  {\wftype[\ctxsnoc{\GG}[a]{\tyi}]{C} \\
    \wftm[\GG]{M}{\tyi} \\
    \wftm[\ctxpdim{\GG}[x]]{Q_{\tmiin}}{\usubst{C}{\tmiin{x}}{a}}}
  {\wftm[\GG]{\tmielim{a.C}{M}{x.Q_{\tmiin}}}{\usubst{C}{M}{a}}}
  \and
  \inferrule
  {\wftype[\ctxsnoc{\GG}[a]{\tyi}]{C} \\
    \wfpdim[\GG]{r} \\
    \wftm[\ctxpdim{\GG}[x]]{Q_{\tmiin}}{\usubst{C}{\tmiin{x}}{a}}}
  {\eqtm[\GG]{\tmielim{a.C}{\tmiin{r}}{x.Q_{\tmiin}}}{\usubstdim{Q_{\tmiin}}{r}{x}}{\usubst{C}{\tmiin{r}}{a}}}
\end{mathpar}
The issue is in the computation rule, which applies the interval substitution $\usubstdim{-}{r}{x}$ to
$Q_{\tmiin}$. If our interval is affine, then this substitution will be nonsensical if $Q_{\tmiin}$ already
mentions $r$. Moreover, it is not clear how to restrict the premises of $\tmielim$ to ensure the substitution
is sensible without ending up with an insufficiently powerful principle. On a more conceptual level, the
$\tyi$ type is suspicious in an affine system in that \emph{structural} maps out of $\tyi$ correspond to
\emph{affine} maps out of the interval.

The problem of higher inductive types is one reason why research in cubical type theory and models has shifted
from substructural to structural interval variables. There is also the fact that structural variables are
simply easier to work with. Still, the BCH model does have some intriguing advantages; for one, univalence can
be implemented in $\tygel$-like rather $\tyv$-like fashion, and the former admits simpler implementations of
coercion and composition.

\section{Applying internal parametricity}\label{sec:practice}

Now that we have laid out what we need of parametric cubical type theory, we can get started proving
theorems. We will begin with a classic application of parametricity: relating inductive types to their Church
encodings, in this case booleans.

\subsection{Booleans}\label{sec:practice:bool}

The \emph{Church booleans} are the polymorphic binary operators, the elements of the type
$\tychurchbool \eqdef \typi[A]{\tyuniv}{\tyarr{A}{\tyarr{A}{A}}}$. Clearly this type has at least two
elements, $\tmlam[A]{\tmlam[t]{\tmlam[\_]{t}}}$ and $\tmlam[A]{\tmlam[\_]{\tmlam[f]{f}}}$. It is a classical
consequence of parametricity that these are the \emph{only} two elements of $\tychurchbool$. Using internal
parametricity, we can prove that $\tychurchbool$ is indeed isomorphic to the standard type of booleans
($\tybool$).

\begin{thm}
  $\tybool \simeq \tychurchbool$.
\end{thm}
\begin{proof}
  It is easy to define functions $\wftm{F}{\tybool \to \tychurchbool}$ and
  $\wftm{G}{\tychurchbool \to \tybool}$ in either direction.
  \[
    F \eqdef \tmlam[b]{\tmlam[A]{\tmlam[t]{\tmlam[f]{\tmif{\_.A}{b}{t}{f}}}}}
    \qquad
    G \eqdef \tmlam[k]{k(\tybool)(\tmtrue)(\tmfalse)}
  \]
  Moreover, it is easy to check by case-analysis that $G(Fb)$ is path-equal to $b$ for any $b : \tybool$.

  We use parametricity to prove the other inverse condition. Let some $\ctxsnoc{}[k]{\tychurchbool}$ along
  with $\ctxsnoc{\ctxsnoc{\ctxsnoc{}[A]{\tyuniv}}[t]{A}}[f]{A}$ be given. We intend to show that $F(Gk)Atf$ is
  path-equal to $kAtf$. We define a relation $\wftm{R}{\tyarr{\typrod{\tybool}{A}}{\tyuniv}}$ as follows.
  \[
    R\tmpair{b}{a} \eqdef \typath{A}{FbAtf}{a}
  \]
  That is, $R$ is the graph of $\tmlam[b]{FbAtf}$. Abstracting a bridge interval variable $\bmx$, we can
  apply $k$ at the $\tygel$-type corresponding to $R$.
  \[
    \wftm{k(\tygel{\bmx}{\tybool}{A}{R})}{\tygel{\bmx}{\tybool}{A}{R} \to \tygel{\bmx}{\tybool}{A}{R} \to \tygel{\bmx}{\tybool}{A}{R}}
  \]
  We see that $\tmtrue$ and $t$ are related by $R$: we have
  $\wftm{\tmplam[\_]{t}}{R\tmpair{\tmtrue}{t}}$. Likewise, we have
  $\wftm{\tmplam[\_]{f}}{R\tmpair{\tmfalse}{f}}$. We apply $k$ at the two $\tmgel$ terms corresponding to
  these witnesses of the relation.
  \[
    \wftm{k(\tygel{\bmx}{\tybool}{A}{R})(\tmgel{\bmx}{\tmtrue}{t}{\tmplam[\_]{t}})(\tmgel{\bmx}{\tmfalse}{f}{\tmplam[\_]{f}})}{\tygel{\bmx}{\tybool}{A}{R}}
  \]
  If we substitute $\bm0$ for $\bmx$, each $\tygel$ and $\tmgel$ term reduces to its first term argument,
  leaving $k(\tybool)(\tmtrue)(\tmfalse)$, which is $Gk$. Likewise, if we substitute $\bm1$, we get
  $kAtf$. When we bind $\bmx$ and project the relation witness from this term, we therefore wind up with the
  following.
  \[
    \wftm{\tmungel{\bmx.k(\tygel{\bmx}{\tybool}{A}{R})(\tmgel{\bmx}{\tmtrue}{t}{\tmplam[\_]{t}})(\tmgel{\bmx}{\tmfalse}{f}{\tmplam[\_]{f}})}}{R\tmpair{Gk}{kAtf}}
  \]
  By definition of $R$, this is exactly our goal: a path from $F(Gk)Atf$ to $kAtf$. By function
  extensionality, we get a term in $\typath{\tychurchbool}{F(Gk)}{k}$.
\end{proof}

This argument follows the shape of a classical parametricity proof: we define a relation, apply a function to
related arguments (here represented by $\tmgel$ terms), and conclude that the outputs are also related (via
$\tmungel$). We can apply similar arguments to characterize other Church encodings. For example, we can show
that the type $\typi[A]{\tyuniv}{A \to (A \to A) \to A}$ is isomorphic to the natural numbers; in that case,
we would also use $\tmextent$ to construct a bridge in the function type.

Note that because the system is predicative, it does not appear possible to simply \emph{define} inductive
types using Church encodings. In the absence of a primitive boolean type in $\tyuniv$, $\tychurchbool$ can
only eliminate into small types (that is, types in the universe $\tyuniv$). When there \emph{is} a primitive
boolean type, however, $\tychurchbool$ inherits its properties: we can define functions from $\tychurchbool$
into large type by induction by factoring through the map $\tychurchbool \to \tybool$.

The picture gets more complex when we consider Church encodings that are parameterized over ``external''
types, such as the following encoding of the coproduct.
\[
  A + B \mathrel{\overset{?}{\simeq}} \typi[C]{\tyuniv}{(A \to C) \to (B \to C) \to C}
\]
A classical proof would rely on the \emph{identity extension lemma}~\cite{reynolds83}, which implies in
particular that the relational interpretation of a closed type ($A$ or $B$ here) is the identity
relation. This is not the case in BCM-style internal parametricity. In particular, the principle fails for the
universe: the types $\tybridge{\tyuniv}{A}{B}$ and $\typath{\tyuniv}{A}{B}$ are not the same, as one is
isomorphic to $A \times B \to \tyuniv$ and the other is isomorphic to $A \simeq B$.

If we focus our attention on small types, we will see that any concrete type $A$ we can think of will satisfy
$\tybridge{A}{a}{b} \simeq \typath{A}{a}{b}$ for all $a,b : A$; however, there is no way to prove for an
arbitrary $A$. We say that types that do satisfy this principle are \emph{bridge-discrete}. We can show that
the universe of bridge-discreteness types is well-behaved and closed under most type formers.

\subsection{Bridge-discrete types}\label{sec:practice:discrete}

In any type, we have a canonical map from paths to bridges induced by coercion. A type is
\emph{bridge-discrete} when this map is an isomorphism.

\begin{defi}
  For $\wftype{A}$ and $\wftm{M,N}{A}$, define $\wftm{\tmloosen{A}}{\typath{A}{M}{N} \to \tybridge{A}{M}{N}}$
  by $\tmloosen{A} \eqdef \tmlam[p]{\coe{x.\tybridge{A}{\tmpapp{p}{0}}{\tmpapp{p}{x}}}{0}{1}{\tmblam[\_]{\tmpapp{p}{0}}}}$.
\end{defi}

\begin{rem}\label{rem:loosen-refl}
  For any $\wftm{M}{A}$, $\tmloosen{A}$ takes the reflexive path on $M$ to the reflexive bridge on $A$:
  we have $\wftm{\tmplam[y]{\coe{x.\tybridge{A}{M}{M}}{y}{1}{\tmblam[\_]{M}}}}{\typath{\tybridge{A}{M}{M}}{\tmloosen{A}{\tmplam[\_]{M}}}{\tmblam{\_}{M}}}$.
\end{rem}

\begin{defi}
  Given $\wftype{A}$, define $\wftype{\tyisbdisc{A}}$ as follows.
  \[
    \tyisbdisc{A} \eqdef \typi*[a]{A}{\typi[b]{A}{\tyisiso{\typath{A}{a}{b}}{\tybridge{A}{a}{b}}{\tmloosen{A}}}}
  \]
\end{defi}

As we mentioned in \cref{sec:parametric:gel}, the type $\tyisiso$ is always a proposition~\cite[Theorem
4.3.2]{hott-book}; any two proofs of $\tyisiso$ are connected by a path. A function type with propositional
codomain is again a proposition~\cite[Example 2.6.2]{hott-book}, so $\tyisbdisc{A}$ is a proposition. We
define the \emph{universe of bridge-discrete types} as $\tybdisc \eqdef \tysigma[A]{\tyuniv}{\tyisbdisc{A}}$.

Before continuing, we recall some standard results from univalent type theory. The proofs we reference are
conducted using Martin-L\"{o}f identity types, but can be readily adapted to cubical path types by way of
Lemma~\ref{lem:j}.

\begin{prop}\label{prop:path-retract}
  Let $\wftype{A}$ and let $\wftype[\ctxsnoc{\ctxsnoc{}[a]{A}}[b]{A}]{R}$ be a relation on $A$. Suppose we
  have a family of maps with right inverses:
  \begin{itemize}[label=$\triangleright$]
  \item $\wftm{F}{\typi*[a]{A}{\typi[b]{A}{\tyarr{R\tmpair{a}{b}}{\typath{A}{a}{b}}}}}$,
  \item $\wftm{G}{\typi*[a]{A}{\typi[b]{A}{\tyrinv{R\tmpair{a}{b}}{\typath{A}{a}{b}}{Fab}}}}$.
  \end{itemize}
  The $Fab$ is an isomorphism for all $a, b : A$.
\end{prop}
\begin{proof}\cite[Corollary 1.2.6]{rijke18}.
\end{proof}

\begin{prop}\label{prop:fiberwise-iso}
  Let $\wftype{A}$, $\wftype[\ctxsnoc{}[a]{A}]{B_0,B_1}$, and $\wftm{F}{\typi[a]{A}{B_0 \to B_1}}$ be
  given. Then $\wftm{\tmlam[\tmpair{a}{b}]{\tmpair{a}{Fab}}}{(\tysigma[a]{A}{B_0}) \to \tysigma[a]{A}{B_1}}$
  is an isomorphism if and only if $Fa$ is an isomorphism for all $a : A$.
\end{prop}
\begin{proof}\cite[Theorem 4.7.7]{hott-book}.
\end{proof}

\begin{defi}
  A type is \emph{contractible} if it is a proposition and inhabited.
\end{defi}

\begin{prop}\label{prop:contractible-iso}
  Any function between contractible types is an isomorphism.
\end{prop}
\begin{proof}
  This is an elementary consequence of the definition.
\end{proof}

\begin{prop}\label{prop:singleton-contractibility}
  For any $\wftype{A}$ and $\wftm{M}{A}$, the type $\tysigma[a]{A}{\typath{A}{M}{a}}$ is contractible.
\end{prop}
\begin{proof}\cite[Lemma 3.11.8]{hott-book}.
\end{proof}

Taken together, these results give us a convenient method for showing that a type is bridge-discrete without
reference to $\tmloosen{A}$.

\begin{lem}
  Suppose we have a family of maps with right inverses:
  \begin{itemize}[label=$\triangleright$]
  \item $\wftm{F}{\typi*[a]{A}{\typi[b]{A}{\tyarr{\tybridge{A}{a}{b}}{\typath{A}{a}{b}}}}}$,
  \item $\wftm{G}{\typi*[a]{A}{\typi[b]{A}{\tyrinv{\tybridge{A}{a}{b}}{\typath{A}{a}{b}}{Fab}}}}$.
  \end{itemize}
  Then $A$ is bridge-discrete. In particular, if $\tybridge{A}{a}{b}$ and $\typath{A}{a}{b}$ are isomorphic
  for all $a,b : A$, then $A$ is bridge-discrete.
\end{lem}
\begin{proof}
  By Proposition~\ref{prop:path-retract}, $Fab$ is an isomorphism for all $a,b : A$. By Proposition~\ref{prop:fiberwise-iso}, we conclude that $\tysigma[b]{A}{\tybridge{A}{a}{b}}$ and
  $\tysigma[b]{A}{\typath{A}{a}{b}}$ are isomorphic for all $a : A$. The latter is contractible by
  Proposition~\ref{prop:singleton-contractibility}, so the former is contractible as well. Thus
  $\wftm{\tmlam[\tmpair{b}{p}]{\tmpair{b}{\tmloosen{A}{p}}}}{(\tysigma[b]{A}{\typath{A}{a}{b}}) \to
    \tysigma[b]{A}{\tybridge{A}{a}{b}}}$ is an isomorphism for all $b : A$, so $A$ is bridge-discrete by
  Proposition~\ref{prop:fiberwise-iso}.
\end{proof}

\begin{lem}\label{lem:bridge-discrete-over}
  Let $\wftype{A}$ and $\wftype[\ctxsnoc{}[a]{A}]{B}$ be given. If $B$ is bridge-discrete for all
  $\ctxsnoc{}[a]{A}$, then we have the following isomorphism for all $a_0, a_1 : A$, $t : \usubst{B}{a_0}{a}$,
  $t' : \usubst{B}{a_1}{a}$, and $p : \typath{A}{a_0}{a_1}$.
  \[
    \typath{x.\usubst{B}{\tmpapp{p}{x}}{a}}{t}{t'}
    \simeq
    \tybridge{\bmx.\usubst{B}{\tmbapp{\tmloosen{A}{p}}{\bmx}}{a}}{t}{t'}
  \]
\end{lem}
\begin{proof}
  By Lemma~\ref{lem:j}, it suffices to prove the theorem when $a_1$ is $a_0$ and $p$ is $\tmplam{\_}{a_0}$. In
  that case it follows from Remark~\ref{rem:loosen-refl} and the assumption that $B$ is bridge-discrete.
\end{proof}

\begin{thm}\label{thm:sigma-bridge-discrete}
  Given $\wftype{A}$ and $\wftype[\ctxsnoc{}[a]{A}]{B}$, if $A$ is bridge-discrete and $B$ is bridge-discrete
  for all $\ctxsnoc{}[a]{A}$, then $\tysigma[a]{A}{B}$ is bridge-discrete.
\end{thm}
\begin{proof}
  Given $t,t' : \tysigma[a]{A}{B}$, we can characterize paths between $t$ and $t'$ as pairs of paths between
  their components.
  \begin{align*}
    \typath{\tysigma[a]{A}{B}}{t}{t'} \simeq \tysigma[p]{\typath{A}{\tmfst{t}}{\tmfst{t'}}}{\typath{x.\usubst{B}{\tmpapp{p}{x}}{a}}{\tmsnd{t}}{\tmsnd{t'}}}
  \end{align*}
  In the forward direction we have
  $\tmlam[p]{\tmpair{\tmplam[x]{\tmfst{\tmpapp{p}{x}}}}{\tmplam[x]{\tmsnd{\tmpapp{p}{x}}}}}$, and in the
  reverse we have $\tmlam[\tmpair{q_0}{q_1}]{\tmplam[x]{\tmpair{\tmpapp{q_0}{x}}{\tmpapp{q_1}{x}}}}$; these
  are clearly inverses. We can repeat the proof to obtain an analogous characterization of bridges in
  $\tysigma[a]{A}{B}$.
  \begin{align*}
    \tybridge{\tysigma[a]{A}{B}}{t}{t'} \simeq \tysigma[p]{\tybridge{A}{\tmfst{t}}{\tmfst{t'}}}{\tybridge{\bmx.\usubst{B}{\tmbapp{p}{\bmx}}{a}}{\tmsnd{t}}{\tmsnd{t'}}}
  \end{align*}
  By assumption, we know that $\typath{A}{\tmfst{t}}{\tmfst{t'}}$ and $\tybridge{A}{\tmfst{t}}{\tmfst{t'}}$
  are isomorphic via $\tmloosen{A}$. To show that the product types are isomorphic, it then suffices to show
  the second component types are isomorphic over $\tmloosen{A}$, \ie, that the following holds for all
  $p : \typath{A}{\tmfst{t}}{\tmfst{t'}}$.
  \[
    \typath{x.\usubst{B}{\tmpapp{p}{x}}{a}}{\tmsnd{t}}{\tmsnd{t'}}
    \simeq
    \tybridge{\bmx.\usubst{B}{\tmbapp{\tmloosen{A}{p}}{\bmx}}{a}}{\tmsnd{t}}{\tmsnd{t'}}
  \]
  This is immediate by Lemma~\ref{lem:bridge-discrete-over}.
\end{proof}

\begin{thm}
  Given $\wftype{A}$ and $\wftype[\ctxsnoc{}[a]{A}]{B}$, if $A$ is bridge-discrete and $B$ is bridge-discrete
  for all $\ctxsnoc{}[a]{A}$, then $\typi[a]{A}{B}$ is bridge-discrete.
\end{thm}
\begin{proof}
  Analogous to Theorem~\ref{thm:sigma-bridge-discrete}, using Lemmas~\ref{lem:funext-over} and~\ref{lem:bridge-funext}.
\end{proof}

\begin{thm}
  If $\wftype{A}$ is bridge-discrete, then $\typath{A}{a}{b}$ is bridge-discrete for all $a,b : A$.
\end{thm}
\begin{proof}
  Given $p,q : \typath{A}{a}{b}$, We have the following chain of isomorphisms.
  \begin{align*}
    \typath{\typath{A}{a}{b}}{p}{q}
    &\simeq \typath{x.\typath{A}{\tmpapp{p}{x}}{\tmpapp{q}{x}}}{\tmplam{\_}{a}}{\tmplam{\_}{b}} \\
    &\simeq \typath{x.\tybridge{A}{\tmpapp{p}{x}}{\tmpapp{q}{x}}}{\tmloosen{A}{\tmplam{\_}{a}}}{\tmloosen{A}{\tmplam{\_}{b}}} \\
    &\simeq \typath{x.\tybridge{A}{\tmpapp{p}{x}}{\tmpapp{q}{x}}}{\tmblam{\_}{a}}{\tmblam{\_}{b}} \\
    &\simeq \tybridge{\typath{A}{a}{b}}{p}{q}
  \end{align*}
  The first step is by reordering interval abstractions, the second by Remark~\ref{rem:loosen-refl}, the third
  by assumption that $A$ is bridge-discrete, and the fourth by reordering abstractions again.
\end{proof}

\begin{cor}
  If $\wftype{A}$ is bridge-discrete, then $\tybridge{A}{a}{b}$ is bridge-discrete for all $a,b : A$.
\end{cor}

\begin{thm}
  $\tybool$ is bridge-discrete.
\end{thm}
\begin{proof}
  We must define a right inverse to
  $\wftm{\tmloosen{\tybool}}{\typath{\tybool}{b}{b'} \to \tybridge{\tybool}{b}{b'}}$ for every
  $b,b' : \tybool$. For simplicity, we prove the case where $b = \tmtrue$ and $b' = \tmfalse$; the other cases
  follow by the same argument. In this case, we first need a function of the following type.
  \[
    \wftm{\tmtighten}{\tybridge{\tybool}{\tmtrue}{\tmfalse} \to \typath{\tybool}{\tmtrue}{\tmfalse}}
  \]

  We make use of the type $\tygel{\bmx}{\tybool}{\tybool}{\typath{\tybool}}$, the bridge from $\tybool$ to
  $\tybool$ corresponding to the path relation. This type has two canonical elements given by reflexivity at
  $\tmtrue$ and $\tmfalse$.
  \[
    \tmtrue_{\bmx} \eqdef \tmgel{\bmx}{\tmtrue}{\tmtrue}{\tmplam[\_]{\tmtrue}}
    \qquad\qquad \tmfalse_{\bmx} \eqdef \tmgel{\bmx}{\tmfalse}{\tmfalse}{\tmplam[\_]{\tmfalse}}
  \]
  Given $\ctxbdim{}[x]$, we define an auxiliary function
  $\wftm{\tmtighten[\bmx]}{\tybool \to \tygel{\bmx}{\tybool}{\tybool}{\typath{\tybool}}}$ sending each
  $b : \tybool$ to the corresponding such element.
  \[
    \tmtighten[\bmx] \eqdef \tmlam[b]{\tmif{\_.\tygel{\bmx}{\tybool}{\tybool}{\typath{\tybool}}}{b}{\tmtrue_{\bmx}}{\tmfalse_{\bmx}}}
  \]
  We then define $\tmtighten \eqdef \tmlam[q]{\tmungel{\bmx.\tmtighten[\bmx]{\tmbapp{q}{\bmx}}}}$, applying
  $\tmtighten[\bmx]$ pointwise to the input bridge.

  To equate $\tmloosen{\tybool}(\tmtighten{q})$ with $q$, we need a term as follows.
  \[
    \wftm{\tmloosentighten}{\typi[q]{\tybridge{\tybool}{\tmtrue}{\tmfalse}}{\typath{\tybridge{\tybool}{\tmtrue}{\tmfalse}}{\tmloosen{\tybool}{\tmtighten{q}}}{q}}}
  \]
  We again begin by defining an auxiliary function $\tmloosentighten[\bmx]$ of the following type.
  \[
    \wftm{\tmloosentighten[\bmx]}{\typi[b]{\tybool}{\typath{\tybool}{\tmbapp{(\tmbridgefunext{\tmloosen{\tybool} \circ \tmtighten}}{\bmx})(b)}{b}}} 
  \]
  We define $\tmloosentighten[\bmx](b)$ by induction on $b$. When $b$ is $\tmtrue$, we have the following
  chain of equalities.
  \begin{align*}
    (\tmbapp{\tmbridgefunext{\tmloosen{\tybool} \circ \tmtighten}}{\bmx})(\tmtrue)
    &= \tmbapp{\tmloosen{\tybool}{\tmungel{\bmx.\tmtighten[\bmx]{\tmtrue}}}}{\bmx} \\
    &= \tmbapp{\tmloosen{\tybool}{\tmungel{\bmx.\tmtrue_{\bmx}}}}{\bmx} \\
    &= \tmbapp{\tmloosen{\tybool}{\tmplam[\_]{\tmtrue}}}{\bmx}
  \end{align*}
  The first equation is \rulename{Extent-$\beta$}, the second is by definition of $\tmtighten[\bmx]$, and the
  third is \rulename{Gel-$\beta$}. Finally, $\tmbapp{\tmloosen{\tybool}{\tmplam[\_]{\tmtrue}}}{\bmx}$ is
  path-equal to $\tmtrue$ by Remark~\ref{rem:loosen-refl}. The $\tmfalse$ case follows by the same argument.
  Note that both $\wftm{\tmloosentighten[\bm\Ge](\tmtrue)}{\typath{\tybool}{\tmtrue}{\tmtrue}}$ and
  $\wftm{\tmloosentighten[\bm\Ge](\tmfalse)}{\typath{\tybool}{\tmfalse}{\tmfalse}}$ are reflexive paths for
  $\Ge \in \{0,1\}$.

  Given $q : \typath{\tybool}{\tmtrue}{\tmfalse}$, we see that the pointwise application
  $\tmpapp{\tmloosentighten[\bmx](\tmbapp{q}{\bmx})}{y}$ fills the following square.
  \[
    \begin{tikzpicture}
      \def\height{2}
      \def\width{9}
      \def\awidth{1.5}
      \def\aheight{0.5}
      \def\alength{0.5}
      \def\pad{1}

      \draw (-\awidth, \height+\aheight) [->,thick] to node [above] {\small $\bmx$} (-\awidth+\alength, \height+\aheight) ;
      \draw (-\awidth, \height+\aheight) [->] to node [left] {\small $y$} (-\awidth, \height+\aheight-\alength) ;
      \node (tl) at (0, \height) {$\tmtrue$} ;
      \node (tr) at (\width, \height) {$\tmfalse$} ;
      \node (bl) at (0, 0) {$\tmtrue$} ;
      \node (br) at (\width, 0) {$\tmfalse$} ;
      \draw (tl) [->,thick] to node [above] {$\tmbapp{(\tmbridgefunext{\tmloosen{\tybool} \circ \tmtighten}}{\bmx})(\tmbapp{q}{\bmx})$} (tr) ;
      \draw (bl) [->,thick] to node [below] {$\tmbapp{q}{\bmx}$} (br) ;
      \draw (tl) [->] to node [left] {$\tmtrue$} (bl) ;
      \draw (tr) [->] to node [right] {$\tmfalse$} (br) ;
      \node at (0.5*\width, 0.5*\height) {$\tmpapp{\tmloosentighten[\bmx](\tmbapp{q}{\bmx})}{y}$} ;
      \node at (-\awidth-\pad, 0) {} ;
      \node at (\width+\awidth+\pad, 0) {} ;
    \end{tikzpicture}
  \]
  By \rulename{Extent-$\beta$}, the top of this square is equal to
  $\tmbapp{\tmloosen{\tybool}(\tmtighten{q})}{\bmx}$. We may therefore define
  $\tmloosentighten \eqdef
  \tmlam[q]{\tmplam[y]{\tmblam[x]{\tmpapp{\tmloosentighten[\bmx]{\tmbapp{q}{\bmx}}}{y}}}}$.
\end{proof}

The pattern of argument we used for $\tybool$ generalizes to characterize the bridge types of other inductive
types, and in particular to show that inductive types preserve bridge-discreteness. (We will see something
like it again in \cref{sec:practice:smash}.) The fact that relativity is used (via $\tygel$-types) in these
proofs is an interesting parallel to the use of univalence to characterize the path types of higher inductive
types (\eg,~\cite[\S8.1]{hott-book}).

The bridge-discrete types are even closed under $\tygel$-types, which means that we can also carry out
parametricity arguments in $\tybdisc$. For example, we can show that the Church encoding
$\typi[A]{\tybdisc}{\tyarr{\tmfst{A}}{\tyarr{\tmfst{A}}{\tmfst{A}}}}$ is also isomorphic to $\tybool$.

\begin{thm}
  Let $\wftype{A_0,A_1}$ and $\wftype[\ctxsnoc{\ctxsnoc{}[a_0]{A_0}}[a_1]{A_1}]{R}$ be given. If $A_0$ and
  $A_1$ are bridge-discrete and $Ra_0a_1$ is bridge-discrete for all $a_0,a_1$, then
  $\tygel{\bmx}{A_0}{A_1}{a_0.a_1.R}$ is bridge-discrete for all $\ctxbdim{}[x]$.
\end{thm}
\begin{proof}
  Abbreviate $G_{\bmx} \eqdef \tygel{\bmx}{A_0}{A_1}{a_0.a_1.R}$. We show
  $\typath{G_{\bmx}}{g}{g'} \simeq \tybridge{G_{\bmx}}{g}{g'}$ for all $\ctxbdim{}[x]$ and
  $\wftm{g,g'}{G_{\bmx}}$. Note that when $\bmx$ is an endpoint, this holds by the assumptions that $A_0$ and
  $A_1$ are bridge-discrete.

  We apply $\tmextent$ at $\bmx$, first with $g$ and then with $g'$. It then remains to show that for all
  $a_0,a_0' : A_0$, $a_1,a_1' : A_1$, $q : \tybridge{\bmx.G_{\bmx}}{a_0}{a_1}$,
  $q' : \tybridge{\bmx.G_{\bmx}}{a_0'}{a_1'}$, and $\ctxbdim{}[x]$, we have
  $\typath{G_{\bmx}}{\tmbapp{q}{\bmx}}{\tmbapp{q'}{\bmx}} \simeq
  \tybridge{G_{\bmx}}{\tmbapp{q}{\bmx}}{\tmbapp{q'}{\bmx}}$ agreeing with the $\tmloosen{A}$ isomorphism when
  $\bmx = \bm0$ and $\tmloosen{B}$ isomorphism when $\bmx = \bm1$.  By Proposition~\ref{prop:bridge-isoext},
  it is enough to give an isomorphism
  \[
    \tybridge{\bmx.\typath{G_{\bmx}}{\tmbapp{q}{\bmx}}{\tmbapp{q'}{\bmx}}}{p_0}{p_1} \simeq
    \tybridge{\bmx.\tybridge{G_{\bmx}}{\tmbapp{q}{\bmx}}{\tmbapp{q'}{\bmx}}}{\tmloosen{A_0}{p_0}}{\tmloosen{A_1}{p_1}}
  \]
  for every $p _0 : \typath{A_0}{a_0}{a_0'}$ and $p_1 : \typath{A_1}{a_1}{a_1'}$. By identity elimination (Lemma~\ref{lem:j}), we may assume that $p_0$ and $p_1$ are reflexive paths, in which case (with the help of Remark~\ref{rem:loosen-refl}) we need to show the following for all $q,q' : \tybridge{\bmx.G_{\bmx}}{a_0}{a_1}$.
  \[
    \tybridge{\bmx.\typath{G_{\bmx}}{\tmbapp{q}{\bmx}}{\tmbapp{q'}{\bmx}}}{\tmplam[\_]{a_0}}{\tmplam[\_]{a_1}} \simeq
    \tybridge{\bmx.\tybridge{G_{\bmx}}{\tmbapp{q}{\bmx}}{\tmbapp{q'}{\bmx}}}{\tmblam[\_]{a_0}}{\tmblam[\_]{a_1}}
  \]
  Now we flip the binders on either side, leaving us to prove the following.
  \[
    \typath{\tybridge{\bmx.G_{\bmx}}{a_0}{a_1}}{q}{q'} \simeq
    \tybridge{\tybridge{\bmx.G_{\bmx}}{a_0}{a_1}}{q}{q'}
  \]
  In other words, we need to show that $\tybridge{\bmx.G_{\bmx}}{a_0}{a_1}$ is bridge-discrete; this type is
  isomorphic to $R$ by relativity, so we are finished by assumption.
\end{proof}

\subsection{The law of the excluded middle}\label{sec:practice:lem}

As a corollary to the bridge-discreteness of $\tybool$, we can refute the law of the excluded middle for
propositions. First, let us introduce a few variations on the excluded middle.
\begin{align*}
  \tyleminfty &\eqdef \typi[A]{\tyuniv}{\tysigma[b]{\tybool}{\tmif{\_.\tyuniv}{b}{A}{\tynot{A}}}} \\
  \tylem &\eqdef \typi[A]{\tyuniv}{\tyisprop{A} \to \tysigma[b]{\tybool}{\tmif{\_.\tyuniv}{b}{A}{\tynot{A}}}} \\
  \tywlem &\eqdef \typi[A]{\tyuniv}{\tysigma[b]{\tybool}{\tmif{\_.\tyuniv}{b}{\tynot{A}}{\tynot{\tynot{A}}}}}
\end{align*}

The \emph{unrestricted excluded middle}, $\tyleminfty$, is already refuted by univalence~\cite[Corollary
4.2.7]{hott-book}. In short, we can obtain a contradiction by examining the action of $\tyleminfty$ on the
negation isomorphism $\wftm{\tmboolnot}{\tybool \simeq \tybool}$ between $\tybool$ and itself. In univalent
type theory, it is therefore customary to restrict the law to propositions (Definition~\ref{def:proposition}).
The \emph{excluded middle for propositions}, $\tylem$, is validated in the simplicial model of univalent type
theory~\cite{kapulkin20}.

In parametric type theory, however, even this law is refuted. In fact, we can contradict the \emph{weak
  excluded middle}, $\tywlem$, which applies only to negated types. It follows from function extensionality
that negated types are always propositions, so we have $\tylem \to \tywlem$.

\begin{lem}\label{lem:to-bridge-discrete}
  If $\wftype{A}$ is bridge-discrete, then any function $\wftm{F}{\tyuniv \to A}$ is constant.
\end{lem}
\begin{proof}
  For any pair of types $B_0,B_1$, we can apply $F$ at the empty relation between them.
  \[
    \wftm{\tmblam[x]{F(\tygel{\bmx}{B_0}{B_1}{\_.\_.\bot})}}{\tybridge{A}{FB_0}{FB_1}}
  \]
  When $A$ is bridge-discrete, this induces a path between $FB_0$ and $FB_1$.
\end{proof}

\begin{thm}
  $\tynot{\tywlem}$.
\end{thm}
\begin{proof}
  Suppose we have $\wftm{w}{\tywlem}$. By Lemma~\ref{lem:to-bridge-discrete}, we know that $\tmfst \circ w$ is
  constant, so $\tmfst{w\top}$ and $\tmfst{w \bot}$ are equal. We obtain a contradiction by case analysis;
  clearly $\tmfst{w\top}$ must be $\tmfalse$ and $\tmfst{w\bot}$ must be $\tmtrue$.
\end{proof}

For a deeper exploration of the relationship between parametricity and the excluded middle, we refer to Booij,
Escard{\'{o}}, Lumsdaine, and Shulman~\cite{booij16}.

\subsection{The smash product}\label{sec:practice:smash}

Now we come to our motivating example: proving coherence laws for the smash product. In this section, we adopt
some conventions for dealing with pointed types, elements of $\tyunivptd \eqdef \tysigma[A]{\tyuniv}{A}$. We
give pointed types names like $A_*,B_*,\ldots$ and write $A,B,\ldots$ and $a_0,b_0,\ldots$ for their first and
second components respectively. Given two pointed types $A_*,B_*$, the type of basepoint-preserving functions
between them is defined as $A_* \to B_* \eqdef \tysigma[f]{A \to B}{\typath{B}{fa_0}{b_0}}$. The identity
function is a basepoint-preserving function $\wftm{\tmpair{\tmlam[a]{a}}{\tmplam[\_]{a_0}}}{A_* \to A_*}$, and
there is a unique pointed constant function $\wftm{\tmpair{\tmlam[\_]{b_0}}{\tmplam[\_]{b_0}}}{A_* \to B_*}$
between any pair of pointed types. The type of pointed functions can itself be made a pointed type
$A_* \to_* B_*$ by taking the pointed constant function as basepoint, but we will not need this here. As with
types, we write $f_*$ for basepoint-preserving functions, $f$ for the underlying function, and $f_0$ for the
proof that it preserves the basepoint. Finally, we write $\tybool_*$ for the booleans with basepoint
$\tmtrue$.

The underlying type of the smash product is given by the following higher inductive type.
\[
  \begin{array}{l}
    \dataheading{A_* \wedge B_*} \\
    \niceconstr{\tmsmpair{a : A}{b : B}}{A_* \wedge B_*} \\
    \niceconstr{\tmsmbasel}{A_* \wedge B_*} \\
    \niceconstr{\tmsmbaser}{A_* \wedge B_*} \\
    \niceconstr{\tmsmgluel}[\ctxpdim{\ctxsnoc{}[b]{B}}[x]]{A_* \wedge B_*}{\tube{x=0}{\tmsmbasel} \mid \tube{x=1}{\tmsmpair{a_0}{b}}} \\
    \niceconstr{\tmsmgluer}[\ctxpdim{\ctxsnoc{}[a]{A}}[x]]{A_* \wedge B_*}{\tube{x=0}{\tmsmbaser} \mid \tube{x=1}{\tmsmpair{a}{b_0}}}
  \end{array}
\]
In words, $A_* \wedge B_*$ is the ordinary product $A \times B$ quotiented by the relation collapsing together
all elements of the form $\tmpair{a_0}{b}$ or $\tmpair{a}{b_0}$. Elements of the former form are identified
with a new ``hub'' point $\tmsmbasel$, while elements of the latter are identified with a separate point
$\tmsmbaser$, producing a shape shown in \cref{fig:smash}. We write $A_* \wedge_* B_*$ for the smash product
viewed as a pointed type with basepoint $\tmsmpair{a_0}{b_0}$.

\begin{figure}
  \centering
  \[
  \begin{tikzpicture}
    \def\height{1}
    \def\width{3}
    \def\slant{1}
    \def\aoffset{0.3}
    \def\boffset{0.5}

    \coordinate (bl) at (0, 0);
    \coordinate (bm) at (\width, 0);
    \coordinate (br) at (2*\width, 0);
    \coordinate (ml) at (\slant, \height);
    \coordinate (mm) at (\width+\slant, \height);
    \coordinate (mr) at (2*\width+\slant, \height);
    \coordinate (tl) at (2*\slant, 2*\height);
    \coordinate (tm) at (\width+2*\slant, 2*\height);
    \coordinate (tr) at (2*\width+2*\slant, 2*\height);
    \coordinate (top) at (\width+\slant, 3*\height);
    \coordinate (bot) at (\width+\slant, -\height);

    \draw[dashed,|-|] ($ (bl) + (-\aoffset,\slant*\aoffset) $) -- node[sloped,left,pos=0] {$A$} node[sloped,above,pos=1] {$a_0$} ($ (ml) + (-\aoffset,\slant*\aoffset) $);
    \draw[dashed,-|] ($ (ml) + (-\aoffset,\slant*\aoffset) $) -- ($ (tl) + (-\aoffset,\slant*\aoffset) $);

    \draw[dashed,|-|] ($ (tl) + (0,\boffset) $) -- node[above, pos=1] {$b_0$} ($ (tm) + (0,\boffset) $);
    \draw[dashed,-|] ($ (tm) + (0,\boffset) $) -- node[right,pos=1] {$B$} ($ (tr) + (0,\boffset) $);

    \fill[draw=black,fill=lightgray,thick,line join=bevel,fill opacity=0.8] (bm) -- (tm) -- (bot) -- (bm);
    \fill[draw=black,fill=lightgray,thick,fill opacity=0.8] (bl) -- (br) -- (tr) -- (tl) -- (bl);
    \fill[draw=black,fill=lightgray,thick,fill opacity=0.8] (bm) -- (tm);
    \fill[draw=black,fill=lightgray,thick,line join=bevel,fill opacity=0.8] (ml) -- (mr) -- (top) -- (ml);

    \node[circle,fill=black,inner sep=2pt,label={$~~\tmsmbasel$}] at (top) {};
    \node[circle,fill=black,inner sep=2pt,label=below:{$~~\tmsmbaser$}] at (bot) {};

    \coordinate (fun) at (0.4*\width+\slant,\height);
    \node[circle,fill=black,inner sep=2pt,label=below:{$\tmsmpair{a_0}{b}$}] at (fun) {};
    \draw[thick,->-=0.5] (top) -- node[right, pos=0.65] {$\scriptstyle \tmsmgluel{b}{-}$} (fun);

  \end{tikzpicture}
\]
\caption{The smash product of $\tmpair{A}{a_0}$ and $\tmpair{B}{b_0}$}%
\label{fig:smash}
\end{figure}

We will begin by focusing on the following theorem.

\begin{thm}\label{thm:smash-binary}
  Any family of pointed functions $\typi[A_*,B_*]{\tyunivptd}{(A_* \wedge_* B_* \to A_* \wedge_* B_*)}$ is
  either the polymorphic identity or the polymorphic constant pointed function, up to a path.
\end{thm}

In an effort to show we have nothing up our sleeves, we will avoid sweeping gory details---that is, coherence
proofs---under the rug. However, we encourage the reader to focus on the broad strokes of the argument, and as
such we will be less diligent about \emph{explaining} the gory details.

The relations we use in the following will all be graphs of functions. As such, we introduce the following
shorthand notation.

\begin{defi}
  Given $f : A \to B$, write $\tygr{r}{A}{B}{f} \eqdef \tygel{\bmr}{A}{B}{a.b.\typath{B}{fa}{b}}$. Given
  $f_* : A_* \to B_*$, define
  $\tygr*{r}{A_*}{B_*}{f_*} \eqdef \tmpair{\tygr{r}{A}{B}{f}}{\tmgel{\bmr}{a_0}{b_0}{f_0}} \in \tyunivptd$.
\end{defi}

We prove a \emph{graph lemma} (Lemma~\ref{lem:smash-graph-lemma}) that relates the smash product of
$\tygr*$-types with the action of the smash product on their underlying functions.  First, the following two
technical definitions will be handy for concisely filling coherence conditions.

\begin{defi}[Concatenation by inverse]
  Let $\wftm{M}{A}$, $\wfpdim{r}$, and $\wftm[\ctxpdim{}[x]]{N}{A}$ with
  $\eqtm[\ctxcst{}{r=1}]{M}{\usubstdim{N}{1}{x}}{A}$ be given. For any $\wfpdim{s}$, define
  $\wftm{\tmconcinv{A}{r}{s}{M}{x.N}}{A}$ as follows.
  \[
    \tmconcinv{A}{r}{s}{M}{x.N} \eqdef \hcomp{A}{1}{s}{M}{\tube{r=0}{\_.M},\tube{r=1}{x.N}}
  \]
  \[
    \begin{tikzpicture}
      \def\height{2}
      \def\width{6}
      \def\awidth{1.5}
      \def\aheight{0.5}
      \def\alength{0.5}
      \def\pad{1}

      \draw (-\awidth, \height+\aheight) [->] to node [above] {\small $r$} (-\awidth+\alength, \height+\aheight) ;
      \draw (-\awidth, \height+\aheight) [->] to node [left] {\small $s$} (-\awidth, \height+\aheight-\alength) ;
      \node (tl) at (0, \height) {$\bullet$} ;
      \node (tr) at (\width, \height) {$\bullet$} ;
      \node (bl) at (0, 0) {$\bullet$} ;
      \node (br) at (\width, 0) {$\bullet$} ;
      \draw (tl) [->,dashed] to node [above] {} (tr) ;
      \draw (bl) [->] to node [below] {$M$} (br) ;
      \draw (tl) [->] to node [left] {$M$} (bl) ;
      \draw (tr) [->] to node [right] {$\usubst{N}{s}{x}$} (br) ;
      \node at (0.5*\width, 0.5*\height) {$\tmconcinv{A}{r}{s}{M}{x.N}$} ;
      \node at (-\awidth-\pad, 0) {} ;
      \node at (\width+\awidth+\pad, 0) {} ;
    \end{tikzpicture}
  \]
  The term $\tmconcinv{A}{r}{0}{M}{x.N}$ is the result of concatenating $M$ (as a path in direction $r$) with
  the inverse of $x.N$; we need the general form $\tmconcinv{A}{r}{s}{M}{x.N}$ to relate the composite to
  other terms.
\end{defi}

\begin{lem}[Join connection]
  For any $\wftm{P}{\typath{A}{M}{N}}$, we have a term as follows.
  \[
    \wftm{\tmorcnx{A}{P}}{\typath{x.\typath{A}{\tmpapp{P}{x}}{N}}{P}{\tmplam[\_]{N}}}
  \]
  \[
    \begin{tikzpicture}
      \def\height{2}
      \def\width{6}
      \def\awidth{1.5}
      \def\aheight{0.5}
      \def\alength{0.5}
      \def\pad{1}

      \draw (-\awidth, \height+\aheight) [->] to node [above] {\small $x$} (-\awidth+\alength, \height+\aheight) ;
      \draw (-\awidth, \height+\aheight) [->] to node [left] {\small $y$} (-\awidth, \height+\aheight-\alength) ;
      \node (tl) at (0, \height) {$\bullet$} ;
      \node (tr) at (\width, \height) {$\bullet$} ;
      \node (bl) at (0, 0) {$\bullet$} ;
      \node (br) at (\width, 0) {$\bullet$} ;
      \draw (tl) [->] to node [above] {$\tmpapp{P}{x}$} (tr) ;
      \draw (bl) [->] to node [below] {$N$} (br) ;
      \draw (tl) [->] to node [left] {$\tmpapp{P}{y}$} (bl) ;
      \draw (tr) [->] to node [right] {$N$} (br) ;
      \node at (0.5*\width, 0.5*\height) {$\tmpapp{\tmpapp{\tmorcnx{A}{P}}{x}}{y}$} ;
      \node at (-\awidth-\pad, 0) {} ;
      \node at (\width+\awidth+\pad, 0) {} ;
    \end{tikzpicture}
  \]
\end{lem}
\begin{proof}
  By Lemma~\ref{lem:j}, it suffices to construct a term when $P$ is a constant path
  $\wftm{\tmplam[\_]{M}}{\typath{A}{M}{M}}$, in which case we have
  $\wftm{\tmplam[\_]{\tmplam[\_]{M}}}{\typath{\typath{A}{M}{M}}{\tmplam[\_]{M}}{\tmplam[\_]{M}}}$.
\end{proof}

The smash product has a functorial action on pointed functions, which we define as follows.

\begin{defi}
  Given $f_* : A_* \to C_*$ and $g_* : B_* \to D_*$, we inductively define a map
  $f_* \wedge g_* \in A_* \wedge B_* \to C_* \wedge D_*$ as follows.
  \[
    \begin{array}{lcl}
      (f_* \wedge g_*)(\tmsmpair{a}{b}) &\eqdef& \tmsmpair{fa}{gb} \\
      (f_* \wedge g_*)(\tmsmbasel) &\eqdef& \tmsmbasel \\
      (f_* \wedge g_*)(\tmsmbaser) &\eqdef& \tmsmbaser \\
      (f_* \wedge g_*)(\tmsmgluel{b}{y}) &\eqdef& \tmconcinv{C_* \wedge D_*}{y}{0}{\tmsmgluel{gb}{y}}{z.\tmsmpair{\tmpapp{f_0}{z}}{gb}} \\
      (f_* \wedge g_*)(\tmsmgluer{a}{y}) &\eqdef& \tmconcinv{C_* \wedge D_*}{y}{0}{\tmsmgluer{y}{fa}}{z.\tmsmpair{fa}{\tmpapp{g_0}{z}}}
    \end{array}
  \]
\end{defi}

We now prove the graph lemma: that there is a map from the smash product of two $\tygr*$-types to the
$\tygr$-type corresponding to the smash of their underlying functions. We expect that this map is in fact an
isomorphism and that a similar principle holds for $\tygel$-types more generally, but such results are not
necessary here.

\begin{lem}[Graph Lemma for $\wedge$]%
  \label{lem:smash-graph-lemma}
  For any $\wfbdim{\bmr}$, there is a map
  \[
    \tmsmashgr{r} \in \tygr*{\bmr}{A_*}{C_*}{f_*} \wedge \tygr*{\bmr}{B_*}{D_*}{g_*} \to \tygr{\bmr}{A_* \wedge B_*}{C_* \wedge D_*}{f_* \wedge g_*}
  \]
  equal to the identity function on $A_* \wedge_* B_*$ when $\bmr = \bm0$ and on $C_* \wedge_* D_*$ when
  $\bmr = \bm1$.
\end{lem}
\begin{proof}
  We define the map by induction on the smash product in the domain.
  \begin{itemize}[label=$\triangleright$]
  \item Case $\tmsmpair{m}{n}$: We test whether $\bmr$ is a constant or variable using $\tmextent$. In the
    constant cases, we return $\tmsmpair{m}{n}$. In the case $\bmr$ is a variable $\bmx$, we learn that
    $m$ and $n$ are the instantiation at $\bmx$ of bridges over their types; by \rulename{Gel-$\eta$}, they
    are of the form $m = \tmgel{\bmx}{a}{c}{p}$ and $n = \tmgel{\bmx}{b}{d}{q}$. We return
    $\tmgel{\bmx}{\tmsmpair{a}{b}}{\tmsmpair{c}{d}}{\tmplam[z]{\tmsmpair{\tmpapp{p}{y}}{\tmpapp{q}{y}}}}$.
  \item Case $\tmsmbasel$: We return $\tmgel{\bmr}{\tmsmbasel}{\tmsmbasel}{\tmplam[\_]{\tmsmbasel}}$.
  \item Case $\tmsmbaser$: Symmetric to $\tmsmbasel$.
  \item Case $\tmsmgluel{n}{y}$: We test whether $\bmr$ is a constant or variable using $\tmextent$. In the
    constant cases, we return $\tmsmgluel{n}{y}$. In the case $\bmr$ is a variable $\bmx$, we learn that
    $n$ is the instantiation at $\bmx$ of a bridge; by \rulename{Gel-$\eta$}, it is of the form
    $n = \tmgel{\bmx}{b}{d}{q}$. We return $\tmgel{\bmx}{\tmsmgluel{b}{y}}{\tmsmgluel{d}{y}}{\tmplam[z]{\cdots}}$,
    where $\cdots$ is the following composite.
    \[
      \bighcomp{C_* \wedge D_*}{1}{0}{\tmsmgluel{\tmpapp{q}{z}}{y}}{
        \begin{array}{lcl}
          \arraytube{y=0}{\_.\tmsmbasel} \\
          \arraytube{y=1}{w.\tmsmpair{\tmpapp{\tmpapp{\tmorcnx{A}{f_0}}{z}}{w}}{\tmpapp{q}{z}}} \\
          \arraytube{z=0}{w.\tmconcinv{C_* \wedge D_*}{y}{w}{\tmsmgluel{gb}{y}}{z.\tmsmpair{\tmpapp{f_0}{z}}{gb}}} \\
          \arraytube{z=1}{\_.\tmsmgluel{d}{y}}
        \end{array}
      }
    \]
  \item Case $\tmsmgluer{m}{y}$: Symmetric to $\tmsmgluel{n}{y}$.
  \end{itemize}
  When $\bmr$ is a constant, the resulting function simplifies to the $\eta$-expansion of the identity
  function on $A_* \wedge B_*$. By a simple induction on $A_* \wedge B_*$, the $\eta$-expansion is path-equal
  to the identity function. We may therefore apply an $\hcomp$ to adjust the boundary and obtain a function
  that is exactly the identity when $\bmr = \bm0$ or $\bmr = \bm1$.
\end{proof}

Finally, we use the fact that $\tybool_* \wedge \tybool_*$ is isomorphic to $\tybool_*$. This is a consequence
of more general facts---that $\tybool_*$ is a unit for the smash product, or alternatively that
$(1 + X) \wedge (1 + Y) \simeq 1 + (X \times Y)$ when we take $1$ for each basepoint---but we prove the
special case directly for simplicity's sake. The importance of $\tybool_*$ arises from the fact that elements
of a pointed type $X_*$ are in correspondence with pointed maps $\tybool_* \to X_*$. As such, we can use
naturality conditions with respect to functions $\tybool_* \to X_*$ to ``probe'' the behavior of a function
polymorphic in pointed types, as we will see in Lemma~\ref{lem:smash-workhorse}.

\begin{lem}[Smash of booleans]\label{lem:smash-bool}
  $\tybool_* \wedge \tybool_*$ is isomorphic to $\tybool_*$; in particular, any element of
  $\tybool_* \wedge \tybool_*$ is path-equal to either $\tmsmpair{\tmtrue}{\tmtrue}$ or
  $\tmsmpair{\tmfalse}{\tmfalse}$.
\end{lem}
\begin{proof}
  In one direction, we define $\wftm{F}{\tybool \to \tybool_* \wedge \tybool_*}$ to send $\tmtrue$ to
  $\tmsmpair{\tmtrue}{\tmtrue}$ and $\tmfalse$ to $\tmsmpair{\tmfalse}{\tmfalse}$. In the other, we define
  $\wftm{G}{\tybool_* \wedge \tybool_* \to \tybool}$ to send $\tmsmpair{\tmfalse}{\tmfalse}$ to $\tmfalse$ and
  all other constructors to $\tmtrue$. Clearly $G \circ F$ is the identity. For the other inverse condition,
  we show $\typi[s]{\tybool_* \wedge \tybool_*}{\typath{\tybool_* \wedge \tybool_*}{s}{F(Gs)}}$ by smash
    product induction as follows.
  \begin{itemize}[label=$\triangleright$]
  \item Case $\tmsmpair{\tmtrue}{\tmtrue}$: Reflexivity.
  \item Case $\tmsmpair{\tmtrue}{\tmfalse}$: \\
    $\tmplam[y]{\hcomp{\tybool_*\wedge\tybool_*}{0}{1}{\tmsmgluel{\tmtrue}{y}}{\tube{y=0}{x.\tmsmgluel{\tmfalse}{x}},\tube{y=1}{\_.\tmsmpair{\tmtrue}{\tmtrue}}}}$.
  \item Case $\tmsmpair{\tmfalse}{\tmfalse}$: Reflexivity.
  \item Case $\tmsmbasel$: $\tmplam[y]{\tmsmgluel{\tmtrue}{y}}$.
  \item Case $\tmsmgluel{\tmtrue}{x}$: $\tmpapp{\tmorcnx{\tybool_* \wedge \tybool_*}{\tmplam[y]{\tmsmgluel{\tmtrue}{y}}}}{x}$.
  \item Case $\tmsmgluel{\tmfalse}{x}$: \\
    $\tmplam[y]{\hcomp{\tybool_*\wedge\tybool_*}{0}{x}{\tmsmgluel{\tmtrue}{y}}{\tube{y=0}{x.\tmsmgluel{\tmfalse}{x}},\tube{y=1}{\_.\tmsmpair{\tmtrue}{\tmtrue}}}}$.
  \end{itemize}
  The cases for $\tmsmpair{\tmtrue}{\tmfalse}$, $\tmsmbaser$, and $\tmsmgluer$ are obtained by taking the
  cases for $\tmsmpair{\tmfalse}{\tmtrue}$, $\tmsmbasel$, and $\tmsmgluel$ respectively and replacing
  $\tmsmgluel$ with $\tmsmgluer$ everywhere.
\end{proof}

The following result, which characterizes terms
$\wftm{F}{\typi[A_*,B_*]{\tyunivptd}{A \to B \to A_* \wedge B_*}}$, is the linchpin of the argument; all uses
of internal parametricity in the final results factor through this lemma. As we only use internal
parametricity with relations that are graphs of functions, this result may also be cast as a corollary of the
\emph{naturality} of such terms, a special case of parametricity. In particular, we use the following
naturality square for $a : A$ and $b : B$, where $\wftm{{[c]}_*}{\tybool_* \to C_*}$ is the pointed function
sending $\tmtrue$ to $c_0$ and $\tmfalse$ to $c$.

\[
  \begin{tikzcd}[column sep=5em, row sep=3em]
    \tybool \times \tybool \ar{r}{F\tybool_*\tybool_*} \ar{d}[left]{[a] \times [b]} & \tybool_* \wedge \tybool_* \ar{d}{{[a]}_* \wedge {[b]}_*} \\
    A \times B \ar{r}[below]{FA_*B_*} & A_* \wedge B_*
  \end{tikzcd}
\]

\begin{lem}[Workhorse lemma]\label{lem:smash-workhorse}
  Let $\wftm{F}{\typi[A_*,B_*]{\tyunivptd}{A \to B \to A_* \wedge B_*}}$. Then $F$ is path equal to one of the
  following.
  \begin{itemize}[label=$\triangleright$]
  \item $\tmlam[\_]{\tmlam[\_]{\tmlam[a]{\tmlam[b]{\tmsmpair{a}{b}}}}}$.
  \item $\tmlam[A_*]{\tmlam[B_*]{\tmlam[\_]{\tmlam[\_]{\tmsmpair{a_0}{b_0}}}}}$.
  \end{itemize}
\end{lem}
\begin{proof}
  We show that the identity of $F$ is determined by the value of
  $F(\tybool_*)(\tybool_*)(\tmfalse)(\tmfalse)$. Let $A_* : \tyunivptd$, $B_* : \tyunivptd$, $a : A$, and
  $b : B$ be given.

  We have a function $\wftm{{[a]}_*}{\tybool_* \to A_*}$ sending $\tmtrue$ to $a_0$ and $\tmfalse$ to $a$,
  likewise $\wftm{{[b]}_*}{\tybool_* \to B_*}$ sending $\tmtrue$ to $b_0$ and $\tmfalse$ to $b$. Abstract a
  bridge variable $\ctxbdim{}[x]$. We abbreviate $G^a_* \eqdef \tygr*{\bmx}{\tybool_*}{A_*}{{[a]}_*}$ and
  $G^b_* \eqdef \tygr*{\bmx}{\tybool_*}{B_*}{{[b]}_*}$. Applying $F$ at $G^a_*$ and $G^b_*$, we have the
  following.
  \[
    \wftm{FG^a_*G^b_*(\tmgel{\bmx}{\tmfalse}{a}{\tmplam[\_]{a}})(\tmgel{\bmx}{\tmfalse}{b}{\tmplam[\_]{b}})}{G^a_* \wedge G^b_*}
  \]
  At $\bmx = \bm0$, this term is $F(\tybool_*)(\tybool_*)(\tmfalse)(\tmfalse)$, and at $\bmx = \bm1$
  it is $FA_*B_*ab$. Now we apply the Graph Lemma to obtain a term in
  $\tygr{x}{\tybool_* \wedge \tybool_*}{A_* \wedge B_*}{{[a]}_* \wedge {[b]}_*}$ with the same boundary.  Finally,
  we apply $\tmungel$ to extract a path from
  $({[a]}_* \wedge {[b]}_*)(F(\tybool_*)(\tybool_*)(\tmfalse)(\tmfalse))$ to $FA_*B_*ab$.  We therefore see that
  $F$ is the pairing function if $F(\tybool_*)(\tybool_*)(\tmfalse)(\tmfalse)$ is
  $\tmsmpair{\tmfalse}{\tmfalse}$ and the constant function if it is $\tmsmpair{\tmtrue}{\tmtrue}$; by Lemma~\ref{lem:smash-bool}, we are in one of these two cases.
\end{proof}

\begin{cor}\label{cor:smash-workhorse-set}
  $\typi[A_*,B_*]{\tyunivptd}{A \to B \to A_* \wedge B_*}$ is a \emph{set}: any pair of paths between two
  elements of the type are path-equal.
\end{cor}
\begin{proof}
  Lemma~\ref{lem:smash-workhorse} shows that the type is isomorphic to $\tybool$, which is a set.
\end{proof}

This is everything we need to prove the final result.

\begin{proof}[Proof of Theorem~\ref{thm:smash-binary}]
  Let $\wftm{F_*}{\typi[A_*,B_*]{\tyunivptd}{A_* \wedge_* B_* \to A_* \wedge_* B_*}}$ be given. To
  characterize $F_*$, we need to characterize its behavior on each constructor of $A_* \wedge B_*$ as well as
  the proof that it preserves the basepoint of $A_* \wedge_* B_*$.

  First, by Lemma~\ref{lem:smash-workhorse}, we know that $\tmplam[a]{\tmplam[b]{FA_*B_*(\tmsmpair{a}{b})}}$
  is either pairing or constant. The values of $FA_*B_*\tmsmbasel$ and $FA_*B_*\tmsmbaser$ must be path-equal
  to $\tmsmbasel$ and $\tmsmbaser$ respectively, as $F$ is basepoint-preserving and $\tmsmbasel$
  ($\tmsmbaser$) is connected to the basepoint by $\tmsmgluel{b_0}{-}$ ($\tmsmgluer{a_0}{-}$).

  Next, observe that we can capture the behavior of $F$ on $\tmsmgluel$ by the following term, which is a path
  in $\typi[A_*,B_*]{\tyunivptd}{A \to B \to A_* \wedge_* B_*}$ between
  $\tmlam[A_*]{\tmlam[B_*]{\tmlam[\_]{\tmlam[\_]{FA_*B_*\tmsmbasel}}}}$ and
  $\tmlam[A_*]{\tmlam[B_*]{\tmlam[\_]{\tmlam[b]{FA_*B_*(\tmsmpair{a}{b})}}}}$.
  \[
    \tmplam[y]{\tmlam[A_*]{\tmlam[B_*]{\tmlam[\_]{\tmlam[b]{FA_*B_*(\tmsmgluel{b}{y})}}}}}
  \]
  By Corollary~\ref{cor:smash-workhorse-set}, this path is path-equal to any other path in this type, in
  particular path-equal to whatever we need it to be to complete this proof. The same applies to
  $\tmsmbaser$. Finally, we can apply the same trick for the basepoint path, writing it as a path in the type
  from Corollary~\ref{cor:smash-workhorse-set} as follows.
  \[
    \tmplam[y]{\tmlam[A_*]{\tmlam[B_*]{\tmlam[\_]{\tmlam[\_]{\tmpapp{f_0A_*B_*}{y}}}}}} \qedhere
  \]
\end{proof}

Now we argue that this strategy can be used to prove the $n$-ary generalization in a uniform way. (The binary
version is in fact not very useful on its own; the direct proof of commutativity for the smash product is
uncharacteristically straightforward, because the definition of $\wedge$ is completely symmetric.)

\begin{thm}\label{thm:smash-nary}
  Any function
  $\typi[A^0_*,\ldots,A^n_*]{\tyunivptd}{(A^0_* \wedge_* \cdots \wedge_* A^n_*) \to (A^0_* \wedge_* \cdots
    \wedge_* A^n_*)}$ (associating $\wedge_*$ to the right) is either the polymorphic identity or the
  polymorphic constant pointed function.

\end{thm}
\begin{proof}
  We show by induction on $i \le n + 1$ that any
  \[
    \typi[A^0_*,\ldots,A^n_*]{\tyunivptd}{A^0 \to \cdots \to A^{n-i} \to (A^{n-i+1}_* \wedge_* \cdots \wedge_* A^n_*) \to (A^0_* \wedge_* \cdots
      \wedge_* A^n_*)}
  \]
  is either given by iterated pairing or constant. For $i = 0$, it follows from a simple $n$-ary
  generalization of the workhorse lemma (instantiating each type argument with a graph and applying the binary
  Graph Lemma repeatedly). For $i > 0$, it follows from the induction hypothesis by the same argument as in
  the proof of Theorem~\ref{thm:smash-binary}.
\end{proof}

The key here is that we are never involved in an iterated induction on smash products: for each $i$ in the
proof of Theorem~\ref{thm:smash-nary}, we have an argument by induction on one occurrence of the smash
product, but these arguments do not overlap.

\section{Computational interpretation}\label{sec:computational}

We now develop the computational interpretation underlying parametric cubical type theory, building on the
work of Allen for Martin-L\"{o}f type theory~\cite{allen87} and Angiuli \etal\ for cartesian cubical type theory~\cite{angiuli18}. We closely follow the presentation in Angiuli's thesis~\cite{angiuli19}; we will give a
reasonably complete tour through the definitions, but rely on~\cite{angiuli19} for many results that are
essentially unaffected by the addition of bridge intervals and parametricity primitives.

An interpretation in these frameworks is built from two components: a deterministic \emph{operational
  semantics} on closed untyped terms and a \emph{value type system}. The former explains the evaluation of
terms; the latter explains which closed values are names for types and which closed values are elements of
said types. Given these two components, we derive an interpretation of the open judgments---$\wftype[\GG]{A}$
and so on---by extending the value type system first to arbitrary closed terms (roughly, a term is well-typed
when it evaluates to a well-typed value) and then to open terms (an open term is well-typed when its closed
instances are well-typed).

\subsection{Interval contexts and terms}

In the above and the following, \emph{closed} refers to terms that do not contain term variables but that may
contain interval variables. It is essential to consider evaluation of terms containing interval variables in
order to accommodate the terms $\coe{x.A}{r}{s}{M}$ and $\tmungel{\bmx.N}$, which evaluate terms (here $A$ and
$N$) under interval binders. We use $\GPS$ to denote contexts consisting solely of path and bridge interval
variables.
\[
  \GPS \Coloneqq \ctxnil \mid \ctxpdim{\GPS}[x] \mid \ctxbdim{\GPS}[x]
\]
We write $\dctxsubst{\GPS'}{\Gps}{\GPS}$ for interval substitutions, which take terms $M$ in context $\GPS$ to
terms $\usubstdims{M}{\Gps}$ in context $\GPS'$. As always, path interval variables are structural and bridge
interval variables are affine; $\Gps$ cannot identify two bridge variables except by sending both to $\bm0$ or
$\bm1$.

\begin{defi}%
  \label{def:closed-interval-judgments}
  The path interval term judgment $\wfpdim<\GPS>{r}$ is defined to hold when either $r \in \{0,1\}$ or $r = x$
  where $(x : \BI) \in \GPS$; the bridge interval term judgment $\wfbdim<\GPS>{\bmr}$ is defined likewise. The
  interval substitution judgment is then inductively generated by the following rules.
  \begin{mathpar}
    \inferrule
    { }
    {\wfsubst<\GPS'>{\substnil}{\ctxnil}}
    \and
    \inferrule
    {\wfsubst<\GPS'>{\Gps}{\GPS} \\
      \wfpdim<\GPS'>{r}}
    {\wfsubst<\GPS'>{(\substpdim{\Gps}{r}[x])}{(\ctxpdim{\GPS}[x])}}
    \and
    \inferrule
    {\wfbdim<\GPS'>{\bmr} \\
      \wfsubst<\ctxres{\GPS'}{\bmr}>{\Gps}{\GPS}}
    {\wfsubst<\GPS'>{(\substbdim{\Gps}{\bmr}[\bmx])}{(\ctxbdim{\GPS}[\bmx])}}
  \end{mathpar}
  The judgment $\wfcst<\GPS>{\Gx}$ is likewise inductively generated by constraints of the form
  $\wfcst<\GPS>{(r = s)}$ and $\wfcst<\GPS>{(\bmr = \bm\Ge)}$.
\end{defi}

\begin{rem}
  We have an operator $\forall \bmx. -$ on constraints defined as follows.
  \begin{align*}
    \forall \bmx. (r = s) &\eqdef (r = s) \\
    \forall \bmx. (\bmx = \bm\Ge) &\eqdef (\bm0 = \bm1) \\
    \forall \bmx. (\bmr = \bm\Ge) &\eqdef (\bmr = \bm\Ge) \quad \text{if $\bmr \neq \bmx$}
  \end{align*}
\end{rem}

\subsection{Operational semantics}

An \emph{operational semantics}, which specifies the evaluation of closed terms, is defined by judgments
$\isval{M}$ (``$M$ is a value'') and $M \steps M'$ (``$M$ steps to $M'$'') operating on closed terms. We write
$M \msteps M'$ to mean that $M$ steps to $M'$ in zero or more steps and $M \evals V$ to mean that
$M \msteps V$ for some $\isval{V}$.

We give the defining rules for our operational semantics in \cref{fig:opsem}. We show only those rules that
involve the parametricity primitives; for everything else, we refer to~\cite[\S4.1]{angiuli19}. Although we
choose a specific operational semantics here, the interpretation goes through for any operational semantics
that extends it; we need only the presence of these rules, not the absence of others.

\begin{figure}
  \begin{mathpar}
      \inferrule
      { }
      {\isval{\tybridge{\bmx.A}{M_0}{M_1}}}
      \and
      \inferrule
      { }
      {\isval{\tmblam[x]{P}}}
      \and
      \inferrule
      {Q \steps Q'}
      {\tmbapp{Q}{\bmr} \steps \tmbapp{Q'}{\bmr}}
      \and
      \inferrule
      { }
      {\tmbapp{(\tmblam[x]{P})}{\bmr} \steps \usubstdim{P}{\bmr}{\bmx}}
      \and
      \inferrule
      { }
      {\hcomp{\tybridge{\bmx.A}{M_0}{M_1}}{r}{s}{M}{\sys{\xi_i}{y.N_i}} \steps \\\\
        \raisebox{-0.8em}{\tmblam[x]{\hcomp{A}{r}{s}{\tmbapp{M}{\bmx}}{\sys{\xi_i}{y.\tmbapp{N_i}{\bmx}},\tube{\bmx=\bm0}{\_.M_0},\tube{\bmx=\bm1}{\_.M_1}}}}}
      \and
      \inferrule
      { }
      {\coe{y.\tybridge{\bmx.A}{M_0}{M_1}}{r}{s}{Q} \steps \tmblam[x]{\comp{y.A}{r}{s}{\tmbapp{Q}{\bmx}}{\tube{\bmx=\bm0}{y.M_0},\tube{\bmx=\bm1}{y.M_1}}}}
      \and
      \inferrule
      {\Ge \in \{0,1\}}
      {\tmextent{\bm\Ge}{M}{a_0.N_0}{a_1.N_1}{a_0.a_1.\overline{a}.\overline{N}} \steps \usubst{N_\Ge}{M}{a}}
      \and
      \inferrule
      { }
      {\tmextent{\bmx}{M}{a_0.N_0}{a_1.N_1}{a_0.a_1.\overline{a}.\overline{N}} \steps \tmbapp{\usubst{\usubst{\usubst{\overline{N}}{\usubstdim{M}{\bm0}{\bmx}}{a_0}}{\usubstdim{M}{\bm1}{\bmx}}{a_1}}{\tmblam[x]{M}}{\overline{a}}}{x}}
      \and
      \inferrule
      {\Ge \in \{0,1\}}
      {\tygel{\bm\Ge}{A_0}{A_1}{a_0.a_1.R} \steps A_\Ge}
      \and
      \inferrule
      { }
      {\isval{\tygel{\bmx}{A_0}{A_1}{a_0.a_1.R}}}
      \and
      \inferrule
      {\Ge \in \{0,1\}}
      {\tmgel{\bm\Ge}{M_0}{M_1}{P} \steps M_\Ge}
      \and
      \inferrule
      { }
      {\isval{\tmgel{\bmx}{M_0}{M_1}{P}}}
      \and
      \inferrule
      {Q \steps Q'}
      {\tmungel{\bmx.Q} \steps \tmungel{\bmx.Q'}}
      \and
      \inferrule
      { }
      {\tmungel{\bmx.\tmgel{\bmx}{M_0}{M_1}{P}} \steps P}
      \and
      \mprset{vskip=0.1em}
      \inferrule
      {M_{\Ge,y} \eqdef \hcomp{A_\Ge}{r}{y}{\usubstdim{Q}{\bm\Ge}{\bmx}}{\sys{\usubstdim{\xi_i}{\bm\Ge}{\bmx}}{y.\usubstdim{Q_i}{\bm\Ge}{\bmx}}} \\
        P \eqdef \comp{y.\usubst{R}{M_{0,y},M_{1,y}}{a_0,a_1}}{r}{s}{\tmungel{\bmx.Q}}{\sys{\forall \bmx. \xi_i}{y.\tmungel{\bmx.Q_i}}}}
      {\hcomp{\tygel{\bmx}{A_0}{A_1}{a_0.a_1.R}}{r}{s}{Q}{\sys{\xi_i}{y.Q_i}} \steps
        \tmgel{\bmx}{M_{0,s}}{M_{1,s}}{P}}
      \and
      \inferrule
      {M_{\Ge,y} \eqdef \coe{y.A_\Ge}{r}{y}{\usubstdim{Q}{\bm\Ge}{\bmx}} \\
        P \eqdef \coe{y.\usubst{R}{M_{0,y},M_{1,y}}{a_0,a_1}}{r}{s}{\tmungel{\bmx.Q}}}
      {\coe{y.\tygel{\bmx}{A_0}{A_1}{a_0.a_1.R}}{r}{s}{Q} \steps \tmgel{\bmx}{M_{0,s}}{M_{1,s}}{P}}
    \end{mathpar}
  \caption{Operational semantics of parametric cubical type theory}%
  \label{fig:opsem}
\end{figure}

\subsection{Judgments from a value type system}

A \emph{value type system} specifies the values that are names for types and the values that each such type
classifies. For practical purposes, it useful to first introduce \emph{candidate value type systems} and then
impose additional conditions under which a candidate is an actual type system.

\begin{defi}
  A \emph{candidate value type system} $\Gt$ is a quaternary relation $\Gt(\GPS,V,V',\Gf)$ ranging over contexts
  $\GPS$, values $V,V'$ in context $\GPS$, and binary relations $\Gf$ on values in context $\GPS$.
\end{defi}

We read an instance $\Gt(\GPS,V,V',\Gf)$ of the relation as specifying that (1) the values $V$ and $V'$ are
equal types in context $\GPS$ and that (2) these type names stand for the relation $\Gf$: values $W$ and $W'$
are equal elements of $V$ (likewise $V'$) in context $\GPS$ when $\Gf(W,W')$ holds.

Given a candidate value type system, we derive \emph{candidate judgments} extending the defining relations to
non-value terms. In~\cite{allen87}, a term is a type (resp.\ well-typed) when it evaluates to a type value
(resp.\ well-typed value). In a setting with interval variables, it becomes necessary to require a stronger
``coherent evaluation'' condition: to be well-typed, a term must not merely evaluate to a well-typed value,
but do so in a way that interacts in a sensible way with interval substitutions. First, we define
``incoherent'' extensions of value type systems and terms to terms.

\begin{defi}
  Given a candidate value type system, we write $\syseval{\Gt}(\GPS,A,A',\Gf)$ for (possibly non-value) terms $A,A'$
  to mean that $A \evals V$ and $A' \evals V'$ for some $V,V'$ with $\Gt(\GPS,V,V',\Gf)$. Given a relation
  $\Gf$ on values, we define a relation $\releval{\Gf}$ on terms: $\releval{\Gf}(M,M')$ holds when
  $M \evals V$ and $M' \evals V'$ for some $V,V'$ with $\Gf(V,V')$.
\end{defi}

To cut down to the coherently well-behaved types and terms, we introduce a notion of \emph{$\GPS$-relation}, a
family of relations indexed by the substitutions into $\GPS$.

\begin{defi}
  A $\GPS$-relation $\Ga$ is a family of binary relations $\Ga_\Gps$, indexed by substitutions
  $\dctxsubst{\GPS'}{\Gps}{\GPS}$ into $\GPS$ and where each $\Ga_\Gps$ relates terms in context
  $\GPS'$. Given a $\GPS$-relation $\Ga$ and $\dctxsubst{\GPS'}{\Gps}{\GPS}$, we define a $\GPS'$-relation
  $\usubstdims{\Ga}{\Gps}$ by ${(\usubstdims{\Ga}{\Gps})}_{\Gps'} \eqdef \Ga_{\Gps\Gps'}$.
\end{defi}

We now define the coherent \emph{candidate judgments}: $\candeqtype{\GPS}{A}{A'}{\Ga}{\Gt}$, which asserts
that $A$ and $A'$ coherently evaluate to equal type names standing for the $\GPS$-relation $\Ga$, and
$\candeqtm{\GPS}{M}{M'}{\Ga}$, which asserts that $M$ and $M'$ coherently evaluate to values equal in $\Ga$.

\begin{defi}
  We define the candidate judgments as follows.
  \begin{itemize}[label=$\triangleright$]
  \item $\candeqtype{\GPS}{A}{A'}{\Ga}{\Gt}$ holds when for every $\dctxsubst{\GPS_1}{\Gps_1}{\GPS}$ and
    $\dctxsubst{\GPS_2}{\Gps_2}{\GPS_1}$, we have
    \begin{enumerate}
    \item $\usubstdims{A}{\Gps_1} \evals A_1$ and $\usubstdims{A'}{\Gps_1} \evals A_1'$ for some $A_1,A_1'$,
    \item there is some $\Gf$ such that $\syseval{\tau}(\GPS_2,-,-,\Gf)$ relates
      $(\usubstdims{A_1}{\Gps_2},\usubstdims{A}{\Gps_1\Gps_2})$ and its reverse,
      $(\usubstdims{A_1'}{\Gps_2},\usubstdims{A'}{\Gps_1\Gps_2})$ and its reverse, and
      $(\usubstdims{A_1}{\Gps_2}, \usubstdims{A'_1}{\Gps_2})$,
    \end{enumerate}
    and $\Ga$ is a $\GPS$-relation on values such that
    $\syseval{\Gt}(\GPS',\usubstdims{A}{\Gps},\usubstdims{A'}{\Gps},\Ga_\Gps)$ for all
    $\dctxsubst{\GPS'}{\Gps}{\GPS}$.
  \item $\candeqtm{\GPS}{M}{M'}{\Ga}$ holds when for every $\dctxsubst{\GPS_1}{\Gps_1}{\GPS}$ and
    $\dctxsubst{\GPS_2}{\Gps_2}{\GPS_1}$, we have
    \begin{enumerate}
    \item $\usubstdims{M}{\Gps_1} \evals M_1$ and $\usubstdims{M'}{\Gps_1} \evals M_1'$ for some $M_1,M_1'$,
    \item $\releval{(\Ga_{\Gps_1\Gps_2})}$ relates $(\usubstdims{M_1}{\Gps_2},\usubstdims{M}{\Gps_1\Gps_2})$ and
      its reverse, $(\usubstdims{M_1'}{\Gps_2},\usubstdims{M'}{\Gps_1\Gps_2})$ and its reverse, and
      $(\usubstdims{M_1}{\Gps_2}, \usubstdims{M'_1}{\Gps_2})$.
    \end{enumerate}
  \end{itemize}
\end{defi}

\noindent
The conditions in the definition of $\candeqtype{\GPS}{A}{A'}{\Ga}{\Gt}$, for example, ask that we have the
square shown below: whether we apply $\Gps_2$ to $A\Gps_1$ or first evaluate and then apply $\Gps_2$, we get
the same result up to the equality defined by $\syseval{\tau}$.
\[
  \begin{tikzcd}
    A\Gps_1 \ar{d}[left]{-\Gps_2} \ar[Rightarrow]{r} & A_1 \ar{d}[right]{-\Gps_2} \\
    A\Gps_1\Gps_2 \ar[phantom]{r}{\syseval{\tau}} & A_1\Gps_2
  \end{tikzcd}
\]
Note that the candidate judgments are stable under interval substitution by definition: for example, if
$\candeqtm{\GPS}{M}{M'}{\Ga}$, then $\candeqtm{\GPS'}{M\Gps}{M'\Gps}{\Ga\Gps}$ for any
$\dctxsubst{\GPS'}{\Gps}{\GPS}$.

A candidate is a value type system when the typing relation satisfies several additional conditions, which
require that each type names at most one relation, that the type and element relations are partial equivalence
relations, and that any value type is \emph{coherently} a type.

\begin{defi}
  A \emph{value type system} $\Gt$ is a candidate value type system satisfying the following.
  \begin{description}
  \item[Unicity] If $\Gt(\GPS,V,V',\Gf)$ and $\Gt(\GPS,V,V',\Gf')$, then $\Gf = \Gf'$.
  \item[PER] $\Gt(\GPS,-,-,\Gf)$ is a partial equivalence relation (PER) for all $\GPS,\Gf$.
  \item[PER-valuation] If $\Gt(\GPS,V,V',\Gf)$, then $\Gf$ is a PER\@.
  \item[Value-coherence] If $\Gt(\GPS,V,V',\Gf)$, then $\candeqtype{\GPS}{V}{V'}{\Ga}{\Gt}$ for some $\Ga$.
  \end{description}
\end{defi}

Likewise, we will require that the values related by the relations associated to types are in fact coherently
related.

\begin{defi}
  We say a $\GPS$-relation $\Ga$ is \emph{value-coherent} and write $\isvalcoh{\Ga}$ if $\Ga_\Gps(V,V')$
  implies $\candeqtm{\GPS'}{\usubstdims{V}{\Gps}}{\usubstdims{V'}{\Gps}}{\usubstdims{\Ga}{\Gps}}$ for all
  $\Gps$ and $V,V'$.
\end{defi}

Given a value type system, we obtain typing judgments first on closed and then on open terms. For types, we
also distinguish between \emph{pretypes} and \emph{types}, the latter of which are required to support Kan
operations. For the following series of definitions, we fix an ambient value type system $\Gt$.

\begin{defi}
  We define the closed judgments as follows.
  \begin{itemize}[label=$\triangleright$]
  \item $\eqpretype<\GPS>{A}{A'}$ holds when $\candeqtype{\GPS}{A}{A'}{\Ga}{\Gt}$ for some value-coherent
    $\Ga$.
  \item Presupposing $\eqpretype<\GPS>{A}{A}$, $\eqtm<\GPS>{M}{M'}{A}$ holds when
    $\candeqtype{\GPS}{A}{A}{\Ga}{\Gt}$ with $\candeqtm{\GPS}{M}{M'}{\Ga}$.
  \end{itemize}
\end{defi}

\noindent
We define $\wfpretype<\GPS>{A}$ to mean $\eqpretype<\GPS>{A}{A}$, likewise $\wftm<\GPS>{M}{A}$ to mean
$\eqtm<\GPS>{M}{M}{A}$. We will abbreviate future reflexive judgments in this fashion without comment. When we
have $\wfpretype<\GPS>{A}$, we write $\sem{A}$ for the (necessarily unique) value $\GPS$-relation assigned to
$A$ by the value type system.

We now extend the closed judgments to \emph{open judgments}, defined on terms containing arbitrary variables.
We do so by means of a \emph{context instantiation judgment} $\eqlist<\GPS>{\Gg}{\Gg'}{\GG}$, which specifies
the ways a general context $\GG$ may be instantiated by closed terms over $\GPS$.

\begin{defi}
  We define the context instantiations $\eqlist<\GPS>{\Gg}{\Gg'}{\GG}$ inductively as follows.
  \begin{itemize}[label=$\triangleright$]
  \item $\eqlist<\GPS>{\substnil}{\substnil}{\ctxnil}$.
  \item $\eqlist<\GPS>{(\substsnoc{\Gg}{M}[a])}{(\substsnoc{\Gg'}{M'}[a])}{(\ctxsnoc{\GG}[a]{A})}$ when
    $\eqlist<\GPS>{\Gg}{\Gg'}{\GG}$ and $\eqtm<\GPS>{M}{M'}{A\Gg}$.
  \item $\eqlist<\GPS>{(\substpdim{\Gg}{r}[x])}{(\substpdim{\Gg}{r}[x])}{(\ctxpdim{\GG}[x])}$ when
    $\eqlist<\GPS>{\Gg}{\Gg'}{\GG}$ and $\wfpdim<\GPS>{r}$.
  \item $\eqlist<\GPS>{(\substbdim{\Gg}{\bmr}[\bmx])}{(\substbdim{\Gg}{\bmr}[\bmx])}{(\ctxbdim{\GG}[x])}$ when
    $\wfbdim<\GPS>{\bmr}$ and $\eqlist<\ctxres{\GPS}{\bmr}>{\Gg}{\Gg'}{\GG}$.
  \item $\eqlist<\GPS>{\Gg}{\Gg'}{(\ctxcst{\GG}{\Gx})}$ when $\eqlist<\GPS>{\Gg}{\Gg'}{\GG}$ and $\Gx\Gg$ is
    true.
  \end{itemize}
\end{defi}

\noindent
The open type and term judgments are then defined to hold when their closed instantiations hold.

\begin{defi}
  We define the open judgments as follows.
  \begin{itemize}[label=$\triangleright$]
  \item $\eqpretype[\GG]{A}{A'}$ holds when
    $\eqpretype<\GPS>{\usubstlist{A}{\Gg}}{\usubstlist{A'}{\Gg'}}$ for
    all $\eqlist<\GPS>{\Gg}{\Gg'}{\GG}$.
  \item $\eqtm[\GG]{M}{M'}{A}$ holds when
    $\eqtm<\GPS>{\usubstlist{M}{\Gg}}{\usubstlist{M'}{\Gg'}}{\usubstlist{A}{\Gg}}$
    for all $\eqlist<\GPS>{\Gg}{\Gg'}{\GG}$.
  \end{itemize}
\end{defi}

\noindent
We note that, in contrast, we define the open \emph{interval} judgments without reference to the terms in the
context $\GG$. It is therefore not the case that, for example, $\eqpdim[\ctxsnoc{}[v]{\bot}]{0}{1}$; interval
judgments are prior to term judgments.

\begin{defi}
  The judgment $\wfpdim[\GG]{r}$ is defined to hold when either $r \in \{0,1\}$ or $(\ctxpdim{}[x]) \in \GG$;
  an equality $\eqpdim[\GG]{r}{s}$ is defined to hold when $\wfpdim[\GG]{r,s}$ are in the equivalence relation
  closure of the constraints appearing in $\GG$. The judgments $\eqbdim[\GG]{\bmr}{\bms}$ and
  $\eqcst[\GG]{\Gx}{\Gx'}$ are defined likewise.
\end{defi}

\begin{defi}
  We define the well-formed contexts inductively.
  \begin{itemize}[label=$\triangleright$]
  \item $\eqctx{\ctxnil}{\ctxnil}$.
  \item $\eqctx{(\ctxsnoc{\GG}[a]{A})}{(\ctxsnoc{\GG'}[a]{A'})}$ when $\eqctx{\GG}{\GG'}$ and
    $\eqpretype[\GG]{A}{A'}$.
  \item $\eqctx{(\ctxpdim{\GG}[x])}{(\ctxpdim{\GG'}[x])}$ when $\eqctx{\GG}{\GG'}$.
  \item $\eqctx{(\ctxbdim{\GG}[x])}{(\ctxbdim{\GG'}[x])}$ when $\eqctx{\GG}{\GG'}$.
  \item $\eqctx{(\ctxcst{\GG}{\Gx})}{(\ctxcst{\GG'}{\Gx})}$ when
    $\eqctx{\GG}{\GG'}$ and $\eqcst[\GG]{\Gx}{\Gx'}$.
  \end{itemize}
\end{defi}

\noindent
A pretype $A$ is a \emph{(Kan) type} when it supports the Kan operations, that is, when the operators $\coe$
and $\hcomp$ are well-typed at $A$ and satisfy the necessary equations.

\begin{defi}[Kan types]
  Presupposing $\eqpretype<\GPS>{A}{A'}$, we say $\eqtype<\GPS>{A}{A'}$ when the following conditions hold.
  \begin{itemize}[label=$\triangleright$]
  \item For any $\dctxsubst{(\ctxpdim{\GPS'}[x])}{\Gps}{\GPS}$, if $\wfpdim<\GPS'>{r,s}$ and
    $\eqtm<\GPS'>{M}{M'}{\usubstdim{A\Gps}{r}{x}}$, then
    \begin{itemize}
    \item $\eqtm<\GPS'>{\coe{x.A\Gps}{r}{s}{M}}{\coe{x.A'\Gps}{r}{s}{M'}}{\usubstdim{A\Gps}{s}{x}}$,
    \item $\eqtm<\GPS'>{\coe{x.A\Gps}{r}{r}{M}}{M}{\usubstdim{A\Gps}{r}{x}}$,
    \end{itemize}
  \item For any $\dctxsubst{\GPS'}{\Gps}{\GPS}$, if $\wfpdim<\GPS'>{r,s}$, $n \in \BN$, $\wfcst<\GPS'>{\Gx_i}$
    for all $i < n$, and
    \begin{itemize}
    \item $\eqtm<\GPS'>{M}{M'}{A\Gps}$
    \item $\eqtm<\ctxpdim{\GPS'}[x]>{N_i}{N'_j}{A\Gps}$ for all $i,j < n$,
    \item $\eqtm<\GPS'>{M}{\usubstdim{N_i}{r}{x}}{A\Gps}$ for all $i < n$,
    \end{itemize}
    then
    \begin{itemize}
    \item $\eqtm<\GPS'>{\hcomp{A\Gps}{r}{s}{M}{\sys{\Gx_i}{x.N_i}}}{\hcomp{A'\Gps}{r}{s}{M'}{\sys{\Gx_i}{x.N_i'}}}{A\Gps}$,
    \item $\eqtm<\GPS'>{\hcomp{A\Gps}{r}{s}{M}{\sys{\Gx_i}{x.N_i}}}{\usubstdim{N_i}{s}{x}}{A\Gps}$ if $\Gx_i$ is true,
    \item $\eqtm<\GPS'>{\hcomp{A\Gps}{r}{r}{M}{\sys{\Gx_i}{x.N_i}}}{M}{A\Gps}$.
    \end{itemize}
  \end{itemize}
\end{defi}

\noindent
The extension of the type judgment to open terms is defined as for the pretype judgment: $\eqtype[\GG]{A}{A'}$
holds when $\eqtype<\GPS>{\usubstlist{A}{\Gg}}{\usubstlist{A'}{\Gg'}}$ for all
$\eqlist<\GPS>{\Gg}{\Gg'}{\GG}$.

We may also define the open substitution judgment following the pattern of the
instantiation judgment.

\begin{defi}%
  \label{def:substitutions}
  We define the substitutions $\eqsubst[\GG]{\Gg}{\Gg'}{\GG}$ inductively as follows.
  \begin{itemize}[label=$\triangleright$]
  \item $\eqsubst[\GG]{\substnil}{\substnil}{\ctxnil}$.
  \item $\eqsubst[\GG]{(\substsnoc{\Gg}{M}[a])}{(\substsnoc{\Gg'}{M'}[a])}{(\ctxsnoc{\GG}[a]{A})}$ when
    $\eqsubst[\GG]{\Gg}{\Gg'}{\GG}$ and $\eqtm[\GG]{M}{M'}{A\Gg}$.
  \item $\eqsubst[\GG]{(\substpdim{\Gg}{r}[x])}{(\substpdim{\Gg}{r'}[x])}{(\ctxpdim{\GG}[x])}$ when
    $\eqsubst[\GG]{\Gg}{\Gg'}{\GG}$ and $\eqpdim[\GG]{r}{r'}$.
  \item $\eqsubst[\GG]{(\substbdim{\Gg}{\bmr}[\bmx])}{(\substbdim{\Gg}{\bmr'}[\bmx])}{(\ctxbdim{\GG}[x])}$ when
    $\eqbdim[\GG]{\bmr}{\bmr'}$ and $\eqsubst<\ctxres{\GG}{\bmr}>{\Gg}{\Gg'}{\GG}$.
  \item $\eqsubst[\GG]{\Gg}{\Gg'}{(\ctxcst{\GG}{\Gx})}$ when $\eqsubst[\GG]{\Gg}{\Gg'}{\GG}$ and $\Gx\Gg$ is
    true.
  \end{itemize}
\end{defi}

\noindent
Now that we have laid out the extrapolation of open judgments from a value type system, it remains to
construct a particular type system that will validate the inference rules we presented in
\cref{sec:cubical,sec:parametric}.

\subsection{Constructing a value type system}%
\label{sec:computational:construct}

We obtain a value type system by a fixed-point construction, first defining the least candidate value type system closed
under our desired type formers and then showing that it constitutes a value type system. To start, we define the
pieces corresponding to each type former. Relative to~\cite{angiuli19}, the novelties here are the
$\tybridge$- and $\tygel$-types.
\[
  \begin{array}{lcl}
    \sysbridge{\Gt} \eqdef \\
    \;\{(\GPS,\tybridge{\bmx.A}{M_0}{M_1},\tybridge{\bmx.A'}{M'_0}{M'_1},\Gf) \mid \\
    \;\;\; \exists \Ga. \ \candeqtype{\ctxbdim{\GPS}[x]}{A}{A'}{\Ga}{\Gt} \land \isvalcoh{\Ga} \\
    \;\;\; { } \land (\forall \Ge \in \{0,1\}.\ \candeqtm{\GPS}{M_\Ge}{M'_\Ge}{\usubstdim{\Ga}{\bm\Ge}{\bmx}}) \\
    \;\;\; { } \land \Gf = \{(\tmblam[x]{M},\tmblam[x]{M'}) \mid \candeqtm{\ctxbdim{\GPS}[x]}{M}{M'}{\Ga} \land \forall \Ge.\ \candeqtm{\GPS}{\usubstdim{M}{\bm\Ge}{\bmx}}{M_\Ge}{\usubstdim{\Ga}{\bm\Ge}{\bmx}} \} \} \\
    ~\\
    \sysgel{\Gt} \eqdef \\
    \;\{(\GPS,\tygel{\bmx}{A_0}{A_1}{a_0.a_1.R},\tygel{\bmx}{A'_0}{A'_1}{a_0.a_1.R'},\Gf) \mid \\
    \;\;\; \exists \Ga^0,\Ga^1,\Gb^{(-,-,-,-,-)}. \\
    \;\;\; (\forall \Ge.\ \candeqtype{\ctxres{\GPS}{\bmx}}{A_\Ge}{A_\Ge'}{\Ga}{\Gt} \land \isvalcoh{\Ga^\Ge}) \\
    \;\;\; { } \land (\forall \dctxsubst{\GPS'}{\Gps}{(\ctxres{\GPS}{\bmx})}.\ \forall M_0,M_1,M_0',M_1'.\ (\forall \Ge.\ \Ga^\Ge_\Gps(M_\Ge,M_\Ge')) \implies { } \\
    \;\;\; \quad\;\;  \candeqtype{\GPS'}{\usubst{R}{M_0,M_1}{a_0,a_1}}{\usubst{R'}{M'_0,M'_1}{a_0,a_1}}{\Gb^{(\Gps,M_0,M_1,M_0',M_1')}}{\Gt} \\
    \;\;\; \quad\;\; { } \land \isvalcoh{\Gb^{(\Gps,M_0,M_1,M_0',M_1')}})\\
    \;\;\; { } \land \Gf = \{(\tmgel{\bmx}{M_0}{M_1}{P},\tmgel{\bmx}{M'_0}{M'_1}{P'}) \mid \\
    \;\;\; \quad \quad \quad\quad  \forall \Ge. (\candeqtm{\ctxres{\GPS}{\bmx}}{M_\Ge}{M'_\Ge}{\Ga^\Ge}) \land \candeqtm{\ctxres{\GPS}{\bmx}}{P}{P'}{\Gb^{(\id,M_0,M_1,M_0',M_1')}} \} \}
  \end{array}
\]
Next, we have an operator on candidate value type systems that applies one level of type formers.
\[
  \begin{array}{lcl}
    K(\Gt) &\eqdef& \sysbridge{\Gt} \cup \sysgel{\Gt} \cup \cdots
  \end{array}
\]
Finally, we obtain a least fixed-point $\Gt_0$ of this operator by the Knaster-Tarski fixed-point theorem~\cite[2.35]{davey02}. It is tedious but straightforward to check that this candidate value type system is in
fact a value type system~\cite[Lemma 4.8]{angiuli19}. To construct a value type system with a universe, we can
repeat the fixed-point construction with the addition of a type $\tyuniv$ interpreted by the relation $\Gt_0$,
producing a new type system $\Gt_1$ that is closed under the same type formers as $\Gt_0$ but also contains
$\Gt_0$ as a universe. This can be repeated further to produce a hierarchy of value type systems
$\Gt_0 \subseteq \Gt_1 \subseteq \Gt_2 \subseteq \cdots$ each containing its predecessors as universes; for
our purposes, a single universe is sufficient.

As an immediate consequence of the way the typing judgments are defined, we
have a \emph{canonicity} theorem: any closed well-typed term is guaranteed to evaluate to a value of that
type. In particular, any closed term of natural number type evaluates to a numeral.

\subsection{Building up inference rules}%
\label{sec:computational-rules}

With a value type system in hand, it remains to verify that the judgments are closed under the inference rules
introduced in \cref{sec:cubical,sec:parametric}. We go through the typing rules for $\tygel$-types in
detail. The rules for $\tybridge$-types are simpler to verify, as the reduction rules are all ``cubically
stable'': they do not depend on the status of any interval term. (In comparison, $\tmgel{\bmr}{M_0}{M_1}{P}$
may be a value or step depending on whether $\bmr$ is a variable or constant.) The rules for $\tmextent$ do
involve unstable transitions, but require no ideas that are not present in the proofs for $\tygel$-types; in
particular, the $\hcomp$ reduction for $\tygel$ involves $\tmextent$-like variable capture. The reader may see~\cite{cavallo19a} for complete proofs of these results.

We rely on the following five lemmas to work with the candidate judgments. These are rephrasings of Lemmas
A.2, A.3, and A.5 from~\cite{chtt-iv}; each follows straightforwardly by unfolding definitions.

\begin{lem}[Coherent type value]\label{lem:coherent-type-value}
  Suppose $A,A'$ are terms. If for every $\dctxsubst{\GPS'}{\GPS}{\GPS}$, either
  $\Gt(\GPS',A\Gps,A'\Gps,\Ga_\Gps)$ or $\candeqtype{\GPS'}{A\Gps}{A'\Gps}{\Ga\Gps}{\Gt}$, then
  $\candeqtype{\GPS}{A}{A'}{\Ga}{\Gt}$.
\end{lem}

\begin{lem}[Coherent term value]\label{lem:coherent-term-value}
  Suppose $\candwftype{\GPS}{A}{\Ga}{\Gt}$ and $M,M'$ are terms. If for every $\dctxsubst{\GPS'}{\GPS}{\GPS}$,
  either $\Ga_\Gps(M\Gps,M'\Gps)$ or $\candeqtm{\GPS'}{M\Gps}{M'\Gps}{\Ga\Gps}$, then
  $\candeqtm{\GPS}{M}{M'}{\Ga}$.
\end{lem}

\begin{lem}[Coherent type expansion]\label{lem:coherent-type-expansion}
  Suppose $A$ is a term and ${(A_\Gps)}_{\dctxsubst{\GPS'}{\Gps}{\GPS}}$ is a family of terms such that
  $A\Gps \msteps A_\Gps$ and $\candeqtype{\GPS'}{A_\Gps}{A_\id\Gps}{\Ga\Gps}{\Gt}$ for all
  $\dctxsubst{\GPS'}{\Gps}{\GPS}$. Then $\candeqtype{\GPS}{A}{A_\id}{\Ga}{\Gt}$.
\end{lem}

\begin{lem}[Coherent term expansion]\label{lem:coherent-term-expansion}
  Suppose $\candwftype{\GPS}{A}{\Ga}{\Gt}$, $M$ is a term, and ${(M_\Gps)}_{\dctxsubst{\GPS'}{\Gps}{\GPS}}$ is a
  family of terms such that $M\Gps \msteps M_\Gps$ and $\candeqtm{\GPS'}{M_\Gps}{M_\id\Gps}{\Ga\Gps}$ for all
  $\dctxsubst{\GPS'}{\Gps}{\GPS}$. Then $\candeqtm{\GPS'}{M}{M_\id}{\Ga}$.
\end{lem}

\begin{lem}[Evaluation]\label{lem:evaluation}
  Suppose $\eqtm<\GPS>{M}{M'}{A}$. Then $M \evals V$ and $M' \evals V'$ with $\eqtm<\GPS>{M = V}{V' = M'}{A}$.
\end{lem}

We now check the rules for $\tygel$-types as presented in \cref{fig:gel}. We prove that each rule holds when
the ambient context is an arbitrary interval context $\GPS$. The open rules---for an arbitrary context
$\GG$---then follow mechanically, as the open type and term judgments are defined by their closed
instantiations.

It is convenient to prove the boundary reduction equations for a type or term former \emph{before} the general
introduction rule; for example, we show first $\eqpretype{\tygel{\bm\Ge}{A_0}{A_1}{a_0.a_1.R}}{A_\Ge}$ and
then $\wfpretype{\tygel{\bmr}{A_0}{A_1}{a_0.a_1.R}}$.

\begin{rul}[\rulename{Gel-Form-$\partial$}]\label{rule:gel-formation-boundary}
  For any $\Ge \in \{0,1\}$, $\wfpretype<\GPS>{A_\Ge}$, and terms $A_{1-\Ge}$, $R$, we have
  $\eqpretype<\GPS>{\tygel{\bm\Ge}{A_0}{A_1}{a_0.a_1.R}}{A_\Ge}$.
\end{rul}
\begin{proof}
  By Lemma~\ref{lem:coherent-type-expansion}, taking $A_\Gps \eqdef A_\Ge\Gps$: we have
  $\tygel{\bm\Ge}{A_0}{A_1}{a_0.a_1.R}\Gps \steps A_\Gps$ and
  $\candeqtype{\GPS'}{A_\Ge\Gps}{A_\Ge\Gps}{\sem{A_\Ge}\Gps}{\Gt}$ for all $\Gps$.
\end{proof}

As described above, this ``closed'' principle implies the open rule. Given $\wfctx{\GG}$ and
$\wfpretype[\GG]{A_\Ge}$, we have by definition that $\eqpretype<\GPS>{A_\Ge\Gg}{A_\Ge\Gg'}$ for all
$\eqlist<\GPS>{\Gg}{\Gg'}{\GG}$. Thus $\eqpretype<\GPS>{\tygel{\bm\Ge}{A_0}{A_1}{a_0.a_1.R}\Gg}{A_\Ge\Gg'}$
for all such instantiations by the rule just proven, which means that
$\eqpretype[\GG]{\tygel{\bm\Ge}{A_0}{A_1}{a_0.a_1.R}}{A_\Ge}$.

The following lemma gets us part of the way to the formation rule. We also need that the relation for
$\tygel$-types is value-coherent and supports the Kan operations; we will return to these later.

\begin{lem}[Gel formation candidate]\label{lem:gel-formation-candidate}
  If we have $\wfbdim<\GPS>{\bmr}$, $\eqpretype<\ctxres{\GPS}{\bmr}>{A_\Ge}{A'_\Ge}$ for $\Ge \in \{0,1\}$, and
  $\eqpretype[\ctxsnoc{\ctxsnoc{\ctxres{\GPS}{\bmr}}[a_0]{A_0}}[a_1]{A_1}]{R}{R'}$, then
  $\candeqtype{\GPS}{\tygel{\bmr}{A_0}{A_1}{a_0.a_1.R}}{\tygel{\bmr}{A'_0}{A'_1}{a_0.a_1.R'}}{\Gg}{\Gt}$
  with $\Gg$ defined on $\dctxsubst{\GPS'}{\Gps}{\GPS}$ as follows.
  \[
    \Gg_\Gps \eqdef \left\{
      \begin{array}{ll}
        \{(\tmgel{\bmx}{M_0}{M_1}{P},\tmgel{\bmx}{M'_0}{M'_1}{P'}) \mid \\
        \;\;\forall \Ge. (\eqtm<\ctxres{\GPS'}{\bmx}>{M_\Ge}{M'_\Ge}{A\Gps}) \\
        \;\; { } \land \eqtm<\ctxres{\GPS'}{\bmx}>{P}{P'}{\usubst{R}{M_0,M_1}{a_0,a_1}} \} , &\text{if $\bmr\Gps = \bmx$} \\
        \Ga^\Ge\Gps, &\text{if $\bmr\Gps = \bm\Ge \in \{\bm0,\bm1\}$}
      \end{array}
      \right.
  \]
\end{lem}
\begin{proof}
  By Lemma~\ref{lem:coherent-type-value}. For every $\dctxsubst{\GPS'}{\Gps}{\GPS}$, either
  $\bmr\Gps = \bmx$ for some $\bmx$, in which case we have
  $\Gt(\GPS',\tygel{\bmr}{A_0}{A_1}{a_0.a_1.R}\Gps,\tygel{\bmr}{A_0'}{A_1'}{a_0.a_1.R'}\Gps,\Gg_\Gps)$ by
  definition of the value type system, or $\bmr\Gps = \bm\Ge \in \{\bm0,\bm1\}$, in which case we have
  $\tygel{\bmr}{A_0}{A_1}{a_0.a_1.R}\Gps \sim A_\Ge\Gps \sim A_\Ge'\Gps \sim
  \tygel{\bmr}{A'_0}{A'_1}{a_0.a_1.R'}\Gps$ by way of \rulename{Gel-Form-$\partial$}.
\end{proof}

\begin{rul}[\rulename{Gel-Intro-$\partial$}]\label{rule:gel-introduction-boundary}
  For any $\Ge \in \{0,1\}$, $\wfpretype<\GPS>{A_\Ge}$, and $\wftm<\GPS>{M_\Ge}{A_\Ge}$, and terms
  $M_{1-\Ge}$, $P$, we have $\eqtm<\GPS>{\tmgel{\bm\Ge}{M_0}{M_1}{P}}{M_\Ge}{A_\Ge}$.
\end{rul}
\begin{proof}
  By Lemma~\ref{lem:coherent-term-expansion}, taking $M_\Gps \eqdef M_\Ge\Gps$.
\end{proof}

\begin{rul}[\rulename{Gel-Intro}]\label{rule:gel-introduction}
  If we have $\wfbdim<\GPS>{\bmr}$, $\eqtm<\ctxres{\GPS}{\bmr}>{M_\Ge}{M'_\Ge}{A_\Ge}$ for $\Ge \in \{0,1\}$,
  $\eqpretype[\ctxsnoc{\ctxsnoc{\ctxres{\GPS}{\bmr}}[a_0]{A_0}}[a_1]{A_1}]{R}{R'}$, and
  $\eqtm<\ctxres{\GPS}{\bmr}>{P}{P'}{\usubst{R}{M_0,M_1}{a_0,a_1}}$, then
  $\candeqtm{\GPS}{\tmgel{\bmr}{M_0}{M_1}{P}}{\tmgel{\bmr}{M'_0}{M'_1}{P'}}{\Gg}$ for $\Gg$ as in the statement of
  Lemma~\ref{lem:gel-formation-candidate}.
\end{rul}
\begin{proof}
  By Lemma~\ref{lem:coherent-term-value}, proceeding as in Lemma~\ref{lem:gel-formation-candidate} by cases on
  $\bmr\Gps$ for each $\Gps$: we use the definition of $\Gg$ when $\bmr\Gps$ is a variable and
  \rulename{Gel-Intro-$\partial$} when $\bmr\Gps$ is a constant.
\end{proof}

\begin{lem}[Gel formation pretype]\label{lem:gel-formation-pretype}
  If we have $\wfbdim<\GPS>{\bmr}$, $\eqpretype<\ctxres{\GPS}{\bmr}>{A_\Ge}{A'_\Ge}$ for $\Ge \in \{0,1\}$, and
  $\eqpretype[\ctxsnoc{\ctxsnoc{\ctxres{\GPS}{\bmr}}[a_0]{A_0}}[a_1]{A_1}]{R}{R'}$, then
  $\eqpretype<\GPS>{\tygel{\bmr}{A_0}{A_1}{a_0.a_1.R}}{\tygel{\bmr}{A'_0}{A'_1}{a_0.a_1.R'}}$.
\end{lem}
\begin{proof}
  A combination of Lemma~\ref{lem:gel-formation-candidate} and \rulename{Gel-Intro}, the latter of
  which shows that the relation for $\tygel$ is value-coherent.
\end{proof}

\begin{rul}[\rulename{Gel-$\beta$}]\label{rule:gel-beta}
  If $\wftm<\ctxbdim{\GPS}[x]>{P}{\usubst{R}{M_0,M_1}{a_0,a_1}}$, then
  \[\eqtm<\GPS>{\tmungel{\bmx.\tmgel{\bmx}{M_0}{M_1}{P}}}{P}{\usubst{R}{M_0,M_1}{a_0,a_1}}.\]
\end{rul}
\begin{proof}
  By Lemma~\ref{lem:coherent-term-expansion}: we have
  $\tmungel{\bmx.\tmgel{\bmx}{M_0}{M_1}{P}}\Gps \steps P\Gps$ for all $\Gps$.
\end{proof}

\begin{rul}[\rulename{Gel-Elim}]\label{rule:gel-elimination}
  If $\wfpretype<\GPS>{A_\Ge}$ for $\Ge \in \{0,1\}$,
  $\wfpretype[\ctxsnoc{\ctxsnoc{\GPS}[a_0]{A_0}}[a_1]{A_1}]{R}$, and
  $\eqtm<\ctxbdim{\GPS}[x]>{Q}{Q'}{\tygel{\bmx}{A_0}{A_1}{R}}$, then we have the following.
  \[
    \eqtm<\GPS>{\tmungel{\bmx.Q}}{\tmungel{\bmx.Q'}}{\usubst{R}{\usubstdim{Q}{\bm0}{\bmx},\usubstdim{Q}{\bm1}{\bmx}}{a_0,a_1}}
  \]
\end{rul}
\begin{proof}
  For every $\dctxsubst{\GPS'}{\Gps}{\GPS}$, we have by Lemma~\ref{lem:evaluation} that $Q\Gps \evals Q_\Gps$
  and $Q'\Gps \evals Q'_\Gps$ for some
  $\eqtm<\ctxbdim{\GPS'}[x]>{Q\Gps = Q_\Gps}{Q'_\Gps =
    Q'\Gps}{\tygel{\bmx}{A_0\Gps}{A_1\Gps}{a_0.a_1.R\Gps}}$. By definition of the relation for $\tygel$-types, we
  have $Q_\Gps = \tmgel{\bmx}{M_{0,\Gps}}{M_{1,\Gps}}{P_\Gps}$ and
  $Q_\Gps' = \tmgel{\bmx}{M_{0,\Gps}'}{M_{1,\Gps}'}{P_\Gps'}$ for some terms such that
  $\eqtm<\GPS'>{P_\Gps}{P'_\Gps}{\usubst{R\Gps}{M_{0,\Gps},M_{1,\Gps}}{a_0,a_1}}$. By
  \rulename{Gel-Intro-$\partial$} and functionality of $R$, it follows that also
  $\eqtm<\GPS'>{P_\Gps}{P'_\Gps}{\usubst{R\Gps}{\usubstdim{Q}{\bm0}{\bmx}\Gps,\usubstdim{Q}{\bm1}{\bmx}\Gps}{a_0,a_1}}$. We
  have $\tmungel{\bmx.Q}\Gps\msteps P_\Gps$ for each $\Gps$, thus
  $\eqtm<\GPS>{\tmungel{\bmx.Q}}{P_\id}{\usubst{R}{\usubstdim{Q}{\bm0}{\bmx},\usubstdim{Q}{\bm1}{\bmx}}{a_0,a_1}}$
  by Lemma~\ref{lem:coherent-term-expansion}; likewise,
  $\eqtm<\GPS>{\tmungel{\bmx.Q'}}{P'_\id}{\usubst{R}{\usubstdim{Q}{\bm0}{\bmx},\usubstdim{Q}{\bm1}{\bmx}}{a_0,a_1}}$. We
  conclude by transitivity that
  $\eqtm<\GPS>{\tmungel{\bmx.Q} = P_\id}{P'_\id =
    \tmungel{\bmx.Q'}}{\usubst{R}{\usubstdim{Q}{\bm0}{\bmx},\usubstdim{Q}{\bm1}{\bmx}}{a_0,a_1}}$.
\end{proof}

\begin{rul}[\rulename{Gel-$\eta$}]\label{rule:gel-eta}
  If $\wfpretype<\ctxres{\GPS}{\bmr}>{A_\Ge}$ for $\Ge \in \{0,1\}$,
  $\wfpretype[\ctxsnoc{\ctxsnoc{\ctxres{\GPS}{\bmr}}[a_0]{A_0}}[a_1]{A_1}]{R}$, and
  $\wftm<\ctxbdim{\ctxres{\GPS}{\bmr}}[x]>{Q}{\tygel{\bmx}{A_0}{A_1}{a_0.a_1.R}}$, then we have the following.
  \[
    \eqtm<\GPS>{\usubstdim{Q}{\bmr}{\bmx}}{\tmgel{\bmr}{\usubstdim{Q}{\bm0}{\bmx}}{\usubstdim{Q}{\bm1}{\bmx}}{\tmungel{\bmx.Q}}}{\tygel{\bmr}{A_0}{A_1}{a_0.a_1.R}}
  \]
\end{rul}
\begin{proof}
  By Lemma~\ref{lem:evaluation}, we have
  $\eqtm<\ctxbdim{\ctxres{\GPS}{\bmr}}[x]>{Q}{V}{\tygel{\bmx}{A_0}{A_1}{a_0.a_1.R}}$ for some $Q \evals V$. By
  definition of the relation for $\tygel$-types, we know $V = \tmgel{\bmx}{M_0}{M_1}{P}$ for some suitably-typed
  $M_0$, $M_1$, and $P$. By \rulename{Gel-Intro-$\partial$}, \rulename{Gel-$\beta$}, and \rulename{Gel-Intro},
  we conclude the following.
  \[
    \eqtm<\ctxbdim{\ctxres{\GPS}{\bmr}}[x]>{V}{\tmgel{\bmx}{\usubstdim{V}{\bm0}{\bmx}}{\usubstdim{V}{\bm1}{\bmx}}{\tmungel{\bmx.V}}}{\tygel{\bmx}{A_0}{A_1}{a_0.a_1.R}}
  \]
  We can replace $V$ with $Q$ everywhere in this equation using \rulename{Gel-Intro} and
  \rulename{Gel-Elim}. Substituting $\bmr$ for $\bmx$ then gives the result.
\end{proof}

It only remains to show that $\tygel$-types support the Kan operations. We will go through the proof for
$\hcomp$; the proof for $\coe$ has an identical structure. We will begin by proving reduction lemmas for the
constant and variable cases.

\begin{lem}\label{lem:gel-hcom-boundary}
  Let $\wftype<\GPS>{A_\Ge}$ for some $\Ge \in \{0,1\}$. If $\wfpdim<\GPS>{r,s}$, $n \in \BN$,
  $\wfcst<\GPS>{\Gx_i}$, $\wftm<\GPS>{Q}{A_\Ge}$, $\eqtm<\ctxpdim{\GPS}[y]>{Q_i}{Q_j}{A_\Ge}$ for all
  $i,j < n$, and $\eqtm<\GPS>{Q}{\usubstdim{Q_i}{r}{y}}{A_\Ge}$ for all $i < n$, then
  $\eqtm<\GPS>{\hcomp{\tygel{\bm\Ge}{A_0}{A_1}{a_0.a_1.R}}{r}{s}{Q}{\sys{\xi_i}{y.Q_i}}}{\hcomp{A_\Ge}{r}{s}{Q}{\sys{\xi_i}{y.Q_i}}}{A_\Ge}$.
\end{lem}
\begin{proof}
  By Lemma~\ref{lem:coherent-term-expansion}: every substitution instance of the left-hand side steps to the
  corresponding instance of the right-hand side, which is well-typed because $A_\Ge$ is Kan.
\end{proof}

\begin{lem}\label{lem:gel-hcom-beta}
  Let $\wftype<\GPS>{A_\Ge}$ for $\Ge \in \{0,1\}$ and
  $\wftype[\ctxsnoc{\ctxsnoc{\GPS}[a_0]{A_0}}[a_1]{A_1}]{R}$. Abbreviate
  $G \eqdef \tygel{\bmx}{A_0}{A_1}{a_0.a_1.R}$. For any $\wfpdim<\ctxbdim{\GPS}[x]>{r,s}$, $n \in \BN$,
  $\wfcst<\ctxbdim{\GPS}[x]>{\Gx_i}$, $\wftm<\ctxbdim{\GPS}[x]>{Q}{G}$,
  $\eqtm<\ctxpdim{\ctxbdim{\GPS}[x]}[y]>{Q_i}{Q_j}{G}$ for all $i,j < n$, and
  $\eqtm<\ctxbdim{\GPS}[x]>{Q}{\usubstdim{Q_i}{r}{y}}{G}$ for all $i < n$, we have
  $\eqtm<\ctxbdim{\GPS}[x]>{\hcomp{G}{r}{s}{Q}{\sys{\xi_i}{y.Q_i}}}{\tmgel{\bmx}{M_{0,s}}{M_{1,s}}{P}}{G}$
  where $M_{\Ge,-}$ and $P$ are defined as follows.
  \begin{align*}
    M_{\Ge,y} &\eqdef
  \hcomp{A_\Ge}{r}{y}{\usubstdim{Q}{\bm\Ge}{\bmx}}{\sys{\usubstdim{\xi_i}{\bm\Ge}{\bmx}}{y.\usubstdim{Q_i}{\bm\Ge}{\bmx}}} \\
    P &\eqdef \comp{y.\usubst{R}{M_{0,y},M_{1,y}}{a_0,a_1}}{r}{s}{\tmungel{\bmx.Q}}{\sys{\forall \bmx. \xi_i}{y.\tmungel{\bmx.Q_i}}}
  \end{align*}
\end{lem}
\begin{proof}
  By Lemma~\ref{lem:coherent-term-expansion}. For every $\dctxsubst{\GPS'}{\Gps}{(\ctxbdim{\GPS}[x])}$, we
  have two cases.
  \begin{itemize}[label=$\triangleright$]
  \item $\bmx\Gps = \bm\Ge \in \{\bm0,\bm1\}$. Then
    $\hcomp{G}{r}{s}{Q}{\sys{\xi_i}{y.Q_i}}\Gps \steps \hcomp{A_\Ge}{r}{s}{Q}{\sys{\xi_i}{y.Q_i}}\Gps$, and we
    have
    $\eqtm<\GPS'>{\hcomp{A_\Ge}{r}{s}{Q}{\sys{\xi_i}{y.Q_i}}\Gps}{\tmgel{\bmx}{M_{0,s}}{M_{1,s}}{P}\Gps}{G\Gps}$
    by \rulename{Gel-Intro-$\partial$} and the assumption that $A$ is Kan.
  \item $\bmx\Gps$ is a variable.  Then
    $\hcomp{G}{r}{s}{Q}{\sys{\xi_i}{y.Q_i}}\Gps \steps \tmgel{\bmx}{M_{0,s}}{M_{1,s}}{P}\Gps$, and we have
    $\wftm<\GPS'>{\tmgel{\bmx}{M_{0,s}}{M_{1,s}}{P}\Gps}{G\Gps}$ by \rulename{Gel-Intro-$\partial$},
    \rulename{Gel-Elim}, and the assumption that the $A_\Ge$ and $R$ are Kan. We use here that the capture of
    $\bmx$ by $\tmungel$ in the definition of the reduct commutes with $\Gps$, which relies on the affinity
    of bridge interval substitution. \qedhere
  \end{itemize}
\end{proof}

\begin{rul}[\rulename{Gel-Form}]%
  \label{rule:gel-form}
  If $\wfbdim<\GPS>{\bmr}$, $\eqtype<\ctxres{\GPS}{\bmr}>{A_\Ge}{A'_\Ge}$ for each $\Ge \in \{0,1\}$, and
  $\eqtype[\ctxsnoc{\ctxsnoc{\ctxres{\GPS}{\bmr}}[a_0]{A_0}}[a_1]{A_1}]{R}{R'}$, then we have the following.
  \[
    \eqtype<\GPS>{\tygel{\bmr}{A_0}{A_1}{a_0.a_1.R}}{\tygel{\bmr}{A'_0}{A'_1}{a_0.a_1.R'}}
  \]
\end{rul}
\begin{proof}
  We must check that $\tygel$ supports the Kan operations. We give the proof for $\hcomp$.  Abbreviate
  $G \eqdef \tygel{\bmr}{A_0}{A_1}{a_0.a_1.R}$ and $G' \eqdef \tygel{\bmr}{A_0'}{A_1'}{a_0.a_1.R'}$. Let
  $\dctxsubst{\GPS'}{\Gps}{\GPS}$, $\wfpdim<\GPS'>{r,s}$, $n \in \BN$, $\wfcst<\GPS'>{\Gx_i}$ for all $i < n$,
  $\eqtm<\GPS'>{Q}{Q'}{G\Gps}$, $\eqtm<\ctxpdim{\GPS'}[y]>{Q_i}{Q'_j}{G\Gps}$ for all $i,j < n$, and
  $\eqtm<\GPS'>{Q}{\usubstdim{Q_i}{r}{y}}{G\Gps}$ for all $i < n$ be given. If $\bmr\Gps$ is a constant,
  then we simply apply \rulename{Gel-Form-$\partial$} and Lemma~\ref{lem:gel-hcom-boundary} everywhere.

  If $\bmr\Gps$ is a variable $\bmx$, then
  $\eqtm<\GPS'>{\hcomp{G}{r}{s}{Q}{\sys{\xi_i}{y.Q_i}}}{\tmgel{\bmx}{M_{0,s}}{M_{1,s}}{P}}{G\Gps}$
  and
  $\eqtm<\GPS'>{\hcomp{G'}{r}{s}{Q'}{\sys{\xi_i}{y.Q_i'}}}{\tmgel{\bmx}{M'_{0,s}}{M'_{1,s}}{P'}}{G'\Gps}$
  as defined in Lemma~\ref{lem:gel-hcom-beta}. Then we have the following.
  \begin{itemize}[label=$\triangleright$]
  \item
    $\eqtm<\GPS'>{\hcomp{G}{r}{s}{Q}{\sys{\xi_i}{y.Q_i}}}{\hcomp{G'}{r}{s}{Q'}{\sys{\xi_i}{y.Q_i'}}}{G\Gps}$
    follows from the fact that
    $\eqtm<\GPS'>{\tmgel{\bmx}{M_{0,s}}{M_{1,s}}{P}}{\tmgel{\bmx}{M'_{0,s}}{M'_{1,s}}{P'}}{G\Gps}$, which
    holds by \rulename{Gel-Intro-$\partial$}, \rulename{Gel-Elim}, and the assumption that the $A_\Ge$ and $R$
    are Kan.
  \item $\eqtm<\GPS'>{\hcomp{G}{r}{s}{Q}{\sys{\xi_i}{y.Q_i}}}{\usubstdim{Q_i}{s}{y}}{G\Gps}$ if $\Gx_i$ is
    true follows by cases on $\Gx_i$. If $\bmx$ does not occur in $\Gx_i$, then $\forall \bmx. \Gx_i =
    \Gx_i$. It follows by the boundary equations for $\hcomp$ in $A_\Ge$ and $R$ that the composite is equal
    to
    $\usubstdim{\tmgel{\bmx}{\usubstdim{Q_i}{\bm0}{\bmx}}{\usubstdim{Q_i}{\bm1}{\bmx}}{\tmungel{\bmx.Q_i}}}{s}{y}$,
    and this term is equal to $\usubstdim{Q_i}{s}{y}$ by \rulename{Gel-$\eta$}. If $\bmx$ does occur in
    $\Gx_i$, then the constraint must be either $\bmx = \bm0$ or $\bmx = \bm1$, in which case it is
    contradictory that $\Gx_i$ is true.
  \item $\eqtm<\GPS'>{\hcomp{G}{r}{r}{Q}{\sys{\xi_i}{y.Q_i}}}{Q}{G\Gps}$ holds by the corresponding Kan
    equations for the $A_\Ge$ and $R$ together with \rulename{Gel-Intro} and \rulename{Gel-$\eta$}. \qedhere
  \end{itemize}

\end{proof}

\section{Formal parametric type theory}\label{sec:formal}

While we have anchored our type theory in a computational interpretation, we would also like to use parametric
cubical type theory as a logic for reasoning about other settings. For this reason, we abstract a formal type
theory from the collection of inference rules we have developed in the preceding sections. The proofs of those
inference rules, as given for $\tygel$-types in \cref{sec:computational-rules}, establish that the
computational interpretation is one model of the formalism. In \cref{sec:presheaf}, we see that the theory can
also be interpreted in cartesian-affine bicubical sets.

We focus on parametric type theory here; for the cubical ingredients, we defer to prior work~\cite[Appendix
B]{angiuli19}.  In the pure parametric case, the theory is defined by the judgments shown in \cref{fig:formal}
and their equality counterparts. We take care to ensure our definition constitutes a generalized algebraic
theory (GAT)~\cite{cartmell86}, using for example explicit substitutions.\footnote{%
  We will nevertheless permit ourselves a certain amount of routine syntactic sugar; for one, we will not
  fully annotate terms.
} %
Ensuring admissibility of substitution---that every term is equal to one containing no explicit
substitutions---requires some innovation. In particular, the theory presented in~\cite{bernardy15} does not
satisfy admissibility of substitution, a consequence of the way rules using interval terms (such as bridge
elimination) are formulated. Rectifying this issue motivates the introduction of the context restriction
operator $\ctxres{-}{-}$ we have already encountered. We present a formulation of context restriction as an
explicit context former characterized as a left adjoint to extension by an interval variable.

\begin{figure}
  \centering
  \begin{tabular}{ll}
    $\wfctx*{\GG}$ & $\GG$ is a context \\
    $\wfbdim*[\GG]{\bmr}$ & $\bmr$ is a bridge interval term in context $\GG$ \\
    $\wftype*[\GG]{A}$ & $A$ is a type in context $\GG$ \\
    $\wftm*[\GG]{M}{A}$ & $M$ is a term of type $A$ in context $\GG$ \\
    $\wfsubst*[\GG]{\Gd}{\GD}$ & $\Gd$ is a substitution for context $\GD$ in context $\GG$ \\
  \end{tabular}

  \caption{Judgments of formal parametric type theory}%
  \label{fig:formal}
\end{figure}

We defer serious metatheoretic analysis of the formalism we present, such as normalization or decidability of
equality, to future work.

\subsection{The bridge interval}

The main novelty is our treatment of bridge interval restriction. Rather than relying on an operation
$\ctxres{-}{\bmr}$ on raw contexts---which would destroy the algebraic character of the theory---we treat context
restriction as a primitive context-forming operation.
\begin{mathparpagebreakable}
  \inferrule[ctx-nil]
  { }
  {\wfctx*{\ctxnil}}
  \and
  \inferrule[ctx-term]
  {\wftype*[\GG]{A}}
  {\wfctx*{\ctxsnoc{\GG}{A}}}
  \and
  \inferrule[ctx-$\BFI$]
  {\wfctx*{\GG}}
  {\wfctx*{\ctxbdim{\GG}}}
  \and
  \inferrule[ctx-restrict]
  {\wfctx*{\GG} \\
    \wfbdim*[\GG]{\bmr}}
  {\wfctx*{\ctxres*{\GG}{\bmr}}}
\end{mathparpagebreakable}
As is usual for ordinary terms, interval terms include variables and are closed under (explicit)
substitutions. We defer the matter of the constants $\bm0$ and $\bm1$ for the moment.
\begin{mathparpagebreakable}
  \inferrule[$\BFI$-var]
  { }
  {\wfbdim*[\ctxbdim{\GG}]{\bdimvar}}
  \and
  \inferrule[$\BFI$-subst]
  {\wfbdim*[\GD]{\bmr} \\
    \wfsubst*[\GG]{\Gd}{\GD}}
  {\wfbdim*[\GG]{\bdimsubst{\bmr}{\Gd}}}
\end{mathparpagebreakable}

Restriction is characterized by its relationship with extension by a bridge interval variable. Given an
interval term $\wfbdim*[\GG]{\bmr}$ and substitution $\wfsubst*[\ctxres*{\GG}{\bmr}]{\Gd}{\GD}$, we may build
a substitution $\wfsubst*[\GG]{\substbdim*{\Gd}{\bmr}}{\ctxbdim{\GD}}$. Conversely, given
$\wfsubst*[\GG]{\Gd}{\ctxbdim{\GD}}$, we may project a term $\wfbdim*[\GG]{\bdimsubst{\bdimvar}{\Gd}}$ and
substitution $\wfsubst*[\ctxres*{\GG}{\bdimsubst{\bdimvar}{\Gd}}]{\substbtranspose{\Gd}}{\GD}$. This sets up
an adjunction between the category of contexts $\GG$ and its slice over the bridge interval, which is to say
the category of substitutions $\wfbdim*[\GG]{\bmr}$, with $\ctxres{-}{-}$ as the left adjoint and
$\ctxbdim{-}$ as the right.
\begin{mathpar}
  \inferrule[subst-$\BFI$]
  {\wfbdim*[\GG]{\bmr} \\
    \wfsubst*[\ctxres*{\GG}{\bmr}]{\Gd}{\GD}}
  {\wfsubst*[\GG]{\substbdim*{\Gd}{\bmr}}{\ctxbdim{\GD}}}
  \and
  \inferrule[subst-restrict]
  {\wfsubst*[\GG]{\Gd}{\ctxbdim{\GD}}}
  {\wfsubst*[\ctxres*{\GG}{\bdimsubst{\bdimvar}{\Gd}}]{\substbtranspose{\Gd}}{\GD}}
  \\
  \inferrule[subst-eq-$\BFI$]
  {\wfctx*{\GD} \\
    \wfsubst*[\GG]{\Gd}{\ctxbdim{\GD}}}
  {\eqsubst*[\GG]{\Gd}{\substbdim*{\substbtranspose{\Gd}}{\bdimsubst{\bdimvar}{\Gd}}}{\ctxbdim{\GD}}}
  \and
  \inferrule[subst-eq-restrict]
  {\wfbdim*[\GG]{\bmr} \\
    \wfsubst*[\ctxres*{\GG}{\bmr}]{\Gd}{\GD}}
  {\eqsubst*[\ctxres*{\GG}{\bmr}]{\Gd}{\substbtranspose{(\substbdim*{\Gd}{\bmr})}}{\GD}}
\end{mathpar}
These rules induce a functorial action by interval extension,
$\wfsubst*[\ctxbdim{\GG}]{\funcbdim{\Gd} \eqdef
  \substbdim*{(\substconc{\Gd}{\substbtranspose{\substid}})}{\bdimvar}}{\ctxbdim{\GD}}$, as well as an action
by restriction,
$\wfsubst*[\ctxres*{\GG}{\bdimsubst{\bmr}{\Gd}}]{\funcres{\Gd}{\bmr} \eqdef
  \substbtranspose{(\substconc{\substbdim*{\substid}{\bmr}}{\Gd})}}{\ctxres*{\GD}{\bmr}}$. Using these, we
additionally require that the correspondence is natural.

\begin{mathpar}
  \inferrule[subst-$\BFI$-natural]
  {\wfsubst*[\GG]{\Gd}{\GD} \\
    \wfbdim*[\GX]{\bmr} \\
    \wfsubst*[\ctxres*{\GX}{\bmr}]{\Gg}{\GG}}
  {\eqsubst*[\GX]{\substbdim*{(\substconc{\Gd}{\Gg})}{\bmr}}{\substconc{\funcbdim{\Gd}}{(\substbdim*{\Gg}{\bmr})}}{\ctxbdim{\GD}}}
  \and
  \inferrule[subst-restrict-natural]
  {\wfsubst*[\GG]{\Gd}{\ctxbdim{\GD}} \\
    \wfsubst*[\GX]{\Gg}{\GG}}
  {\eqsubst*[\ctxres*{\GX}{\bdimsubst{\bdimvar}{\substconc{\Gd}{\Gg}}}]{\substbtranspose{(\substconc{\Gd}{\Gg})}}{\substconc{\substbtranspose{\Gd}}{(\funcres{\Gg}{\bdimsubst{\bdimvar}{\Gd}})}}{\GD}}
\end{mathpar}

The structural laws
and constants are then given as generating substitutions (together with the expected equations between them,
such as $\substconc{\substprojbdim}{\substbface{\Ge}} = \substid$ and naturality laws).

\begin{mathpar}
  \inferrule[subst-face]
  {\Ge \in \{0,1\}}
  {\wfsubst*[\GG]{\substbface{\Ge}}{\ctxbdim{\GG}}}
  \and
  \inferrule[subst-degen]
  { }
  {\wfsubst*[\ctxbdim{\GG}]{\substprojbdim}{\GG}}
  \and
  \inferrule[subst-exchange]
  {\wfctx*{\GG}}
  {\wfsubst*[\ctxbdim{\ctxbdim{\GG}}]{\substbex}{\ctxbdim{\ctxbdim{\GG}}}}
\end{mathpar}

Note that the existence of a substitution $\wfsubst*[\GG]{\substbface{\Ge}}{\ctxbdim{\GG}}$ is
slightly stronger than the existence of a term $\wfbdim*[\GG]{\overline{\substbface{\bm\Ge}}}$; the
latter would only give us a substitution
$\wfsubst*[\GG]{\substbdim*{\substid}{\overline{\substbface{\Ge}}}}{\ctxbdim{\ctxres*{\GG}{\bdimsubst{\bdimvar}{\overline{\substbface{\Ge}}}}}}$.

We note that the rules for $\BFI$ we have presented so far are consistent with an interpretation by a
structural interval, in which case context restriction would be the identity function. It is not until we
introduce rules for $\tmextent$ and $\tygel$ that the structural interval ceases to model the theory.

On the cubical side, we can treat path interval variables in the same way as term variables. However, we also
need the principle that bridge and path variables can be exchanged.
\begin{mathparpagebreakable}
  \inferrule[subst-$\BI$]
  {\wfsubst*[\GG]{\Gd}{\GD} \\
    \wfpdim*[\GD]{r}}
  {\wfsubst*[\GG]{\substpdim*{\Gd}{r}}{\ctxpdim{\GD}}}
  \and
  \inferrule[subst-proj-$\BI$]
  { }
  {\wfsubst*[\ctxpdim{\GG}]{\substprojpdim}{\GG}}
  \and
  \inferrule[subst-$\BI\BFI$]
  {\wfctx*{\GG}}
  {\wfsubst*[\ctxpdim{\ctxbdim{\GG}}]{\substpb}{\ctxbdim{\ctxpdim{\GG}}}}
\end{mathparpagebreakable}

The substitution $\substpb$ serves to invert the substitution
$\wfsubst*[\ctxbdim{\ctxpdim{\GG}}]{\substpdim*{\funcbdim{\substprojpdim}}{\pdimsubst{\pdimvar}{\substprojbdim}}}{\ctxpdim{\ctxbdim{\GG}}}$,
and expresses that path terms are always apart from bridge terms. Besides this principle, the cubical and
parametric sides of the theory only interact via the allowance for bridge constraints in $\hcomp$ terms and
the inclusion of rules for computing Kan operations in $\tybridge$- and $\tygel$-types, which we may formulate
following the operational semantics shown in \cref{fig:kan}.

\subsection{Type and term formers}%
\label{sec:formal:type-formers}

With the judgmental infrastructure in place, it is fairly straightforward to translate the computational type
formers introduced in \cref{sec:parametric} to the formal setting. We describe the rules for $\tybridge$-types
here; rules for $\tygel$-types and $\tmextent$ may be found in \cref{app:formal}.  The formation,
introduction, and elimination rules for $\tybridge$-types follow exactly the pattern of Figure~\ref{fig:bridge-types}.

\begin{mathpar}
  \inferrule
  {\wftype*[\ctxbdim{\GG}]{A} \\
    \wftm*[\GG]{M_0}{\tysubst{A}{\substbface0}} \\
    \wftm*[\GG]{M_1}{\tysubst{A}{\substbface1}}}
  {\wftype*[\GG]{\tybridge{A}{M_0}{M_1}}}
  \and
  \inferrule
  {\wftype*[\ctxbdim{\GG}]{A} \\
    \wftm*[\ctxbdim{\GG}]{M}{A}}
  {\wftm*[\GG]{\tmblam{M}}{\tybridge{A}{\tmsubst{M}{\substbface0}}{\tmsubst{M}{\substbface1}}}}
  \and
  \inferrule
  {\wfbdim*[\GG]{\bmr} \\
    \wftype*[\ctxbdim{\ctxres*{\GG}{\bmr}}]{A} \\
    \wftm*[\ctxres*{\GG}{\bmr}]{M_0}{\tysubst{A}{\substbface0}} \\
    \wftm*[\ctxres*{\GG}{\bmr}]{M_1}{\tysubst{A}{\substbface1}} \\
    \wftm*[\ctxres*{\GG}{\bmr}]{P}{\tybridge{A}{M_0}{M_1}}}
  {\wftm*[\GG]{\tmbapp{P}{\bmr}}{\tysubst{A}{\substbdim*{\substid}{\bmr}}}}
\end{mathpar}

It is the elimination rule---along with the rules for $\tmextent$ and $\tygel$-types---that necessitates the
introduction of the interval restriction operator. In~\cite{bernardy15}, bridge elimination is instead
described by a rule of the following kind.

\begin{mathpar}
  \inferrule[tm-app-bcm]
  {\wftype*[\ctxbdim{\GG}]{A} \\
    \wftm*[\GG]{M_0}{\tysubst{A}{\substbface0}} \\
    \wftm*[\GG]{M_1}{\tysubst{A}{\substbface1}} \\
    \wftm*[\GG]{P}{\tybridge{A}{M_0}{M_1}}}
  {\wftm*[\ctxbdim{\GG}]{\mathsf{app}(P)}{A}}
\end{mathpar}

This form of elimination is inter-derivable with our own: one may set
$\tmbapp{P}{\bmr} \eqdef \tmsubst{\mathsf{app}(P)}{\substbdim*{\substid}{\bmr}}$ or conversely
$\mathsf{app}(P) \eqdef \tmbapp{\tmsubst{P}{\substbtranspose{\substid}}}{\bdimvar}$. However, the~\cite{bernardy15} rule produces a formalism in which substitution is not admissible, that is, a theory in
which not every term is equal to one containing no use of the $\tmsubst{-}{-}$ operator. Given $P$ as in the
rule and a substitution $\wfsubst*[\GD]{\Gg}{\ctxbdim{\GG}}$, there is no way to reduce the term
$\tmsubst{\mathsf{app}(P)}{\Gg}$ unless it happens that $\GD = \ctxbdim{\GD'}$ and $\Gg = \funcbdim{\Gg'}$ for
some $\wfsubst*[\GD']{\Gg'}{\GG}$, in which case
$\tmsubst{\mathsf{app}(P)}{\Gg} = \mathsf{app}(\tmsubst{P}{\Gg'})$. By contrast, we may reduce a term
$\tmsubst{(\tmbapp{P}{\bmr})}{\Gg}$ using the functorial action of restriction, as prescribed by the rule
below.

\begin{mathpar}
  \inferrule
  {\wfsubst*[\GD]{\Gg}{\GG} \\
    \wfbdim*[\GG]{\bmr} \\
    \wftype*[\ctxbdim{\ctxres*{\GG}{\bmr}}]{A} \\
    \wftm*[\ctxres*{\GG}{\bmr}]{M_0}{\tysubst{A}{\substbface0}} \\
    \wftm*[\ctxres*{\GG}{\bmr}]{M_1}{\tysubst{A}{\substbface1}} \\
    \wftm*[\ctxres*{\GG}{\bmr}]{P}{\tybridge{A}{M_0}{M_1}}}
  {\eqtm*[\GG]{\tmsubst{(\tmbapp{P}{\bmr})}{\Gg}}{\tmbapp{\tmsubst{P}{\funcres{\Gg}{\bmr}}}{\bdimsubst{\bmr}{\Gg}}}{\tysubst{\tysubst{A}{\substbdim*{\substid}{\bmr}}}{\Gg}}}
\end{mathpar}

Finally, the $\beta$-, $\eta$-, and boundary rules for $\tybridge$-types can be expressed as follows. Note that these
rules respectively make use of the unit
$\wfsubst*[\GG]{\substbdim*{\substid}{\bmr}}{\ctxbdim{\ctxres*{\GG}{\bmr}}}$ and counit
$\wfsubst*[\ctxres*{\ctxbdim{\GG}}{\bdimvar}]{\substbtranspose{\substid}}{\GG}$ of the adjunction between
$\ctxres{-}{-}$ and $\ctxbdim{-}$.

\begin{mathpar}
  \inferrule
  {\wfbdim*[\GG]{\bmr} \\
    \wftype*[\ctxbdim{\ctxres*{\GG}{\bmr}}]{A} \\
    \wftm*[\ctxbdim{\ctxres*{\GG}{\bmr}}]{M}{A}}
  {\eqtm*[\GG]{\tmbapp{\tmlam{M}}{\bmr}}{\tmsubst{M}{\substbdim*{\substid}{\bmr}}}{\tysubst{A}{\substbdim*{\substid}{\bmr}}}}
  \and
  \inferrule
  {\wftype*[\ctxbdim{\GG}]{A} \\
    \wftm*[\GG]{M_0}{\tysubst{A}{\substbface0}} \\
    \wftm*[\GG]{M_1}{\tysubst{A}{\substbface1}} \\
    \wftm*[\GG]{P}{\tybridge{A}{M_0}{M_1}}}
  {\eqtm*[\GG]{P}{\tmblam{\tmbapp{\tmsubst{P}{\substbtranspose{\substid}}}{\bdimvar}}}{\tybridge{A}{M_0}{M_1}}}
  \and
  \inferrule
  {\Ge \in \{0, 1\} \\
    \wftype*[\ctxbdim{\GG}]{A} \\
    \wftm*[\GG]{M_0}{\tysubst{A}{\substbface0}} \\
    \wftm*[\GG]{M_1}{\tysubst{A}{\substbface1}} \\
    \wftm*[\GG]{P}{\tybridge{A}{M_0}{M_1}}}
  {\eqtm*[\GG]{\tmbapp{\tmsubst{P}{\substbtranspose{\substbface\Ge}}}{\bdimsubst{\bdimvar}{\substbface\Ge}}}{M_\Ge}{\tysubst{A}{\substbface\Ge}}}

\end{mathpar}

\section{A semantics in bicubical sets}\label{sec:presheaf}

We now describe a second semantics for the formal type theory of \cref{sec:formal} in a presheaf category of
\emph{bicubical sets}, adapting Angiuli et al.'s presheaf semantics for cubical type theory~\cite{abcfhl}.

\begin{defi}
  We define the \emph{cartesian-affine bicube category} $\bicube$ to have as objects interval contexts $\GPS$
  and as morphisms interval substitutions $\wfsubst<\GPS'>{\Gps}{\GPS}$, as specified in Definition~\ref{def:closed-interval-judgments}.
\end{defi}

\begin{rem}
  The category $\bicube$ is equivalent to a product $\cart \times \aff$ of two \emph{cube categories}, the
  \emph{cartesian cube category} $\cart$ consisting of path interval contexts and the \emph{affine cube
    category} $\aff$ consisting of bridge interval contexts.
\end{rem}

The presheaf category $\PSh{\bicube}$ is the category of contravariant functors from $\bicube$ to $\Set$,
meaning that its objects are families of sets indexed by interval contexts with transition maps for each
interval substitution. This parallels the situation in the computational interpretation, where types are given
meaning by families of relations indexed by such contexts. We use $\YoSym$ (hiragana `yo') to denote the
Yoneda embedding $\bicube \to \PSh{\bicube}$.

\begin{rem}
  Bernardy, Coquand, and Moulin instead interpret their type theory in a category of \emph{refined presheaves}
  on $\aff$~\cite{bernardy15}. Roughly, a refined presheaf is a $\GPS$-indexed family where for each
  $\GPS \in \aff$, we have not merely a set but a \emph{$\GPS$-set}, a family of sets indexed by sub-contexts
  $\GPS' \subseteq \GPS$. This refinement is used to validate the equivalents in their setting of equations
  $\tybridge{\bmx.\tygel{\bmx}{A_0}{A_1}{R}} = R$ and
  $C = \tmblam[x]{\tygel{\bmx}{A_0}{A_1}{\tybridge{\bmx.\tmbapp{C}{\bmx}}}}$, as mentioned in
  \cref{sec:parametric:gel}. When we build parametric type theory on a cubical base, we no longer need these
  equations to hold exactly, as we can prove they hold up to a path using univalence (Theorem~\ref{thm:relativity}).
\end{rem}

\subsection{Judgments and cubical type theory}

We recall the presheaf interpretation of the judgments of cubical type theory developed in~\cite{cchm,abcfhl},
which draw on earlier presheaf interpretations of dependent type theory~\cite{hofmann97}.

\begin{defi}
  A \emph{semantic context} is a presheaf $G \in \PSh{\bicube}$; a \emph{semantic substitution} between
  contexts $G',G$ is a presheaf morphism (\ie, natural transformation) $\Ga : G' \to G$.
\end{defi}

\begin{defi}
  A \emph{semantic pretype} over a context $G \in \PSh{\bicube}$ is a presheaf $T \in \PSh{(\int G)}$ over the
  category of elements $\int G$, which is to say the following data:
  \begin{itemize}[label=$\triangleright$]
  \item for every $\GPS \in \bicube$ and $g \in G(\GPS)$, a set $T(\GPS,g)$;
  \item for every $\wfsubst<\GPS'>{\Gps}{\GPS}$ and $g \in G(\GPS)$, a map
    $T(\Gps) : T(\GPS',G(\Gps)(g)) \to T(\GPS,g)$.
  \end{itemize}
\end{defi}

\begin{defi}
  A \emph{semantic element} $t$ of a pretype $T$ in context $G$ is a family of elements
  $t(\GPS,g) \in T(\GPS,g)$ indexed by $\GPS \in \bicube$ and $g \in G(\GPS)$ such that
  $T(\Gps)(t(\GPS,g)) = t(\GPS',G(\Gps)(g))$ for every $\wfsubst<\GPS'>{\Gps}{\GPS}$ and $g \in G(\GPS)$.
\end{defi}

A semantic \emph{type} is then a pretype equipped with coercion and homogeneous composition operators
implementing the rules shown in Figure~\ref{fig:kan}. We give the definition of coercion operator here and
leave it to the reader to infer the corresponding notion of \emph{homogeneous composition operator}.

\begin{defi}
  Given a pretype $T$ over $G$, a \emph{coercion operator} $c$ for $T$ is a family of elements as follows: for
  every $\GPS \in \bicube$, interval terms $\wfpdim<\GPS>{r,s}$, element $g \in G(\ctxpdim{\GPS}[x])$,
  and $t \in T(\GPS,G(\substpdim{\substid[\GPS]}{r}[x])(g))$, we require an element
  $c(\GPS,r,s,g,t) \in T(\GPS,G(\substpdim{\substid[\GPS]}{s}[x])(g))$. We ask that these satisfy the
  following properties.
  \begin{itemize}[label=$\triangleright$]
  \item $T(\Gps)(c(\GPS,r,s,g,t)) = c(\GPS',r\Gps,s\Gps,G(\Gps)(g),T(\Gps)(t))$ for every
    $\wfsubst<\GPS'>{\Gps}{\GPS}$.
  \item $c(\GPS,r,r,g,t) = t$.
  \end{itemize}
\end{defi}

\begin{defi}
  A \emph{semantic type} $(T,c,h)$ over $G$ is a triple consisting of a semantic pretype $T$ over $G$ with
  coercion and homogeneous composition operators $c$ and $h$.
\end{defi}

\begin{rem}
  A semantic substitution $\Ga : G' \to G$ acts on types and terms over $G$ by reindexing; we write $\Ga^*T$
  and $\Ga^*t$ for the action on types and terms respectively.
\end{rem}

\begin{defi}
  A \emph{semantic interval term} over $G$ is a presheaf morphism $r : G \to \yon{(\ctxpdim{}[x])}$. A
  \emph{semantic constraint} is a morphism $r : G \to \subdec$ where $\subdec$ is the \emph{decidable
    subobject classifier} in $\PSh{\bicube}$, which classifies monomorphisms $m : H' \rightarrowtail H$ in
  $\PSh{\bicube}$ such that $m(\GPS)$ has decidable image for all $\GPS$.
\end{defi}

Angiuli \etal's~\cite[Theorem 1]{abcfhl} shows that cartesian cubical type theory can be interpreted using
these semantic judgments in any presheaf category whose base category contains a suitably structured interval
object.

\begin{prop}
  $\PSh{\bicube}$ interprets cubical type theory with an infinite hierarchy of univalent universes, each
  closed under dependent function and product types, $\typath$-types, and $\tyv$-types.
\end{prop}
\begin{proof}
  By Angiuli \etal's Theorem 1~\cite{abcfhl}. The formulation of cartesian cubical type theory given there is
  slightly different from our own (for example taking $\comp$ rather than $\coe$ and $\hcomp$ as primitive),
  but not in any essential way.

  We note that the statement of the theorem in~\cite{abcfhl} requires that the base category is closed under
  finite products, which is not the case for $\bicube$: the cartesian product of the contexts
  $(\ctxbdim{}[\bmx])$ and $(\ctxbdim{}[\bmy])$ does not exist. However, the proof only actually requires that
  the product functor $- \times (\ctxpdim{}[x])$ exists, and this is indeed the case in $\bicube$.
\end{proof}

\subsection{Bridge interval and restriction}

We now turn to the parametric side of the theory. As with the path interval, we interpret bridge interval
terms in a context $G$ as morphisms $\bmr : G \to \yon{(\ctxbdim{}[x])}$. To interpret bridge interval context
extension and restriction, we observe that we have an adjunction between $\bicube$ and its slice category over
the affine interval $(\ctxbdim{}[\bmx])$. Note that elements of this slice category consist of contexts $\GPS$
paired with bridge interval terms $\wfbdim<\GPS>{\bmr}$.
\[
  \begin{tikzcd}[column sep=4em]
    \bicube/(\ctxbdim{}[\bmx]) \ar[bend left]{r}[above,font=\normalsize]{\cubres} \ar[phantom]{r}{\bot} &
    \bicube\phantom{/(\ctxbdim{}[\bmx])} \ar[bend left]{l}[below,font=\normalsize]{\cubext}
  \end{tikzcd}
\]

The right adjoint $\cubext$ sends a context $\GPS$ to the extended context $(\ctxbdim{\GPS}[\bmx])$ with its
canonical projection $\wfbdim<\ctxbdim{\GPS}[\bmx]>{\bmx}$. The left adjoint is interval restriction: it sends
a pair $(\GPS,\bmr)$ to the restricted context $\ctxres{\GPS}{\bmr}$ defined here as in Section~\ref{sec:parametric:interval}.
\begin{align*}
  \ctxres{\GPS}{\bm\Ge} &\eqdef \GPS \quad \text{if $\bm\Ge \in \{\bm0,\bm1\}$} \\
  \ctxres{(\ctxpdim{\GPS}[y])}{\bmx} &\eqdef \ctxpdim{\ctxres{\GPS}{\bmx}}[y] \\
  \ctxres{(\ctxbdim{\GPS}[y])}{\bmx} &\eqdef
  \left\{
    \begin{array}{ll}
      \GPS & \text{if $\bmx = \bmy$} \\
      \ctxbdim{\ctxres{\GPS}{\bmx}}[y] & \text{if $\bmx \neq \bmy$}
    \end{array}
  \right.
\end{align*}

This adjunction in the base category induces, among other things, the following pair of adjoint functors
between the presheaf category and its slice. We implicitly use the equivalence
$\PSh{(\bicube/\GPS)} \simeq \PSh{\bicube}/\yon{\GPS}$ between presheaves on slice categories and slices over
representables.
\[
  \begin{tikzcd}[column sep=7em, row sep=4em]
    \PSh{\bicube}/\yon{(\ctxbdim{}[\bmx])} \ar[bend left]{r}[above,font=\normalsize]{{\cubres}_!} \ar[phantom]{r}{\bot} &
    \PSh{\bicube} \ar[bend left]{l}[font=\normalsize]{\cubres^* \cong \cubext_!}
  \end{tikzcd}
\]

Here $\cubres^*$ is precomposition with $\cubres$, while $\cubres_!$ and $\cubext_!$ are each defined by
left Kan extension. Both $\cubext_!$ and $\cubres^*$ are left adjoint to $\cubext^*$, so are
necessarily isomorphic. As for $\cubres_!$, it may be explicitly calculated as the following coend.
\[
  \cubres_!(G,\bmr)(\GPS) = \int^{\wfbdim*<\GPS'>{\bms}} \{g \in G(\Psi') \mid \bmr(\Psi')(g) = \bms\} \times \{\Gps \mid \wfsubst<\GPS>{\Gps}{\ctxres{\GPS'}{\bms}}\}
\]
For our purposes, however, it is only necessary to know that the extensions ${\cubres}_!$ and $\cubext_!$
apply the base functors on representables, that is, that
${\cubres}_!(\yon{\GPS},\bmr) \cong \yon{(\ctxres{\GPS}{\bmr})}$ and
$\cubext_!(\yon{\GPS}) \cong \yon{(\substbdim{}{\bmx}[\bmx])} : \yon{(\ctxbdim{\GPS}[\bmx])} \to
\yon{(\ctxbdim{}[\bmx])}$; this is a general property of Kan extensions. Henceforth we write $\cubprod{G}$ for
the object part of $\cubext_!G$ and $\sembvar{G} : \cubprod{G} \to \yon{(\ctxbdim{}[\bmx])}$ for the
associated projection.

We use $\cubres_!$ to interpret the type-theoretic restriction of a context by an interval term, likewise
$\cubprod{-}$ to interpret extension by an interval variable and $\sembvar{G}$ for the variable rule. The
isomorphism between hom-sets given by the adjunction $\cubres_! \dashv \cubext_!$ implements the substitution
constructors \rulename{subst-$\BFI$} and \rulename{subst-restrict}. The structural rules for the bridge
interval derive from natural transformations in the base category via the action of ${(-)}_!$; for example, the
endpoint transformation $\bm\Ge : \funcid \to \pi \circ \cubext$ defined by
$\wfsubst<\GPS>{\bm\Ge(\GPS) \eqdef (\substbdim{\substid[\GPS]}{\bm\Ge}[\bmx])}{(\ctxbdim{\GPS}[\bmx])}$
induces a corresponding transformation $\bm\Ge_! : \funcid \to (\cubprod{-})$ in the presheaf category
interpreting the endpoint substitution \rulename{subst-face}.

\subsection{Type and term formers}

To interpret the rules for forming types and terms---$\tybridge$-types, $\tygel$-types, and $\tmextent$---it
is useful to observe that the semantic judgments, like the computational ones in \cref{sec:computational}, are
determined by their instantiations at interval contexts (\ie, representables). For example, a semantic type
$T$ in context $G$ is determined by the types $g^* T$ for $g : \yon{\GPS} \to G$: recalling that the Yoneda
lemma identifies morphisms $g : \yon{\GPS} \to G$ with elements $g \in G(\GPS)$, we have as we have
$T(\GPS,g) = (g^*T)(\GPS,\substid[\GPS])$. Conversely, if we have a family of types $T_g$ over $\yon{\GPS}$
for every $g : \yon{\GPS} \to G$ such that ${(\yon{\Gps})}^*T_g = T_{g \circ \yon{\Gps}}$ for all
$\wfsubst<\GPS'>{\Gps}{\GPS}$, then this determines a type $T$ over $G$: take
$T(\GPS,g) \eqdef T_g(\GPS,\substid[\GPS])$. A similar principle applies to terms.

The upshot is that we may verify that rules hold in an arbitrary context by showing they hold (naturally) in
any interval context, as we did for the computational interpretation in Section~\ref{sec:computational-rules}.
In the restricted case we may take advantage of the characterizations
${\cubres}_!(\yon{\GPS},\bmr) \cong \yon{(\ctxres{\GPS}{\bmr})}$ and
$\cubext_!(\yon{\GPS}) \cong \yon{(\substbdim{}{\bmx}[\bmx])} : \yon{(\ctxbdim{\GPS}[\bmx])} \to
\yon{(\ctxbdim{}[\bmx])}$, saving us from formal reasoning with the general Kan extension.

\begin{thm}%
  \label{thm:presheaf-bridge}
  $\PSh{\bicube}$ is closed under $\tybridge$-pretypes.
\end{thm}
\begin{proof}
  Per the argument above, we narrow our attention without loss of generality to the cases where the ambient
  context is representable.

  \begin{itemize}[label=$\triangleright$]
  \item \emph{Formation.} \\
    Let a semantic pretype $T$ in context $\cubprod{\yon{\GPS}} \cong \yon{(\ctxbdim{\GPS}[\bmx])}$ be given
    together with endpoint elements $t_0$ of $\yon{(\substbdim{}{\bm0}[\bmx])}^*T$ and $t_1$ of
    $\yon{(\substbdim{}{\bm1}[\bmx])}^*T$. We define a semantic pretype $\sembridge{T}{t_0}{t_1}$ over
    $\yon{\GPS}$ as follows.
    \[
      \sembridge{T}{t_0}{t_1}(\GPS',\Gps) \eqdef
      \{
      a \in T((\ctxbdim{\GPS'}[\bmx]), (\substbdim{\Gps}{\bmx}[\bmx]))
      \mid
      \text{$\forall \Ge.\ T(\substbdim{}{\bm\Ge}[\bmx])(a) = t_\Ge(\GPS',\Gps)$} \}
    \]
    That is, an element of $\sembridge{T}{t_0}{t_1}$ in context $\GPS'$ is an element of $T$ in context
    $(\ctxbdim{\GPS'}[\bmx])$ with the requested endpoints. The action of $\sembridge{T}{t_0}{t_1}$ on
    substitutions is likewise defined from the action of $T$ in the natural way.
  \item \emph{Introduction.} \\
    Similarly, given a semantic element $t$ of $T$ such that $\yon{(\substbdim{}{\bm1}[\bmx])}^*t = t_0$ and
    $\yon{(\substbdim{}{\bm1}[\bmx])}^*t = t_1$, we have an abstracted element $\semblam{t}$ of
    $\sembridge{T}{t_0}{t_1}$ defined as follows.
    \[
      \semblam{t}(\GPS,g) \eqdef t((\ctxbdim{\GPS}[\bmx]), (\substbdim{\Gps}{\bmx}[\bmx]))
    \]
  \item \emph{Elimination.} \\
    To interpret application, we assume now that we have some
    $\bmr : \yon{\GPS} \to \yon{(\ctxbdim{\GPS}[\bmx])}$ and that the pretype $T$ lies in context
    $\cubprod{\cubres_!(\yon{\GPS},\bmr)} \cong \yon{(\ctxbdim{\ctxres{\GPS}{\bmr}}[\bmx])}$. Given
    an element $u$ of $\sembridge{T}{t_0}{t_1}$,
    \[
      \sembapp{u}(\GPS',\Gps) \eqdef
      T(\substbdim{}{\bmr\Gps}[\bmx])(u(\ctxres{\GPS'}{\bmr\Gps}, \ctxres{\Gps}{\bmr}))
    \]
    Here $\wfsubst<\ctxres{\GPS'}{\bmr\Gps}>{\ctxres{\Gps}{\bmr}}{\ctxres{\GPS}{\bmr}}$ is the functorial
    action of restriction on $\Gps$. By definition of the bridge type, the term
    $u(\ctxres{\GPS'}{\bmr\Gps}, \ctxres{\Gps}{\bmr})$ is an element of
    $T(\ctxbdim{(\ctxres{\GPS'}{\bmr\Gps}}[\bmx]),(\substbdim{\ctxres{\Gps}{\bmr}}{\bmx}[\bmx]))$; applying 
    $T(\substbdim{}{\bmr\Gps}[\bmx])$ thus gives an element of $T(\GPS',\Gps)$.
  \end{itemize}

  \noindent We leave it to the reader to check that these definitions are natural and that the $\beta$-,
  $\eta$-, and boundary rules are satisfied.
\end{proof}

We may show that the model interprets $\tybridge$-\emph{types}---that is, that $\tybridge$-pretypes can be
equipped with Kan operations---following the computational definition of $\coe$ and $\hcomp$ in
\cref{fig:opsem}; we leave this to the reader. Alternatively, one may follow the definition of composition for
$\typath$-types in the BCH model~\cite[\S7.2]{bch}.

\begin{thm}
  $\PSh{\bicube}$ interprets $\tygel$-pretypes.
\end{thm}
\begin{proof}
  We prove the formation rule, following the computational definition in Section~\ref{sec:computational:construct}.  It is straightforward to see how the introduction and elimination rules
  follow.

  Let a interval term $\bmr : \yon{\GPS} \to \yon{(\ctxbdim{}[\bmx])}$, semantic pretypes $T_0,T_1$ in context
  $\cubres_!(\yon{\GPS},\bmr) \cong \yon{(\ctxres{\GPS}{\bmr}})$, and a semantic pretype $R$ in context 
  $\yon{(\ctxres{\GPS}{\bmr}}).(T_0 \times T_1)$---here $-.-$ is the semantic equivalent of context 
  extension---be given. We define the $\tygel$-pretype as follows.
  \begin{align*}
    \semgel{\bmr}{T_0}{T_1}{R}(\GPS',\Gps)
    &\eqdef T_{\Ge}(\GPS',\ctxres{\Gps}{\bmr})
    &\text{if $\bmr\Gps = \bm\Ge$} \\
    \semgel{\bmr}{T_0}{T_1}{R}(\GPS',\Gps)
    &\eqdef
      \left\{(a_0,a_1,t) \,\middle|\,
      \begin{array}{l}
        a_\Ge \in T_{\Ge}(\ctxres{\GPS'}{\bmr\Gps},\ctxres{\Gps}{\bmr}) \\
        t \in {(\yon{(\ctxres{\Gps}{\bmr})}.(a_0 \times a_1))}^*R(\ctxres{\GPS'}{\bmr\Gps},\substid)
      \end{array}
    \right\}
    &\text{otherwise}
  \end{align*}
\end{proof}

As with $\tybridge$-types, the Kan operations may be implemented following the computational definition given
in \cref{fig:opsem}. We note that homogeneous composition relies on the closure of the decidable subobject
classifier $\subdec$ under $\forall \bmx. -$; this parallels the use of $\forall x. -$ for composition in
$\tyg$-, $\tyglue$-, or $\tyv$-types in~\cite{bch,cchm,abcfhl}. As $\tybridge$-types resemble BCH
$\typath$-types, so do $\tygel$-types resemble BCH $\tyg$-types. Coercion for $\tygel$ is, however, much
simpler than for its cubical equivalents, because the ``direction'' of a coercion is always a path variable
and therefore orthogonal to the direction $\bmr$ of $\tygel{\bmr}{A}{B}{R}$: one may coerce ``across'' a
$\tyv$-type, but not across a $\tygel$-type.

We finish by sketching the interpretation of $\tmextent$. Suppose we are given dimension term
$\bmr : \yon{\GPS} \to \yon{(\ctxbdim{}[\bmx])}$, type $T$ in context
$\yon{(\ctxbdim{\ctxres{\GPS}{\bmr}}[\bmx])}$, and element $t$ of $\yon{(\substbdim{}{\bmr}[\bmx])}^*T$,
together with clause data for the endpoint and variable cases. For any $\GPS'$ and
$\wfsubst<\GPS'>{\Gps}{\GPS}$, we have $t(\GPS',\Gps) \in T(\GPS',(\substbdim{\Gps}{\bmr\Gps}[\bmx]))$; we
proceed by inspecting the status of $\bmr\Gps$. If $\bmr\Gps$ is an endpoint, then we have
$t(\GPS',\Gps) \in T(\GPS',(\substbdim{\Gps}{\bmr\Gps}[\bmx])) =
(\yon{(\substbdim{}{\bm\Ge}[\bmx])}^*T)(\GPS',\Gps)$ and may pass this term to the appropriate endpoint
clause. If $\bmr\Gps$ is a variable, then we employ the substitution
$\wfsubst<\ctxbdim{\ctxres{\GPS'}{\bmr\Gps}}[\bmy]>{\rho}{\GPS'}$ that renames $\bmr\Gps$ to a fresh variable
$\bmy$. We have
$T(\rho)(t(\GPS',\Gps)) \in
T((\ctxbdim{\ctxres{\GPS'}{\bmr\Gps}}[\bmy]),(\substbdim{\ctxres{\Gps}{\bmr}}{\bmy}[\bmx]))$, which per the
proof of Theorem~\ref{thm:presheaf-bridge} is exactly a bridge at $T$. We may then supply this bridge to the
variable clause of $\tmextent$.

\section{Related and future work}\label{sec:related}

\begin{figure}
  \centering
  \begin{tabular}{lll}
    \toprule
    This paper &~\cite{bernardy15} &~\cite{moulin16} \\
    \midrule
    $\tybridge{\bmx.A}{a_0}{a_1}$ & $A \ni_{\bmx} a$ & $(\forall \bmx. A) \ni a$ \\
    $\tmblam[\bmx]{a}$ & $a \cdot {\bmx}$ & $(\langle \bmx \rangle a)!$ \\
    $\tmbapp{p}{\bmx}$ & $(a,_{\bmx} p)$ & $\llparenthesis a,_{\bmx} p \rrparenthesis$ \\
    $\tmextent{\bmx}{-}{a_0.t_0}{a_1.t_1}{a_0.a_1.\overline{a}.u}$ & $\langle \lambda a. t ,_{\bmx} \lambda a. \lambda \overline{a}. u \rangle$ & $\langle\!{\mid} \lambda a. t ,_{\bmx} \lambda a. \lambda \overline{a}. u {\mid}\!\rangle$ \\
    $\tygel{\bmx}{A_0}{A_1}{a_0.a_1.R}$ & $(a : A) \times_{\bmx} R$ & $A \bowtie_{\bmx} R$ \\
    $\tmgel{\bmx}{a_0}{a_1}{c}$ & $(a,_{\bmx} c)$ & $\llparenthesis a,_{\bmx} p \rrparenthesis$ \\
    $\tmungel{\bmx.a}$ & $a \cdot {\bmx}$ & $(\langle \bmx \rangle a)!$ \\
    \bottomrule
  \end{tabular}
  \caption{Translation dictionary for internal parametricity}%
  \label{fig:translation}
\end{figure}

\subsection{Related work}

Mechanically, our parametric cubical type theory is not much more than the union of Angiuli \etal's cartesian
cubical type theory~\cite{angiuli18,abcfhl,angiuli19} and Bernardy, Coquand, and Moulin's parametric type
theory~\cite{bernardy15}. As mentioned in \cref{sec:parametric:gel,sec:presheaf}, we do drop some equations
required for $\tygel$-types in the BCM type theory which are not necessary in the cubical setting and
complicate model constructions. Accordingly, our proof of relativity is novel. The formulation of context
restriction in formalism is also novel, though inspired by Cheney's work on nominal type theory~\cite{cheney12}, and resolves the issue with admissibility of substitution present in the BCM theory. Finally,
Bernardy \etal\ present unary rather than binary parametricity, but from a conceptual perspective this is only
a cosmetic difference, a matter of how many constants are included in the bridge interval.%
\footnote{%
  We conjecture that binary internal parametricity is more powerful than unary parametricity, but that ternary
  parametricity and so on provide no additional strength, because we can iterate binary parametricity to mimic
  $2^n$-ary parametricity for any $n$.
} %
As our notation is quite different from that of Bernardy \etal, we provide a comparison in
\cref{fig:translation}. Note that the mapping is not one-to-one because of the additional equations imposed in
their theory. We also include notations from Moulin's thesis~\cite{moulin16}. In that work, the notion of a
function $(i : \BI) \to A$ without a fixed endpoint (called a ``ray'') is included separately from bridge
types, and term formers that are primitive in~\cite{bernardy15} are often implemented as combinations of terms
relating first interval dependency to rays and then rays to bridges. In particular, $A \bowtie_{\bmx} R$ is
syntactic sugar for a term $\llparenthesis A , \GPS_A R \rrparenthesis @ \bmx$, while
$\langle\!{\mid} f ,_{\bmx} h {\mid}\!\rangle$ is sugar for
$\langle\!{\mid} f ,\Phi_f h {\mid}\!\rangle @ \bmx$; as a result, the equivalents of $\tygel$ and $\tmextent$
are sometimes called $\Psi$- and $\Phi$-operators respectively in the literature.

A second approach to internal parametricity has been proposed by Nuyts, Vezzosi, and Devriese~\cite{nuyts17}. Their system resembles our own in that it is based on bridges and paths, each of which is
represented by a kind of map from an interval. Whereas our bridge and path structures are more-or-less
orthogonal to each other, Nuyts \etal\ use a modality to connect the two. Terms are checked under different
modalities depending on whether they are used in type or element positions, capturing the phase separation
between type and element-level computation that is often identified as a consequence of parametricity.  We see
the two approaches of Bernardy \etal\ and Nuyts \etal\ as internalizing different perspectives on
parametricity: the former internalizes the relational interpretation, while the latter internalizes this phase
separation.

Nuyts \etal\ also distinguish between \emph{continuous} and \emph{parametric} function types: the former
preserve paths and bridges, while the latter take bridges to paths. By contrast, we consider the former to
already be ``parametric''---as we have seen, one can prove parametricity theorems in our setting using only
this property. However, the stronger condition does obviate the need to identify the class of bridge-discrete
types as a replacement for the identity extension lemma. For example, any parametric function $\tyuniv \to A$
in their setting is constant, without any assumptions on $A$ (\cf\ Lemma~\ref{lem:to-bridge-discrete}),
because it takes the bridges in $\tyuniv$ to paths. Also notable is that their path and bridge intervals both
behave structurally, whereas we use an affine interval for bridges. Given the other divergences from Bernardy
\etal's approach, it is difficult to say how the issues we raise with using structural variables for
parametricity affect their system, if at all; it seems that they are ameliorated by the stronger condition on
parametric functions. One notable limitation is that \emph{iterated} parametricity is impossible, that is, the
results produced by parametricity are not subject to further parametricity theorems. This is addressed in a
successor system~\cite{nuyts18}, which introduces an infinite hierarchy of bridge-like relationships
associated with universe level and is capable of capturing iterated parametricity as well as other modal forms
of hypothesis such as \emph{irrelevant} hypotheses.

Nuyts's thesis~\cite{nuyts20} provides a more systematic analysis of the different univalence-like type
formers used in cubical and parametric type theories---$\tyv$, $\tyglue$, $\tyg$, $\tygel$---as derivable from
a \emph{transpension} type former, characterized as the right adjoint to the interval function type former
$(i : \BI) \to -$. This type former corresponds in our setting to the operator $\tygel{\bmx}{\top}{\top}{-}$;
Nuyts derives $\tygel$ from this special case in combination with quantification over the boundary of
$\bmx$.

Tabareau, Tanter, and Sozeau~\cite{tabareau18} develop a theory of \emph{univalent parametricity} in the
Calculus of (Inductive) Constructions. This system defines a kind of relation across which results can be
transported, much as we transport results across isomorphisms using univalence, but develops a logical
relation incorporating ideas from parametricity in order to improve the usability properties of the transport
function. Although univalence and parametricity are both involved, therefore, the objectives are largely
orthogonal to our own.

Riehl and Shulman's \emph{directed type theory}~\cite{riehl17} is a theory in the same mold as our own: it has
two layers of higher structure, one which is used to express equality and one which is used for general
relations. In their case, the goal is to identify those types whose ``bridge'' structure has the structure of
an $(\infty,1)$-category, then use the theory as a language for synthetic higher category theory. Where our
semantics is based on a product of cube categories, they use a product of simplex categories. Interestingly,
their bisimplicial semantics fails to support a universe whose bridges are relations, for reasons that evoke
our comparison of $\tyv$- and $\tygel$-types in \cref{sec:parametric:gel}~\cite{riehl18}. However, the theory
\emph{does} support a universe of \emph{covariant discrete fibrations} in which bridges correspond to
functions (``directed univalence''). More recently, Weaver and Licata~\cite{weaver20} have developed a cubical
(and constructive) variation on this theory, based on the product of two structural cube categories. Like
Riehl and Shulman's theory, this theory supports a universe satisfying directed univalence, but we suspect it
too fails to support a relativistic universe.

Our work fits into traditions of both proof-relevant equality and proof-relevant parametricity. The former is,
of course, a primary focus of the field of homotopy type theory. Proof-relevant and higher-dimensional
variations on parametricity have been developed by Atkey \etal~\cite{atkey14}, Ghani \etal~\cite{ghani15},
and Sojakova and Johann~\cite{sojakova18}. More generally, Benton, Hofmann, and Nigam~\cite{benton14} use a
proof-relevant logical relation to study abstract effects, and proof-relevant logical families have recently
been deployed as tools for proving metatheorems for dependent type theories~\cite{shulman15,coquand18c}.

\subsection{Future work}

Our exploration in \cref{sec:practice} shows that internal parametricity can be effectively employed to prove
difficult theorems involving higher inductive types. However, this only means that these results can be
obtained in \emph{internally parametric} type theory; we would also like to know they are true in
non-parametric type theory. We believe a fruitful approach would be to combine parametric and non-parametric
type theories into a single, \emph{modal} theory containing a mode for parametric results and a mode for
non-pointwise results. In particular, the presheaf categories $\PSh{\cart}$ and $\PSh{\bicube}$, which
interpret cubical and parametric cubical type theory respectively, can be related by \emph{axiomatic
  cohesion}, which has been previously been used in the design of modal type theories~\cite{schreiber12,shulman18}.

The formalism we develop in \cref{sec:formal} must be supported by metatheoretic results such as normalization
in order to be truly utile. We have implemented an experimental type-checker for (non-cubical) parametric type
theory, \texttt{ptt}, based on \emph{normalization by evaluation}; in theory, this implementation implicitly
contains a proof of normalization for the \cref{sec:formal} formalism. However, we have not attempted to
extract such a proof, nor have we verified the algorithm's correctness. The current \texttt{ptt} theory is
also somewhat weaker than that of \cref{sec:formal}: we found it more convenient to give the $\tygel$ type a
positive eliminator rather than a projection with $\eta$-principle. The $\eta$-expansion rule we have used in
this paper applies only to terms that can be put in the form $\usubstdim{Q}{\bmr}{\bmx}$, a condition that
is to our knowledge expensive and painful (though we believe possible) to check.

\section*{Acknowledgments}

We thank Carlo Angiuli, Steve Awodey, Daniel Gratzer, Kuen-Bang Hou (Favonia), Dan Licata, Anders
M\"{o}rtberg, Emily Riehl, Christian Sattler, Michael Shulman, Jonathan Sterling, and Andrew Swan for many
helpful discussions.

\pagebreak
\appendix

\section{Formal parametric type theory}\label{app:formal}

Rules for pushing substitutions through type and term formers are omitted.

\subsection{Contexts}~%

\begin{mathparpagebreakable}
  \inferrule[ctx-nil]
  { }
  {\wfctx*{\ctxnil}}
  \and
  \inferrule[ctx-term]
  {\wftype*[\GG]{A}}
  {\wfctx*{\ctxsnoc{\GG}{A}}}
  \and
  \inferrule[ctx-$\BFI$]
  {\wfctx*{\GG}}
  {\wfctx*{\ctxbdim{\GG}}}
  \and
  \inferrule[ctx-restrict]
  {\wfctx*{\GG} \\
    \wfbdim*[\GG]{\bmr}}
  {\wfctx*{\ctxres*{\GG}{\bmr}}}
\end{mathparpagebreakable}

\subsection{Interval terms}~%

\begin{mathparpagebreakable}
  \inferrule[$\BFI$-var]
  { }
  {\wfbdim*[\ctxbdim{\GG}]{\bdimvar}}
  \and
  \inferrule[$\BFI$-subst]
  {\wfbdim*[\GD]{\bmr} \\
    \wfsubst*[\GG]{\Gd}{\GD}}
  {\wfbdim*[\GG]{\bdimsubst{\bmr}{\Gd}}}
\end{mathparpagebreakable}

\subsection{Interval term equality}~%

\begin{mathparpagebreakable}
  \inferrule[$\BFI$-subst-id]
  {\wfbdim*[\GG]{\bmr}}
  {\eqbdim*[\GG]{\bdimsubst{\bmr}{\substid}}{\bmr}}
  \and
  \inferrule[$\BFI$-subst-conc]
  {\wfbdim*[\GD_0]{\bmr} \\
    \wfsubst*[\GD_1]{\Gd_0}{\GD_0} \\
    \wfsubst*[\GG]{\Gd_1}{\GD_1}}
  {\eqbdim*[\GG]{\bdimsubst{\bmr}{\substconc{\Gd_0}{\Gd_1}}}{\bdimsubst{\bdimsubst{\bmr}{\Gd_0}}{\Gd_1}}}
  \and
  \inferrule[$\BFI$-subst-term]
  {\wfbdim*[\GG]{\bmr} \\
    \wfsubst*[\ctxres*{\GG}{\bmr}]{\Gd}{\GD}}
  {\eqbdim*[\GG]{\bdimsubst{\bdimvar}{\substsnoc*{\Gd}{\bmr}}}{\bmr}}
\end{mathparpagebreakable}

\subsection{Substitutions}~%

\begin{mathparpagebreakable}
  \inferrule[subst-nil]
  { }
  {\wfsubst*[\GG]{\substnil*}{\ctxnil}}
  \and
  \inferrule[subst-id]
  { }
  {\wfsubst*[\GG]{\substid}{\GG}}
  \and
  \inferrule[subst-conc]
  {\wfsubst*[\GD_1]{\Gd_0}{\GD_0} \\
    \wfsubst*[\GG]{\Gd_1}{\GD_1}}
  {\wfsubst*[\GG]{\substconc{\Gd_0}{\Gd_1}}{\GD_0}}
  \and
  \inferrule[subst-term]
  {\wfsubst*[\GG]{\Gd}{\GD} \\
    \wftm*[\GG]{M}{\tysubst{A}{\Gd}}}
  {\wfsubst*[\GG]{\substsnoc*{\Gd}{M}}{\ctxsnoc{\GD}{A}}}
  \and
  \inferrule[subst-proj]
  {\wftype*[\GG]{A}}
  {\wfsubst*[\ctxsnoc{\GG}{A}]{\substproj}{\GG}}
  \and
  \inferrule[subst-$\BFI$]
  {\wfbdim*[\GG]{\bmr} \\
    \wfsubst*[\ctxres*{\GG}{\bmr}]{\Gd}{\GD}}
  {\wfsubst*[\GG]{\substbdim*{\Gd}{\bmr}}{\ctxbdim{\GD}}}
  \and
  \inferrule[subst-restrict]
  {\wfsubst*[\GG]{\Gd}{\ctxbdim{\GD}}}
  {\wfsubst*[\ctxres*{\GG}{\bdimsubst{\bdimvar}{\Gd}}]{\substbtranspose{\Gd}}{\GD}}
  \and
  \inferrule[subst-face]
  {\Ge \in \{0,1\}}
  {\wfsubst*[\GG]{\substbface{\Ge}}{\ctxbdim{\GG}}}
  \and
  \inferrule[subst-degen]
  { }
  {\wfsubst*[\ctxbdim{\GG}]{\substprojbdim}{\GG}}
  \and
  \inferrule[subst-exchange]
  {\wfctx*{\GG}}
  {\wfsubst*[\ctxbdim{\ctxbdim{\GG}}]{\substbex}{\ctxbdim{\ctxbdim{\GG}}}}
  \and
\end{mathparpagebreakable}

We introduce the following abbreviations for the functorial actions of the three forms of context extension.

\begin{mathparpagebreakable}
  \inferrule
  {\wfsubst*[\GG]{\Gd}{\GD} \\
    \wftype*[\ctxmod{\GD}{\Gm}]{A}}
  {\wfsubst*[\ctxsnoc{\GG}{\tysubst{A}{\Gd}}]{\funcsnoc{\Gd} \eqdef \substsnoc*{(\substconc{\Gd}{\substproj})}{\tmvar}}{\ctxsnoc{\GD}{A}}}
  \and
  \inferrule
  {\wfsubst*[\GG]{\Gd}{\GD}}
  {\wfsubst*[\ctxbdim{\GG}]{\funcbdim{\Gd} \eqdef \substbdim*{(\substconc{\Gd}{\substbtranspose{\substid}})}{\bdimvar}}{\ctxbdim{\GD}}}
  \and
  \inferrule
  {\wfsubst*[\GG]{\Gd}{\GD} \\
    \wfbdim*[\GD]{\bmr}}
  {\wfsubst*[\ctxres*{\GG}{\bdimsubst{\bmr}{\Gd}}]{\funcres{\Gd}{\bmr} \eqdef \substbtranspose{(\substconc{\substbdim*{\substid}{\bmr}}{\Gd})}}{\ctxres*{\GD}{\bmr}}}
\end{mathparpagebreakable}

\subsection{Substitution equality}~%

\begin{mathparpagebreakable}
  \inferrule[subst-nil-eta]
  {\wfsubst*[\GG]{\Gd}{\ctxnil}}
  {\eqsubst*[\GG]{\Gd}{\substnil*}{\ctxnil}}
  \and
  \inferrule[subst-id-conc]
  { }
  {\eqsubst*[\GG]{\substconc{\substid}{\Gd}}{\Gd}{\GD}}
  \and
  \inferrule[subst-conc-id]
  { }
  {\eqsubst*[\GG]{\substconc{\Gd}{\substid}}{\Gd}{\GD}}
  \and
  \inferrule[subst-conc-conc]
  {\wfsubst*[\GD_1]{\Gd_0}{\GD_0} \\
    \wfsubst*[\GD_2]{\Gd_1}{\GD_1} \\
    \wfsubst*[\GG]{\Gd_2}{\GD_2}}
  {\eqsubst*[\GG]{\substconc{(\substconc{\Gd_0}{\Gd_1})}{\Gd_2}}{\substconc{\Gd_0}{(\substconc{\Gd_1}{\Gd_2})}}{\GD_0}}
  \and
  \inferrule[subst-proj-term]
  {\wfsubst*[\GG]{\Gd}{\GD} \\
    \wftype*[\GD]{A} \\
    \wftm*[\GG]{M}{A}}
  {\eqsubst*[\GG]{\substconc{\substproj}{(\substsnoc*{\Gd}{M})}}{\Gd}{\GD}}
  \and
  \inferrule[subst-term-eta]
  {\wftype*[\GD]{A} \\
    \wfsubst*[\GG]{\Gd}{\ctxsnoc{\GD}{A}}}
  {\eqsubst*[\GG]{\Gd}{\substsnoc*{(\substconc{\substproj}{\Gd})}{\tmsubst{\tmvar}{\Gd}}}{\ctxsnoc{\GD}{A}}}
  \and
  \inferrule[subst-eq-$\BFI$]
  {\wfctx*{\GD} \\
    \wfsubst*[\GG]{\Gd}{\ctxbdim{\GD}}}
  {\eqsubst*[\GG]{\Gd}{\substbdim*{\substbtranspose{\Gd}}{\bdimsubst{\bdimvar}{\Gd}}}{\ctxbdim{\GD}}}
  \and
  \inferrule[subst-eq-restrict]
  {\wfbdim*[\GG]{\bmr} \\
    \wfsubst*[\ctxres*{\GG}{\bmr}]{\Gd}{\GD}}
  {\eqsubst*[\ctxres*{\GG}{\bmr}]{\Gd}{\substbtranspose{(\substbdim*{\Gd}{\bmr})}}{\GD}}
  \and
  \inferrule[subst-$\BFI$-natural]
  {\wfsubst*[\GG]{\Gd}{\GD} \\
    \wfbdim*[\GX]{\bmr} \\
    \wfsubst*[\ctxres*{\GX}{\bmr}]{\Gg}{\GG}}
  {\eqsubst*[\GX]{\substbdim*{(\substconc{\Gd}{\Gg})}{\bmr}}{\substconc{\funcbdim{\Gd}}{(\substbdim*{\Gg}{\bmr})}}{\ctxbdim{\GD}}}
  \and
  \inferrule[subst-restrict-natural]
  {\wfsubst*[\GG]{\Gd}{\ctxbdim{\GD}} \\
    \wfsubst*[\GX]{\Gg}{\GG}}
  {\eqsubst*[\ctxres*{\GX}{\bdimsubst{\bdimvar}{\substconc{\Gd}{\Gg}}}]{\substbtranspose{(\substconc{\Gd}{\Gg})}}{\substconc{\substbtranspose{\Gd}}{(\funcres{\Gg}{\bdimsubst{\bdimvar}{\Gd}})}}{\GD}}
  \and
  \inferrule[subst-face-natural]
  {\Ge \in \{0,1\} \\
    \wfsubst*[\GG]{\Gd}{\GD}}
  {\eqsubst*[\GG]{\substconc{\funcbdim{\Gd}}{\substbface\Ge}}{\substconc{\substbface\Ge}{\Gd}}{\ctxbdim{\GD}}}
  \and
  \inferrule[subst-degen-natural]
  {\wfsubst*[\GG]{\Gd}{\GD}}
  {\eqsubst*[\ctxbdim{\GG}]{\substconc{\Gd}{\substprojbdim}}{\substconc{\substprojbdim}{\funcbdim{\Gd}}}{\GD}}
  \and
  \inferrule[subst-exchange-natural]
  {\wfsubst*[\GG]{\Gd}{\GD}}
  {\eqsubst*[\ctxbdim{\ctxbdim{\GG}}]{\substconc{\funcbdim[2]{\Gd}}{\substbex}}{\substconc{\substbex}{\funcbdim[2]{\Gd}}}{\ctxbdim{\ctxbdim{\GD}}}}
  \and
  \inferrule[subst-proj-face]
  {\Ge \in \{0,1\}}
  {\eqsubst*[\GG]{\substconc{\substprojbdim}{\substbface\Ge}}{\substid}{\GG}}
  \and
  \inferrule[subst-proj-exchange]
  { }
  {\eqsubst*[\ctxbdim{\ctxbdim{\GG}}]{\substconc{\substprojbdim}{\substbex}}{\funcbdim{\substprojbdim}}{\ctxbdim{\GG}}}
  \and
  \inferrule[subst-exchange-exchange]
  { }
  {\eqsubst*[\ctxbdim{\ctxbdim{\GG}}]{\substconc{\substbex}{\substbex}}{\substid}{\ctxbdim{\ctxbdim{\GG}}}}
\end{mathparpagebreakable}

\subsection{Types}~%

\begin{mathparpagebreakable}
  \inferrule[ty-subst]
  {\wftype*[\GD]{A} \\
    \wfsubst*[\GG]{\Gd}{\GD}}
  {\wftype*[\GG]{\tysubst{A}{\Gd}}}
\end{mathparpagebreakable}

\subsection{Type equality}~%

\begin{mathparpagebreakable}
  \inferrule[ty-subst-id]
  { }
  {\eqtype*[\GG]{\tysubst{A}{\substid}}{A}}
  \and
  \inferrule[ty-subst-conc]
  {\wftype*[\GD_0]{A} \\
    \wfsubst*[\GD_1]{\Gd_0}{\GD_0} \\
    \wfsubst*[\GG]{\Gd_1}{\GD_1}}
  {\eqtype*[\GG]{\tysubst{A}{\substconc{\Gd_0}{\Gd_1}}}{\tysubst{\tysubst{A}{\Gd_0}}{\Gd_1}}}
\end{mathparpagebreakable}

\subsection{Terms}~%

\begin{mathparpagebreakable}
  \inferrule[tm-var]
  {\wftype*[\GG]{A}}
  {\wftm*[\ctxsnoc{\GG}{A}]{\tmvar}{\tysubst{A}{\substproj}}}
  \and
  \inferrule[tm-subst]
  {\wfsubst*[\GG]{\Gd}{\GD} \\
    \wftm*[\GD]{M}{A}}
  {\wftm*[\GG]{\tmsubst{M}{\Gd}}{\tysubst{A}{\Gd}}}
\end{mathparpagebreakable}

\subsection{Term equality}~%

\begin{mathparpagebreakable}
  \inferrule[tm-subst-id]
  {\wftm*[\GG]{M}{A}}
  {\eqtm*[\GG]{\tmsubst{M}{\substid}}{M}{A}}
  \and
  \inferrule[tm-subst-conc]
  {\wftm*[\GD_0]{M}{A} \\
    \wfsubst*[\GD_1]{\Gd_0}{\GD_0} \\
    \wfsubst*[\GG]{\Gd_1}{\GD_1}}
  {\eqtm*[\GG]{\tmsubst{M}{\substconc{\Gd_0}{\Gd_1}}}{\tmsubst{\tmsubst{M}{\Gd_0}}{\Gd_1}}{\tysubst{\tysubst{A}{\Gd_0}}{\Gd_1}}}
  \and
  \inferrule[tm-subst-term]
  {\wfsubst*[\GG]{\Gd}{\GD} \\
    \wftype*[\GD]{A} \\
    \wftm*[\GG]{M}{\tysubst{A}{\Gd}}}
  {\eqtm*[\GG]{\tmsubst{\tmvar}{\substsnoc*{\Gd}{M}}}{M}{\tysubst{A}{\Gd}}}
\end{mathparpagebreakable}

\subsection{Bridge types}~%

\begin{mathparpagebreakable}
  \inferrule[ty-bridge]
  {\wftype*[\ctxbdim{\GG}]{A} \\
    \wftm*[\GG]{M_0}{\tysubst{A}{\substbface0}} \\
    \wftm*[\GG]{M_1}{\tysubst{A}{\substbface1}}}
  {\wftype*[\GG]{\tybridge{A}{M_0}{M_1}}}
  \and
  \inferrule[tm-blam]
  {\wftype*[\ctxbdim{\GG}]{A} \\
    \wftm*[\ctxbdim{\GG}]{M}{A}}
  {\wftm*[\GG]{\tmblam{M}}{\tybridge{A}{\tmsubst{M}{\substbface0}}{\tmsubst{M}{\substbface1}}}}
  \and
  \inferrule[tm-bapp]
  {\wfbdim*[\GG]{\bmr} \\
    \wftype*[\ctxbdim{\ctxres*{\GG}{\bmr}}]{A} \\
    \wftm*[\ctxres*{\GG}{\bmr}]{M_0}{\tysubst{A}{\substbface0}} \\
    \wftm*[\ctxres*{\GG}{\bmr}]{M_1}{\tysubst{A}{\substbface1}} \\
    \wftm*[\ctxres*{\GG}{\bmr}]{P}{\tybridge{A}{M_0}{M_1}}}
  {\wftm*[\GG]{\tmbapp{P}{\bmr}}{\tysubst{A}{\substbdim*{\substid}{\bmr}}}}
  \and
  \inferrule[tm-bapp-boundary]
  {\Ge \in \{0, 1\} \\
    \wftype*[\ctxbdim{\GG}]{A} \\
    \wftm*[\GG]{M_0}{\tysubst{A}{\substbface0}} \\
    \wftm*[\GG]{M_1}{\tysubst{A}{\substbface1}} \\
    \wftm*[\GG]{P}{\tybridge{A}{M_0}{M_1}}}
  {\eqtm*[\GG]{\tmbapp{\tmsubst{P}{\substbtranspose{\substbface\Ge}}}{\bdimsubst{\bdimvar}{\substbface\Ge}}}{M_\Ge}{\tysubst{A}{\substbface\Ge}}}
  \and
  \inferrule[tm-blam-beta]
  {\wfbdim*[\GG]{\bmr} \\
    \wftype*[\ctxbdim{\ctxres*{\GG}{\bmr}}]{A} \\
    \wftm*[\ctxbdim{\ctxres*{\GG}{\bmr}}]{M}{A}}
  {\eqtm*[\GG]{\tmbapp{\tmlam{M}}{\bmr}}{\tmsubst{M}{\substbdim*{\substid}{\bmr}}}{\tysubst{A}{\substbdim*{\substid}{\bmr}}}}
  \and
  \inferrule[tm-blam-eta]
  {\wftype*[\ctxbdim{\GG}]{A} \\
    \wftm*[\GG]{M_0}{\tysubst{A}{\substbface0}} \\
    \wftm*[\GG]{M_1}{\tysubst{A}{\substbface1}} \\
    \wftm*[\GG]{P}{\tybridge{A}{M_0}{M_1}}}
  {\eqtm*[\GG]{P}{\tmblam{\tmbapp{\tmsubst{P}{\substbtranspose{\substid}}}{\bdimvar}}}{\tybridge{A}{M_0}{M_1}}}
\end{mathparpagebreakable}

\subsection{Gel types}~%

\begin{mathparpagebreakable}
  \inferrule[ty-gel]
  {\wfbdim*[\GG]{\bmr} \\
    \wftype*[\ctxres*{\GG}{\bmr}]{A_0} \\
    \wftype*[\ctxres*{\GG}{\bmr}]{A_1} \\
    \wftype*[\ctxsnoc{\ctxsnoc{\ctxres*{\GG}{\bmr}}{A_0}}{\tysubst{A_1}{\substproj}}]{R}}
  {\wftype*[\GG]{\tygel{\bmr}{A_0}{A_1}{R}}}
  \and
  \inferrule[ty-gel-boundary]
  {\Ge \in \{0,1\} \\
    \wftype*[\GG]{A_0} \\
    \wftype*[\GG]{A_1} \\
    \wftype*[\ctxsnoc{\ctxsnoc{\GG}{A_0}}{\tysubst{A_1}{\substproj}}]{R}}
  {\eqtype*[\GG]{\tygel{\bm\Ge}{\tysubst{A_0}{\substbtranspose{\substbface\Ge}}}{\tysubst{A_1}{\substbtranspose{\substbface\Ge}}}{\tysubst{R}{\funcsnoc[2]{\substbtranspose{\substbface\Ge}}}}}{A_\Ge}}
  \and
  \inferrule[tm-gel]
  {\wfbdim*[\GG]{\bmr} \\
    \wftm*[\ctxres*{\GG}{\bmr}]{M_0}{A_0} \\
    \wftm*[\ctxres*{\GG}{\bmr}]{M_1}{A_1} \\
    \wftype*[\ctxsnoc{\ctxsnoc{\ctxres*{\GG}{\bmr}}{A_0}}{\tysubst{A_1}{\substproj}}]{R} \\
    \wftm*[\ctxres*{\GG}{\bmr}]{P}{\tysubst{R}{\substsnoc*{\substsnoc*{\substid}{M_0}}{M_1}}}}
  {\wftm*[\GG]{\tmgel{\bmr}{M_0}{M_1}{P}}{\tygel{\bmr}{A_0}{A_1}{R}}}
  \and
  \inferrule[tm-gel-boundary]
  {\Ge \in \{0,1\} \\
    \wftm*[\GG]{M_0}{A_0} \\
    \wftm*[\GG]{M_1}{A_1} \\
    \wftype*[\ctxsnoc{\ctxsnoc{\GG}{A_0}}{\tysubst{A_1}{\substproj}}]{R} \\
    \wftm*[\GG]{P}{\tysubst{R}{\substsnoc*{\substsnoc*{\substid}{M_0}}{M_1}}}}
  {\eqtm*[\GG]{\tmgel{\bm\Ge}{\tmsubst{M_0}{\substbtranspose{\substbface\Ge}}}{\tmsubst{M_1}{\substbtranspose{\substbface\Ge}}}{\tmsubst{P}{\substbtranspose{\substbface\Ge}}}}{M_\Ge}{A_\Ge}}
  \and
  \inferrule[tm-ungel]
  {\wftype*[\GG]{A_0} \\
    \wftype*[\GG]{A_1} \\
    \wftype*[\ctxsnoc{\ctxsnoc{\GG}{A_0}}{\tysubst{A_1}{\substproj}}]{R} \\
    \wftm*[\ctxbdim{\GG}]{Q}{\tygel{\bdimvar}{\tysubst{A_0}{\substbtranspose{\substid}}}{\tysubst{A_1}{\substbtranspose{\substid}}}{\tysubst{R}{{\funcsnoc[2]{\substbtranspose{\substid}}}}}}}
  {\wftm*[\GG]{\tmungel{Q}}{\tysubst{R}{\substsnoc*{\substsnoc*{\substid}{\tysubst{Q}{\substbface0}}}{\tysubst{Q}{\substbface1}}}}}
  \and
  \inferrule[tm-gel-beta]
  {\wftm*[\GG]{M_0}{A_0} \\
    \wftm*[\GG]{M_1}{A_1} \\
    \wftype*[\ctxsnoc{\ctxsnoc{\GG}{A_0}}{\tysubst{A_1}{\substproj}}]{R} \\
    \wftm*[\GG]{P}{\tysubst{R}{\substsnoc*{\substsnoc*{\substid}{M_0}}{M_1}}}}
  {\eqtm*[\GG]{\tmungel{\tmgel{\bdimvar}{\tmsubst{M_0}{\substbtranspose{\substid}}}{\tmsubst{M_1}{\substbtranspose{\substid}}}{\tmsubst{P}{\substbtranspose{\substid}}}}}{P}{\tysubst{R}{\substsnoc*{\substsnoc*{\substid}{M_0}}{M_1}}}}
  \and
  \inferrule[tm-gel-eta]
  {\wfbdim*[\GG]{\bmr} \\
    \wftype*[\ctxres*{\GG}{\bmr}]{A_0} \\
    \wftype*[\ctxres*{\GG}{\bmr}]{A_1} \\
    \wftype*[\ctxsnoc{\ctxsnoc{\ctxres*{\GG}{\bmr}}{A_0}}{\tysubst{A_1}{\substproj}}]{R} \\
    \wftm*[\ctxbdim{\ctxres*{\GG}{\bmr}}]{Q}{\tygel{\bdimvar}{\tysubst{A_0}{\substbtranspose{\substid}}}{\tysubst{A_1}{\substbtranspose{\substid}}}{\tysubst{R}{\funcsnoc[2]{\substbtranspose{\substid}}}}}}
  {\eqtm*[\GG]{\tmsubst{Q}{\substbdim*{\substid}{\bmr}}}{\tmgel{\bmr}{\tmsubst{Q}{\substbface0}}{\tmsubst{Q}{\substbface1}}{\tmungel{Q}}}{\tygel{\bmr}{A_0}{A_1}{R}}}
\end{mathparpagebreakable}

\subsection{Extent}~%

\begin{mathparpagebreakable}
  \inferrule[tm-extent]
  {\wfbdim*[\GG]{\bmr} \\
    \wftype*[\ctxbdim{\ctxres*{\GG}{\bmr}}]{A} \\
    \wftype*[\ctxsnoc{\ctxbdim{\ctxres*{\GG}{\bmr}}}{A}]{B} \\
    \wftm*[\GG]{M}{\tysubst{A}{\substbdim*{\substid}{\bmr}}} \\
    \wftm*[\ctxsnoc{\ctxres*{\GG}{\bmr}}{\tysubst{A}{\substbface0}}]{N_0}{\tysubst{B}{\funcsnoc{\substbface0}}} \\
    \wftm*[\ctxsnoc{\ctxres*{\GG}{\bmr}}{\tysubst{A}{\substbface1}}]{N_1}{\tysubst{B}{\funcsnoc{\substbface1}}} \\
    \wftm*[\ctxsnoc{\ctxsnoc{\ctxsnoc{\ctxres*{\GG}{\bmr}}{\tysubst{A}{\substbface0}}}{\tysubst{A}{\substconc{\substbface1}{\substproj}}}}{\tybridge{\tysubst{A}{\substproj[2]}}{\tmsubst{\tmvar}{\substproj}}{\tmvar}}]{N}{\tybridge{\tysubst{B}{\substsnoc*{\substbdim*{(\substconc{\substproj[3]}{\substbtranspose{\substid}})}{\bdimvar}}{\tmbapp{\tmsubst{\tmvar}{\substbtranspose{\substid}}}{\bdimvar}}}}{\tmsubst{N_0}{\substproj[2]}}{\tmsubst{N_1}{\substconc{\funcsnoc{\substproj}}{\substproj}}}}}
  {\wftm*[\GG]{\tmextent{\bmr}{M}{N_0}{N_1}{N}}{\tysubst{B}{\substsnoc*{\substbdim*{\substid}{\bmr}}{M}}}}
  \and
  \inferrule[tm-extent-boundary]
  {\Ge \in \{0,1\} \\
    \wftype*[\ctxbdim{\GG}]{A} \\
    \wftype*[\ctxsnoc{\ctxbdim{\GG}}{A}]{B} \\
    \wftm*[\GG]{M}{\tysubst{A}{\substbface\Ge}} \\
    \wftm*[\ctxsnoc{\GG}{\tysubst{A}{\substbface0}}]{N_0}{\tysubst{B}{\funcsnoc{\substbface0}}} \\
    \wftm*[\ctxsnoc{\GG}{\tysubst{A}{\substbface1}}]{N_1}{\tysubst{B}{\funcsnoc{\substbface1}}} \\
    \wftm*[\ctxsnoc{\ctxsnoc{\ctxsnoc{\GG}{\tysubst{A}{\substbface0}}}{\tysubst{A}{\substconc{\substbface1}{\substproj}}}}{\tybridge{\tysubst{A}{\substproj[2]}}{\tmsubst{\tmvar}{\substproj}}{\tmvar}}]{N}{\tybridge{\tysubst{B}{\substsnoc*{\substbdim*{(\substconc{\substproj[3]}{\substbtranspose{\substid}})}{\bdimvar}}{\tmbapp{\tmsubst{\tmvar}{\substbtranspose{\substid}}}{\bdimvar}}}}{\tmsubst{N_0}{\substproj[2]}}{\tmsubst{N_1}{\substconc{\funcsnoc{\substproj}}{\substproj}}}}}
  {\eqtm*[\GG]{\tmextent{\bdimsubst{\bdimvar}{\substbface\Ge}}{M}{\tmsubst{N_0}{\funcsnoc{\substbtranspose{\substbface\Ge}}}}{\tmsubst{N_1}{\funcsnoc{\substbtranspose{\substid}}}}{\tmsubst{N}{\funcsnoc[3]{\substbtranspose{\substid}}}}}{\tmsubst{N_\Ge}{\substsnoc*{\substid}{M}}}{\tysubst{B}{\substsnoc*{\substbface\Ge}{M}}}}
  \and
  \inferrule[tm-extent-beta]
  {\wfbdim*[\GG]{\bmr} \\
    \wftype*[\ctxbdim{\ctxres*{\GG}{\bmr}}]{A} \\
    \wftype*[\ctxsnoc{\ctxbdim{\ctxres*{\GG}{\bmr}}}{A}]{B} \\
    \wftm*[\ctxbdim{\ctxres*{\GG}{\bmr}}]{M}{A} \\
    \wftm*[\ctxsnoc{\ctxres*{\GG}{\bmr}}{\tysubst{A}{\substbface0}}]{N_0}{\tysubst{B}{\funcsnoc{\substbface0}}} \\
    \wftm*[\ctxsnoc{\ctxres*{\GG}{\bmr}}{\tysubst{A}{\substbface1}}]{N_1}{\tysubst{B}{\funcsnoc{\substbface1}}} \\
    \wftm*[\ctxsnoc{\ctxsnoc{\ctxsnoc{\ctxres*{\GG}{\bmr}}{\tysubst{A}{\substbface0}}}{\tysubst{A}{\substconc{\substbface1}{\substproj}}}}{\tybridge{\tysubst{A}{\substproj[2]}}{\tmsubst{\tmvar}{\substproj}}{\tmvar}}]{N}{\tybridge{\tysubst{B}{\substsnoc*{\substbdim*{(\substconc{\substproj[3]}{\substbtranspose{\substid}})}{\bdimvar}}{\tmbapp{\tmsubst{\tmvar}{\substbtranspose{\substid}}}{\bdimvar}}}}{\tmsubst{N_0}{\substproj[2]}}{\tmsubst{N_1}{\substconc{\funcsnoc{\substproj}}{\substproj}}}}}
  {\eqtm*[\GG]{\tmextent{\bmr}{\tmsubst{M}{\substsnoc*{\substid}{\bmr}}}{N_0}{N_1}{N}}{\tmbapp{\tmsubst{N}{\substsnoc*{\substsnoc*{\substsnoc*{\substid}{\tmsubst{M}{\substbface0}}}{\tmsubst{M}{\substbface1}}}{\tmblam{M}}}}{\bmr}}{\tysubst{B}{\substsnoc*{\substbdim*{\substid}{\bmr}}{M}}}}
\end{mathparpagebreakable}

\bibliographystyle{alpha}
\bibliography{main}

\end{document}